\documentclass[graybox]{svmult}

\usepackage{mathptmx}       
\usepackage{makeidx}         
\usepackage{graphicx}        
\usepackage{multicol}        
\usepackage[bottom]{footmisc}
 
\usepackage{amsmath,amssymb}
\usepackage[mathscr]{eucal}
\usepackage{mathrsfs}
\DeclareMathAlphabet{\mathcal}{OMS}{cmsy}{m}{n}

\makeatletter
\g@addto@macro\bfseries{\boldmath}
\makeatother

\usepackage{bm}
\DeclareMathAlphabet{\mathbfsf}{\encodingdefault}{\sfdefault}{bx}{n}

\DeclareBoldMathCommand\Fb{F}
\DeclareBoldMathCommand\Mb{M}
\DeclareBoldMathCommand\Nb{N}
\DeclareBoldMathCommand\Pb{P}
\DeclareBoldMathCommand\Ob{O}
\DeclareBoldMathCommand\rb{R}
\DeclareBoldMathCommand\ab{a}
\DeclareBoldMathCommand\bb{b}
\DeclareBoldMathCommand\cb{c}
\DeclareBoldMathCommand\eb{e}
\DeclareBoldMathCommand\ib{i}
\DeclareBoldMathCommand\jb{j}
\DeclareBoldMathCommand\kb{k}
\DeclareBoldMathCommand\pb{p}
\DeclareBoldMathCommand\rb{r}
\DeclareBoldMathCommand\ub{u}
\DeclareBoldMathCommand\vb{v}
\DeclareBoldMathCommand\xb{x}

\newcommand{\dvol}{{\rm dvol}}

\newcommand{\id}{\text{id}}

\def\spa{\hskip -3pt}
\newsymbol\rest 1316     

\newcommand{\Ac}{{\mathcal{A}}}

\newcommand{\Sol}{\mathbfsf{ Sol}}

\newcommand{\II}{\leavevmode\hbox{\rm{\small1\kern-3.8pt\normalsize1}}}

\newcommand{\supp}{\textrm{supp}\,}

\newcommand{\Bf}{{\mathscr B}}

\newcommand{\Lf}{{\mathscr L}}

\newcommand{\Mc}{{\mathcal M}}

\newcommand{\ogth}{{\mathfrak o}}
\newcommand{\tgth}{{\mathfrak t}}

\newcommand{\nb}{{\boldsymbol{n}}}

\makeindex             

\begin{document}

\title*{Algebraic QFT in Curved Spacetime and quasifree Hadamard states: an introduction}
\author{Igor Khavkine and Valter Moretti}
\institute{Igor Khavkine \at Department of Mathematics, University of Trento, via Sommarive 14, I-38123 Povo (Trento), Italy,
\email{igor.khavkine@unitn.it}\\
$\null$\\
Valter Moretti  \at Department of Mathematics, University of Trento, via Sommarive 14, I-38123 Povo (Trento), Italy,
\email{moretti@science.unitn.it}}

\maketitle

\abstract*{Within this chapter (published as \cite{chapt:KM}) we
introduce the overall idea of the algebraic formalism of QFT on a fixed
globally hyperbolic spacetime in the framework of unital $*$-algebras.
We point out some general features of CCR algebras, such as simplicity
and the construction of symmetry-induced homomorphisms. For simplicity,
we deal only with a real scalar quantum field. We discuss some known
general results in curved spacetime like the existence of quasifree
states enjoying symmetries induced from the background, pointing out the
relevant original references. We introduce, in particular, the notion of
a Hadamard quasifree algebraic quantum state, both in the geometric and
microlocal formulation, and the associated notion of Wick polynomials.}

\abstract{Within this chapter (published as \cite{chapt:KM}) we
introduce the overall idea of the algebraic formalism of QFT on a fixed
globally hyperbolic spacetime in the framework of unital $*$-algebras.
We point out some general features of CCR algebras, such as simplicity
and the construction of symmetry-induced homomorphisms. For simplicity,
we deal only with a real scalar quantum field. We discuss some known
general results in curved spacetime like the existence of quasifree
states enjoying symmetries induced from the background, pointing out the
relevant original references. We introduce, in particular, the notion of
a Hadamard quasifree algebraic quantum state, both in the geometric and
microlocal formulation, and the associated notion of Wick polynomials. }

\section{Algebraic formalism}
With this preliminary section we introduce some basic definitions and result about algebraic formulation of quantum theory reviewing some basic definitions and results  about the algebraic machinery.
Most literature devoted to the algebraic approach to QFT is written using $C^*$-algebras, in particular Weyl $C^*$-algebras, when dealing with free fields, nevertheless the ``practical'' literature  mostly uses {\em unbounded} field operators which are encapsulated in the notion of $*$-algebra instead of  $C^*$-algebra, whose additional feature is a multiplicatively compatible norm. Actually, at the level of free theories and quasifree (Gaussian) states the two approaches are technically equivalent.  Since we think more plausible  that the non-expert reader acquainted with QFT in Minkowski spacetime is,  perhaps unconsciously,  more familiar with $*$-algebras than $C^*$-algebras,  in the rest of the chapter we adopt the $*$-algebra framework. 

\begin{definition}[Algebras] An {\bf algebra} ${\cal A}$ is a complex vector space which is equipped with an associative product $${\cal A} \times {\cal A} \ni (a,b) \mapsto ab \in {\cal A}$$ which  is distributive with respect to the vector sum operation and satisfies $$\alpha (ab) = (\alpha a)b = a(\alpha b) \quad \mbox{if $\alpha \in \mathbb C$ and $a,b \in {\cal A}$}\:.$$ 
${\cal A}$ is  a {\bf $*$-algebra} if it admits an {\bf involution}, namely an anti-linear map, ${\cal A} \ni a \mapsto a^*$, which is involutive, that is $(a^*)^* =a$, and such that $(ab)^* = b^* a^*$, for any $a,b \in {\cal A}$. \\ ${\cal A}$ is {\bf unital} if it contains  a multiplicative {\bf unit} $\II\in {\cal A}$, that is $\II a = a \II = a$ for all
$a \in {\cal A}$.

A set $G \subset {\cal A}$ is said to {\bf generate} the algebra ${\cal A}$, and the elements of $G$ are said {\bf generators} of $\Ac$, if each element of  $\Ac$ is a finite complex linear combination of products (with arbitrary number of factors) of elements of $G$.

 The {\bf center}, ${\cal Z}_{\Ac}$, of the algebra $\Ac$ is the set of elements $z\in {\cal A}$ commuting with all elements of ${\cal A}$.
\end{definition}

Regarding morphisms of algebras we shall adopt the following standard definitions 

\begin{definition} [Algebra morphisms]
Consider a map $\beta : {\cal A}_1 \to {\cal A}_2$, where ${\cal A}_i$ are algebras.

{\bf (a)}  $\beta$  is an  {\bf algebra homomorphism} if it is a complex  linear map, preserves the product and, if the algebras are unital, preserves the unit elements. 

{\bf (b)} $\beta$  is a  {\bf $*$-algebra homomorphism}
If ${\cal A}_i$ are $*$-algebras,  $\beta$ is a  {\bf algebra homomorphism}  and  preserves the involution.

{\bf (c)}  $\beta$  is an  {\bf algebra isomorphism} or  a   {\bf $*$-algebra isomorphism} if it is an   {\bf algebra homomorphism} or, respectively, a   {\bf $*$-algebra homomorphism}   and it is bijective.

{\bf (d)}  $\beta$  is an  {\bf algebra automorphism} or  a   {\bf $*$-algebra automorphism}
if it is 
a   {\bf algebra isomorphism} or, respectively, a   {\bf $*$-algebra isomorphism}   and $\Ac_1=\Ac_2$.

\noindent Corresponding {\bf anti-linear morphisms} are defined analogously replacing the linearity condition with anti-linearity.
\end{definition}

\begin{remark}$\null$		

{\bf (1)} The unit  $\II$, if exists, turns out to be  unique.  In $*$-algebras it  satisfies $\II = \II^*$.

{\bf (2)} Although we shall not deal with $C^*$-algebras, we recall the reader that a $*$-algebra 
${\cal A}$ is a {\bf $C^*$-algebra} if it is a Banach space with respect to a norm $||\:\:||$ which satisfies $||ab||\leq ||a||\:||b||$ and $||a^*a||=||a||^2$ if $a,b \in {\cal A}$. It turns out that  $||a^*||=||a||$ and, if the $C^*$-algebra is unital, $||\II||=1$.
A unital $*$-algebra admits at most one norm making it a $C^*$-algebra.
\end{remark}

\begin{definition}[Two-sided ideals]
A {\bf two-sided ideal} of an algebra $\Ac$ is a linear complex subspace ${\cal I} \subset \Ac$ such that $ab \in {\cal I}$ and  $ba \in {\cal I}$ if $a \in \Ac$ and $b \in {\cal I}$.

In a $*$-algebra,  a two-sided ideal ${\cal I}$ is said to be a  {\bf two-sided $*$-ideal} if it is also closed with respect to the involution: $a^* \in {\cal I}$ if $a \in {\cal I}$.

An algebra $\Ac$ is {\bf simple} if it does not admit  two-sided ideals different 
form $\{0\}$ and $\Ac$ itself.
\end{definition}

\begin{remark}
It should be evident that the intersection of a class of two-sided ideals (two-sided $*$-ideals) is a two-sided ideal (resp. two-sided $*$-ideal). 
\end{remark}

\subsection{The general algebraic approach to quantum theories} In the algebraic formulation of a quantum theory \cite{haag}, observables are viewed as abstract self-adjoint objects instead of operators in a given Hilbert space. These observable generate a {\em $*$-algebra} or a {\em $C^*$-algebra} depending on the context. The algebra also includes a formal identity $\II$ and {\em complex} linear combinations of observables which, consequently cannot be interpreted as observables. Nevertheless  the use of complex algebras is mathematically convenient.  The justification of a linear structure for the set of the observables is quite easy, the presence of an associative product is instead much more difficult to justify \cite{strocchi}. However,  {\em a posteriori}, this  approach reveals to be powerful  and it is particularly convenient when the theory encompasses many unitarily inequivalent representation of the algebra of observables, as it happens in quantum field theory.

\subsection{Defining $*$-algebras by generators and relations}\label{sec:genrel}
In the algebraic approach, the $*$-algebra of observables cannot be defined simply as some concrete set of (possibly unbounded) operators on some Hilbert space. Instead, the $*$-algebra must be defined abstractly, using some more basic objects. Below we recall an elementary algebraic construction that will be of use in Section~\ref{sec:kg} in defining the CCR algebra of a scalar field. 

We will construct a $*$-algebra from a presentation by \emph{generators} and \emph{relations}.
As we shall see in the Section \ref{sec:kg},  the CCR algebra is {\em generated} by abstract objects, the {\em smeared fields}, $\phi(f)$ and the unit $\II$. In other words, the elements of the algebra are finite linear combinations of products of these objects. However there also are  {\em relations}  among these objects, e.g. $[\phi(f),\phi(g)] = i E(f,g) \II$. We therefore need an abstract  procedure to define this sort of algebras, starting form generators and imposing relations.   We make each of these concepts precise in a general context.

Let us start with the notion of algebra, ${\cal A}_G$, {\em generated} by a set of generators $G$.
 Intuitively, the algebra ${\cal A}_G$ is the {\em smallest} algebra that contains the elements of the generator set $G$ (yet without any algebraic relations between these generators). The following is an example of a definition by a \emph{universal property}~\cite[\textsection I.11]{lang}. 

\begin{definition}[Free algebra]
Given a set $G$ of {\bf generators} (not necessarily finite or even countable), an algebra ${\cal A}_G$ is said to be {\bf freely generated by $G$} (or {\bf free on $G$}) if there is a map $\gamma\colon G\to {\cal A}_G$ such that, for any other algebra ${\cal B}$ and map $\beta\colon G \to {\cal B}$, there exists a \emph{unique} algebra homomorphism $b\colon {\cal A}_G \to {\cal B}$ such that $\beta = b\circ \gamma$. We use the same terminology for $*$- and unital algebras.
\end{definition}

\begin{remark} $\null$

{\bf (1)} Any two algebras freely generated by $G$, given by say $\gamma\colon G \to {\cal A}_G$ and $\gamma'\colon G \to {\cal A}'_G$, are {\em naturally isomorphic}. In this sense ${\cal A}_G$ is uniquely determined by $G$.
 By definition, there exist unique homomorphisms $a\colon {\cal A}'_G \to {\cal A}_G$ and $a'\colon {\cal A}_G \to {\cal A}'_G$ such that $\gamma = a \circ \gamma'$ and $\gamma' = a' \circ \gamma$. Their compositions satisfy the same kind of identity as $b$ in the above definition, namely $\gamma = \id \circ \gamma = (a\circ a') \circ \gamma$ and $\gamma' = \id \circ \gamma' = (a'\circ a) \circ \gamma'$, where we use $\id$ to denote the identity homomorphism on any algebra. Invoking once again uniqueness shows that $a\circ a' = \id = a'\circ a$ and hence that ${\cal A}_G$ and ${\cal A}'_G$ are naturally isomorphic. So, any representative of this isomorphism class could be called \emph{the} algebra freely generated by $G$.

{\bf (2)} To make the above definition useful we must prove that a pair  $({\cal A}_G, \gamma)$ exists for every  set $G$. 
 Consider the complex vector space spanned by the basis $\{e_S\}$, where $S$ runs through all finite ordered sequences of the elements of $G$, say $S = (g_1,\ldots g_k)$, with $k>0$. Define multiplication on be basis elements by concatenation, $e_S e_T = e_{ST}$, where $(g_1,\ldots, g_k) (g'_1,\ldots, g'_l) = (g_1, \ldots, g_k, g'_1, \ldots, g'_l)$ and extend it to the whole vector space by linearity. It is straight forward to see that we have defined an algebra that satisfies the property of being freely generated by $G$. In the case of unital $*$-algebras, we use the same construction, except that the basis is augmented by the element $\II$, with the extra multiplication rule $\II e_S = e_S \II = e_S$, and $S$ now runs through finite ordered sequences of the elements of $G\sqcup G^*$, where $G^*$ is in bijection with $G$, denoted by $*\colon G\to G^*$  and its inverse also by also $*\colon G^* \to G$. The $*$-involution is defined on the basis as $\II^* = \II$ and $e_S^* = 
 e_{S^*}$, where $S^* = (*g_k,\ldots,*g_1)$ for $S=(g_1,\ldots, g_k)$, and extended to the whole linear space by complex anti-linearity.
\end{remark}
Let us pass to the discussion of how to impose some algebraic relations on the algebra ${\cal A}_G$ freely  generated by $G$. To be concrete, think of an algebra ${\cal A}_G$ freely generated by $G$ and assume that we want to impose the relation $l$ stating that $\II a - a \II =0$ for all $a \in {\cal A}_G$ and for a preferred element $\II \in {\cal A}_G$ which will become the identity element of a new algebra  ${\cal A}_{G,l}$.  
 We can define ${\cal A}_{G,l} \cong {\cal A}_G / {\cal I}_l$, where ${\cal I}_l \subset {\cal A}_G$ is the two-sided ideal (resp.~$*$-ideal, in the case of $*$-algebras) generated by $l$, the  set of finite linear combinations of products of  $(\II a - a \II )$ and any other elements of ${\cal A}_G$.
In case a set $R$ of relations  is imposed, one similarly  takes the quotient with respect to the intersection 
${\cal I}_R$ of the ideals ($*$-ideals if working with $*$-algebras)
generated by each relation separately, ${\cal A}_{G,R} \cong {\cal A}_G / {\cal I}_R$.\\
The constructed algebra ${\cal A}_{G,R}$ satisfies the following abstract definition which again relies on a universal property.

\begin{definition}[Presentation by generators and relations]
Given an algebra ${\cal A}_G$ free on $G$ and a set $R$ whose elements are called {\bf relations} (again, not necessarily finite or even countable), together with a map $\rho\colon R \to {\cal A}_G$, an algebra ${\cal A}_{G,R}$ is said to be {\bf presented by the generators $G$ and relations $R$} if there exists an algebra homomorphism $r\colon {\cal A}_G \to {\cal A}_{G,R}$ such that, for any other algebra ${\cal B}$ and map $\beta\colon G\to {\cal B}$ such that the composition of the relations with the canonical homomorphism $b\colon {\cal A}_G \to {\cal B}$ gives $b \circ \rho = 0$, there exists a \emph{unique} algebra homomorphism $b_R\colon {\cal A}_{G,R} \to {\cal B}$ such that $b = b_R \circ r$. We use the same terminology for $*$- and unital algebras.
\end{definition}
\begin{remark}
Analogously to the case of ${\cal A}_G$, this definition easily implies that any two algebras ${\cal A}_{G,R}$,
${\cal A}'_{G,R}$ presented by the generators $G$ and relations $R$ are {\em naturally isomorphic} as the reader can immediately prove by using the universal property of the definition.
 Intuitively, the algebra ${\cal A}_{G,R}$ is therefore {\em the} algebra that is generated by $G$ satisfying \emph{only} the relations $\rho(R) = 0$.
\end{remark}
The presentation in terms of generators and relations works for a variety of algebraic structures, like groups, rings, module, algebras, etc. In fact, the universal property of objects defined in this way is most conveniently expressed using commutative diagrams in the corresponding category~\cite[\textsection I.11]{lang}. The case of groups is extensively discussed in~\cite[\textsection I.12]{lang}. Note that, though uniqueness of these objects is guaranteed by abstract categorical reasoning, their existence is not automatic and must be checked in each category of interest.

\subsection{The GNS construction}
When adopting the algebraic formulation, the notion of (quantum) state must be similarly generalized as follows.

\begin{definition}[States]\label{defstates} Given  an unital  $*$-algebra ${\cal A}$, an (algebraic) {\bf state} $\omega$ over $\cal A$  is a $\mathbb C$-linear map
$\omega : {\cal A} \to \mathbb C$ which is {\em positive} (i.e. $\omega(a^*a) \geq 0$ for all $a \in {\cal A}$) and {\em normalized} (i.e. $\omega(\II)=1$).
\end{definition}

 \noindent   The overall idea underlying this definition is that if, for a given observable  $a=a^* \in \cal A$
 we know all moments $\omega(a^n)$, and thus all expectation values of  polynomials $\omega(p(a))$, we also know the probability distribution associated to every value of $a$
when the state is $\omega$. To give a precise meaning to this idea, we should represent observables $a$ as self-adjoint operators  $\hat{a}$ in some Hilbert space ${\cal H}$, where the values of $a$ correspond to the point of  spectrum $\sigma(\hat{a})$ and the mentioned probability distribution is that generated by a vector $\Psi$ state representing $\omega$ in ${\cal H}$, and the spectral measure of $\hat{a}$. We therefore expect that, in this picture,  $\omega(a) = \langle \Psi| \hat{a} \Psi\rangle$  for some normalized vector $\Psi \in {\cal H}$. This is, in fact, a consequence of the content of the celebrated GNS re-construction procedure for unital $C^*$-algebras \cite{haag, sch,moretti}.  
We will discuss shortly the unital $*$-algebra version of that theorem.
 Note that the general problem of reconstructing even a unique classical state (a probability distribution on phase space) from the knowledge of all of its polynomial moments is much more difficult and is sometimes impossible (due to non-uniqueness). This kind of reconstruction goes under the name of the Hamburger moment problem~\cite[\textsection X.6 Ex.4]{rs2}. In this case, the successful reconstruction of a representation from a state succeeds because of the special hypotheses that go into the GNS theorem, where we know not only the expectation values of $a$ (and the polynomial $*$-algebra generated by it) but also those of all  elements of the algebra of observables.

 In the rest of the chapter $\Lf(V)$ will denote the linear space of linear operators  $T: V\to V$ on the vector space $V$.

\begin{definition}[$*$-Representations]  Let ${\cal A}$ be  a complex  algebra and ${\cal D}$ a dense linear subspace of the Hilbert space ${\cal H}$.\\
{\bf (a)} A  map $\pi : {\cal A} \to  \Lf({\cal D})$ such that it is linear and  product preserving 
is called {\bf representation} of ${\cal A}$ on ${\cal H}$ with {\bf domain} ${\cal D}$. 	
If ${\cal A}$ is furthermore unital, 
a representation is also required to satisfy:  $\pi(\II)=I$. \\
{\bf (b)} If finally ${\cal A}$ is a $*$-algebra, a  {\bf $*$-representation}   of ${\cal A}$ on ${\cal H}$ with {\bf domain} ${\cal D}$ is a representation which satisfies (where $^\dagger$ henceforth denotes the Hermitian adjoint operation in ${\cal H}$)
   $$\pi(a)^\dagger\spa\rest_{{\cal D}} = \pi(a^*) \quad \forall  a\in {\cal A}\:.$$
\end{definition}

As a general result we have the following elementary proposition

\begin{proposition}[On faithful representations] \label{propsimp}
If ${\cal A}$ is a  complex  algebra is simple, then  every representation is either faithful -- i.e., injective -- or  it is the zero representation. 
\end{proposition}

\begin{proof}
If $\pi : \Ac \to \Lf({\cal D})$ is a $*$-representation, $Ker(\pi)$ is evidently a two-sided ideal. Since $\Ac$ is simple there are only two possibilities either $Ker(\pi) = {\cal D}$ so that $\pi$ is the zero representation, or 
$Ker(\pi) = \{0\}$ and thus $\pi$ is injective. $\Box$
\end{proof}

\begin{theorem}[GNS construction]\label{GNS}
If ${\cal A}$ is a complex unital $*$-algebra and $\omega : {\cal A} \to \mathbb C$ is a state, the following facts hold.\\
{\bf (a)} There is a quadruple  $({\cal H}_\omega, {\cal D}_\omega,  \pi_\omega, \Psi_\omega)$, where:\\

(i)  ${\cal H}_\omega$ is a (complex) Hilbert space,\\

(ii)  ${\cal D}_\omega \subset {\cal H}_\omega$ is a dense subspace, \\

(iii) $\pi_\omega : {\cal A} \to \Lf({\cal D}_\omega)$ a  $*$-representation of ${\cal A}$ on ${\cal H}_\omega$ with domain ${\cal D}_\omega$,\\

(iv) $\pi_\omega({\cal A})\Psi_\omega = {\cal D}_\omega$,\\

(v)   $\omega(a) = \langle \Psi_\omega | \pi_\omega(a) \Psi_\omega \rangle$ for every $a\in {\cal A}$.\\

\noindent {\bf (b)} If $({\cal H}'_\omega, {\cal D}'_\omega,  \pi'_\omega, \Psi'_\omega)$
satisfies (i)-(v), then there is $U : {\cal H}_\omega \to {\cal H}'_\omega$ surjective and isometric such that:

(i) $U\Psi_\omega = \Psi'_\omega$, \\

(ii) $U {\cal D}_\omega = {\cal D}'_\omega$,\\

(iii)  $U\pi_\omega(a) U^{-1} = \pi'_\omega(a)$ if $a \in {\cal A}$\:. 
\end{theorem}

\begin{proof}
Consider ${\cal A}$ as complex vector space  and define $N=\{a \in{\cal A}\:|\: \omega(a^*a)=0\}$. $N$ is a subspace as easily follows from sesquilinearity of $(a,b) \mapsto \omega(a^*b)$ and from the Cauchy-Schwartz inequality which holds because  $(a,b) \mapsto \omega(a^*b)$  is non-negative. Define ${\cal D}_{\omega}   \stackrel {\mbox{\scriptsize  def}} {=} {\cal A}/N$ as a complex vector space and equip it with the Hermitian scalar product 
$\langle [a]| [b] \rangle   \stackrel {\mbox{\scriptsize  def}} {=} \mu(a^*b)$, which turns  out to be well-defined (because $\omega(a^*b)= \omega(b^*a)=0$ if $a\in N$ again from Cauchy-Schwartz inequality) and positive.   ${\cal H}_\omega$
is, by definition, the completion of ${\cal D}_{\omega}$ with respect to the mentioned scalar product.
Now observe that $N$ is also a left-ideal  ($\omega((ba)^*ba) = \omega((b^*(ba))^*a)=0$ if $a\in N$)
and consequently $\pi_{\omega}(a)[b]   \stackrel {\mbox{\scriptsize  def}} {=} [ab]$ is well-defined ($[ab]=[ac]$ if $c\in [b]$) and is a unital algebra representation.
Defining $\Psi_\omega   \stackrel {\mbox{\scriptsize  def}} {=} [1]$, we have $\omega(a) = \langle \Psi_\omega| \pi_\omega(a) \Psi_\omega\rangle$.
Finally:
$$\langle \pi_\omega(c) \Psi_\omega| \pi_\omega(a) \pi_\omega(b) \Psi_\omega \rangle =
\omega(c^*(a^*)^*b)= \omega((a^*c)^*b) = \langle \pi_\omega(a^*c) \Psi_\omega| \pi_\omega(b) \Psi_\omega \rangle$$
$$=   \langle \pi_\omega(a^*) \pi_\omega(c) \Psi_\omega| \pi_\omega(b) \Psi_\omega \rangle $$
Summing up, we have:
$$\langle \pi_\omega(a)^\dagger \pi_\omega(c) \Psi_\omega| \pi_\omega(b) \Psi_\omega \rangle = \langle \pi_\omega(c) \Psi_\omega| \pi_\omega(a) \pi_\omega(b) \Psi_\omega \rangle$$ $$ = \langle \pi_\omega(a^*) \pi_\omega(c) \Psi_\omega| \pi_\omega(b) \Psi_\omega \rangle$$
Since $c, b$ are arbitrary and both $\pi_\omega(b)\Psi_\omega$ and $\pi_\omega(b)\Psi_\omega$ range in ${\cal D}_\omega$ which is dense,   we have found that $\pi(a)^\dagger|_{{\cal D}_\omega} = \pi(a^*)$.
The proof of (b) is easy. As a matter of fact the operator $U$ is completely defined by
$U\pi_\omega(a) \Psi_\omega   \stackrel {\mbox{\scriptsize  def}} {=} \pi'_\omega(a)\Psi'_\omega$, we leave to the reader the proof of the fact that it is well-defined and satisfies the required properties. The proof is strictly analogous to the corresponding part of (b) in Proposition \ref{propsymm}  below.  $\Box$
\end{proof}

\noindent There exists a stronger version of that theorem \cite{haag, BR, moretti} regarding the case where  ${\cal A}$ is a unital $C^*$-algebra.  The quadruple $({\cal H}_\omega, {\cal D}_\omega,  \pi_\omega, \Psi_\omega)$ is called {\bf GNS triple} (!) the name is due to the fact that for $C^*$-algebras ${\cal D}_\omega = {\cal H}_\omega$. In that case  the representation $\pi_\omega$ is continuous (norm decreasing more precisely)  with  respect to the operator norm $||\:\:||$ in $\Bf({\cal H}_\omega)$, since 
$\pi_{\omega}(a) \in \Bf({\cal H}_\omega)$ if $a\in {\cal A}$.  \\ As a general fact, we have that a $*$-representations $\pi$ of a unital  $C^*$-algebra ${\cal A}$ on a Hilbert space ${\cal H}$ assuming values in  $\Bf({\cal H})$  is automatically norm decreasing, with  respect to the operator norm $||\:\:||$ in $\Bf({\cal H})$. Moreover $\pi$ is isometric  if and only if it is injective \cite{haag,BR}.

\begin{remark}\label{rem1}$\null$

{\bf (1)}  Since ${\cal D}_\omega$ is dense $\pi_\omega(a)^\dagger$ is always well defined and, in turn, densely defined 
for (iii) in (a).  Hence, $\pi_\omega(a)$ is always closable. Therefore,
 if $a=a^*$, $\pi(a)$ is at least symmetric.  If $\pi(a)$ is self-adjoint the probability distribution of the observable $a$ in the state $\omega$  mentioned in the comment after Def. \ref{defstates} is ${\cal B}(\mathbb R) \ni E \mapsto \langle \Psi_\omega| P^{(\pi_\omega(a))}_E \Psi_\omega\rangle$, where ${\cal B}(\mathbb R)$ is the class of Borel sets on $\mathbb R$ and $P^{(\pi_\omega(a))}$ the projection-valued measure of $\pi_\omega(a)$. The precise technical conditions, and their physical significance, under which an operator $\pi(a)$, with $a=a^*$, might be essentially self-adjoint on ${\cal D}_\omega$ are poorly explored in the literature and deserve further investigation.

{\bf (2)}  The {\bf  weak commutant} $\pi'_w$ of a $*$-representation $\pi$ of ${\cal A}$ on ${\cal H}$ with  domain ${\cal D}$,  is defined as\footnote{$\pi'_w$ can equivalently be defined as $\{A \in \Bf({\cal H}) \:\:|\:\: A\phi(a)= \pi(a^*)^\dagger A\:,\quad \forall a \in {\cal A}\}$.}
$$\pi'_w  \stackrel {\mbox{\scriptsize  def}} {=} \{ A \in \Bf({\cal H}) \:\:|\:\:  \langle \psi | A \pi(a) \phi \rangle = \langle \pi(a)^\dagger \psi| A \phi \rangle \quad \forall a \in {\cal A}\:, \forall \psi, \phi \in {\cal D}\}\:,$$
where $\Bf({\cal H})$ denotes the $C^*$-algebra of all bounded operators on ${\cal H}$. If ${\cal A}$ is a unital $C^*$-algebra, the weak commutant of $\pi$ (with domain given by the whole Hilbert space) coincides to the standard commutant.
We say that a $*$-representation  $\pi$ of ${\cal A}$  on ${\cal H}$ is {\bf weakly irreducible} if its  weak commutant is trivial, that is, it coincides with the set of operators $cI : {\cal H} \to {\cal H}$ for $c\in \mathbb C$. 

{\bf (3)}  The set of states  over the unital $^*$-algebra ${\cal A}$ is a {\bf convex body}. In other words {\em convex combinations} of states are states: $\omega = p \omega_1 + (1-p)\omega_2$ with $p \in (0,1)$ is a state if $\omega_1, \omega_2$ are.

{\bf (4)} A state $\omega$ is said to be {\bf extremal} if  $\omega = p \omega_1 + (1-p)\omega_2$, with $p \in (0,1)$ and $\omega_1, \omega_2$  are states, is possible only if $\omega_1=\omega_2 (= \omega)$.  These states are also called {\bf pure states}. It is possible to prove the following \cite{haag,sch}: 
\begin{proposition}[Pure states and irreducible representations]  Referring to the hypotheses of Theorem \ref{GNS}, $\omega$
is pure if and only if $\pi_\omega$ is weakly  irreducible.
\end{proposition}
(If ${\cal A}$ is a unital $C^*$-algebra the same statement holds but ``weakly'' can be omitted.)
 Therefore, even if $\omega$ is represented by a unit vector $\Psi_\omega$  in ${\cal H}_\omega$, it does not mean that $\omega$ is pure. In standard quantum mechanics it happens because 
${\cal A}$ is implicitly assumed to coincide to the whole $C^*$-algebra  $\Bf({\cal H})$ of everywhere-defined bounded operators over ${\cal H}$ and $\pi_\omega$ is the identity when $\omega$ corresponds to a vector state of ${\cal H}$.

{\bf (5)}  When ${\cal A}$ is a unital $C^*$-algebra, the convex body of states on ${\cal A}$ is hugely larger that the states of the form
$${\cal A}\ni  a \mapsto \omega_\rho(a)   \stackrel {\mbox{\scriptsize  def}} {=} tr(\rho \pi_\omega(a))$$ for a fixed (algebraic) state $\omega$ and where $\rho \in \Bf({\cal H}_\omega)$ is a positive trace class operator with unit trace. These trace-class operators  states associated with an algebraic state $\omega$ form the {\bf folium} of $\omega$ and are called {\bf normal states} in ${\cal H}_\omega$.   If ${\cal A}$ is not $C^*$, the trace  $tr(\rho \pi_\omega(a))$ is not defined in general, because  $ \pi_\omega(a)$ is not bounded and $\rho \pi_\omega(a)$ may not be  well defined nor  trace class in general. Even if $\Ac$ is just a unital $*$-algebra, a unit vector  $\Phi \in {\cal D}_\omega$   defines however  a state by means of 
$${\cal A}\ni  a \mapsto  \omega_{\Phi}(a)   \stackrel {\mbox{\scriptsize  def}} {=} \langle \Phi| \pi_\omega(a) \Phi \rangle\:,$$
 recovering the standard formulation of elementary quantum mechanics. {\em These states are pure when $\omega$
is pure}. More strongly, in this situation $({\cal H}_\omega, {\cal D}_\omega, \pi_{\omega}, \Phi)$ is just a GNS triple of $\omega_{\Phi}$, because $\pi_{\omega}(\Ac)\Phi$ is dense in ${\cal H}_\omega$  (this is because the orthogonal projector onto $\pi_{\omega}(\Ac)\Phi$ cannot vanish and  belongs to the weak commutant $\pi'_{\omega w}$ which is trivial, because $\omega$ is pure.)
 If $\omega$ is not pure,  $\omega_{\Phi}$ may  not be  pure also if it is  represented by a unit vectors.

{\bf (6)} There are unitarily non-equivalent GNS representations of the same unital $*$-algebra ${\cal A}$ associated with states $\omega$, $\omega'$. In other words  there is no surjective isometric operator $U: {\cal H}_\omega \to {\cal H}_{\omega'}$ such that $U\pi_\omega(a) = \pi_{\omega'}(a) U$ for all $a \in {\cal A}$. (Notice that, in the notion of {\em unitary equivalence} it is not required that $U\Psi_{\omega}= \Psi_{\omega'}$). Appearance of unitarily inequivalent representations is natural when ${\cal A}$ has a non-trivial {\em center},  ${\cal Z}_{{\cal A}}$, i.e., it
 contains something more than the elements $c\II$ for $c\in \mathbb C$. 
 Pure states $\omega, \omega'$ such that $\omega(z) \neq \omega'(z)$ for some $z \in {\cal Z}_{\cal A}$ give rise to unitarily inequivalent GNS representations. This easily follows from the fact that $\pi_\omega(z)$ and $\pi_{\omega'}(z)$, by irreducibility of the representations, must be operators of the form $c_zI$ and $c'_zI$ for complex numbers $c_z,c'_z$ in the respective Hilbert spaces ${\cal H}_\omega$ and ${\cal H}_{\omega'}$.
It should be noted that such representations remain inequivalent even if the unitarity of $U$ is relaxed. However, it can happen that some representations are unitarily inequivalent even when the algebra has a trivial center. See Section~\ref{sec:uni-ineq} for a relevant example.
\end{remark}

\begin{remark}
The positivity requirement on states is physically meaningful when every self-adjoint element of the $*$-algebra is a physical observable. It is also a crucial ingredient in the GNS reconstruction theorem. However, in the treatment of gauge theories in the Gupta-Bleuler or BRST formalisms, in order to keep spacetime covariance, one must enlarge the $*$-algebra to include unobservable or {\em ghost} fields. Physically meaningful states are then allowed to fail the positivity requirement on $*$-algebra elements generated by ghost fields. The GNS reconstruction theorem is then not applicable and, in any case, the $*$-algebra is expected to be represented on an indefinite scalar product space (a {\em Krein space}) rather than a Hilbert space. Fortunately, several extensions of the GNS construction have been made, with the positivity requirement replaced by a different one that, instead, guarantees the reconstructed $*$-representation to be on an indefinite scalar product space. Such generalizations and their technical details are discussed in~\cite{hofmann}.
\end{remark}

\noindent Another relevant result arising from the GNS theorem concerns {\em symmetries} represented by $*$-algebra (anti-linear) automorphisms.

\begin{proposition}[Automorphisms induced by invariant states]\label{propsymm} Let ${\cal A}$ be an unital $*$-algebra, $\omega$ a state on it and consider its GNS representation. The following facts hold.\\
{\bf (a)} If $\beta : {\cal A} \to {\cal A}$ is a unital $*$-algebra automorphism (resp. anti-linear automorphism) which leaves fixed $\omega$, i.e., $\omega \circ \beta = \omega$, then there exist a unique bijective bounded operator $U^{(\beta)} : {\cal H}_\omega \to {\cal H}_\omega$ such that:\\

(i)\:\: $U^{(\beta)} \Psi_\omega = \Psi_\omega$ \: and \: $U^{(\beta)} ({\cal D}_\omega) = {\cal D}_\omega$\:,
 \\

(ii) \:\: $U^{(\beta)} \pi_\omega(a) U^{(\beta)-1} x = \pi_\omega\left(\beta(a) \right)x$\:\: if $a \in {\cal A}$ and $x \in {\cal D}_\omega$.\\

\noindent  $U^{(\beta)}$ turns out to be unitary (resp. anti-unitary).

\noindent{\bf (b)} If, varying $t \in \mathbb R$, $\beta_t : {\cal A} \to {\cal A}$ defines a one-parameter group of  unital $*$-algebra automorphisms\footnote{There do not exist  one-parameter group of  unital $*$-algebra {\em anti-linear} automorphisms, this is because  $\beta_t = \beta_{t/2}\circ \beta_{t/2}$ is linear both for $\beta_{t/2}$ linear or anti-linear.} which leaves fixed $\omega$, the corresponding unitary operators $U^{(\beta)}_t$ as in (a) define a one-parameter group of unitary operators in ${\cal H}_\omega$.\\
{\bf (c)} $\{U^{(\beta)}_t\}_{t \in \mathbb R}$ as in (b) is strongly continuous (and thus it admits a self-adjoint generator) if and only if $$\lim_{t \to 0} \:\omega(a^*\beta_t(a))= \omega(a^*a)\quad \mbox{for every $a \in {\cal A}$.}$$ 
\end{proposition}

\begin{proof} Let us start from (a) supposing that $\beta$ is a $*$-automorphism. If an operator satisfying (i) and (ii) exists it also satisfies $U^{(\beta)} \pi_\omega(a)\Psi_\omega =  \pi_\omega\left(\beta(a) \right) \Psi_\omega$. Since $ \pi_\omega({\cal A})\Psi_\omega$ is 
dense in ${\cal H}_\omega$, this identity determines $U^{(\beta)}$ on ${\cal D}_\omega$.
Therefore we are lead to try to define $U_0^{(\beta)} \pi_\omega(a)\Psi_\omega   \stackrel {\mbox{\scriptsize  def}} {=}  \pi_\omega\left(\beta(a) \right) \Psi_\omega$ on ${\cal D}_\omega$. 
From (v) in (a) of Theorem \ref{GNS} it immediately arises that $||U_0^{(\beta)} \pi_\omega(a)\Psi_\omega ||^2 = ||\pi_\omega(a)\Psi_\omega ||^2$. That identity on the one hand proves that $U^{(\beta)}$
is well defined because if $\pi_\omega(b)\Psi_\omega = \pi_\omega(b')\Psi_\omega$ then 
$U^{(\beta)}\pi_\omega(b)\Psi_\omega =U^{(\beta)} \pi_\omega(b')\Psi_\omega$, on the other hand it proves that $U^{(\beta)}$ is isometric on ${\cal D}_\omega$. If we analogously define the other isometric operator $V_0^{(\beta)} \pi_\omega(a)\Psi_\omega   \stackrel {\mbox{\scriptsize  def}} {=}  \pi_\omega\left(\beta^{-1}(a) \right) \Psi_\omega$ on ${\cal D}_\omega$, we see that $U^{(\beta)}_0V x= VU^{(\beta)}_0x$ for every $x\in {\cal D}_{\omega}$. Since ${\cal D}_\omega$ is dense in ${\cal H}_\omega$, these identities extend to analogous identities for the unique bounded extensions of $U^{(\beta)}_0$ and $V$ valid over the whole Hilbert space. In particular the former operator extends into an isometric surjective operator (thus unitary) $U^{(\beta)}$ which, by construction, satisfies (i) and (ii). Notice that $V$, defined on ${\cal D}_\omega$, is the inverse of $U_0^{(\beta)}$ so that, in particular $U^{(\beta)}({\cal D}_\omega)=  U_0^{(\beta)}({\cal D}_\omega) = {\cal D}_\omega$.
 The followed procedure also proves that $U^{(\beta)}$ is uniquely determined by (i) and (ii). The anti-linear case is proved analogously. Anti-linearity of $\beta$ implies that, in $U_0^{(\beta)} \pi_\omega(a)\Psi_\omega   \stackrel {\mbox{\scriptsize  def}} {=}  \pi_\omega\left(\beta(a) \right) \Psi_\omega$,
$U_0^{(\beta)}$ must be anti-linear and thus anti-unitary.\\
 The proof of (b) immediately arises from (a). Regarding (c), we observe that, if $x = \pi_\omega(a)\Psi_\omega$ one has for $t\to 0$ by the GNS theorem,
$$\langle x | U_t^{(\beta)}x \rangle = \omega(a^*\beta_t(a)) \to \omega(a^*a)=
\langle x | x \rangle$$
Since the span of the vectors $x$ is dense in ${\cal H}_\omega$, $U^{(\beta)}_t$ is strongly continuous  due to Proposition 9.24 in \cite{moretti}. $\Box$
\end{proof}

\begin{remark} \label{remgroupcont} $\null$

{\bf (1)}  Evidently, the statements (b) and  (c) can immediately be generalized  to the case of a representation of a generic {\em  group} or, respectively, {\em connected topological group}, $G$.
Assume that $G$ is represented in terms of automorphisms of unital $*$-algebras  $\beta_g : \Ac \to \Ac$ for $g\in G$.  With the same proof of (c), it turns out that, if $\omega$ is invariant under this representation of $G$, the associated representation in the GNS Hilbert space of $\omega$, $\{U^{(\beta)}_g\}_{g \in G}$  is strongly continuous if and only if $$\lim_{g \to e} \:\omega(a^*\beta_g(a))= \omega(a^*a)\quad \mbox{for every $a \in {\cal A}$,}$$ 
where $e\in G$ is the unit element.

{\bf (2)} It could happen in physics that an {\em algebraic symmetry}, i.e., an automorphism (or anti-automorphism) $\beta : \Ac \to \Ac$ exists for a unital $*$-algebra with some physical interpretation, but that this symmetry cannot be completely implemented unitarily (resp. anti-unitarily) in the GNS representation of a state $\omega$ because, referring to the condition in (a) of the proved theorem, either (i) or {\em both} (i) and (ii) of (a) do not hold. In the first case the symmetry is broken because the cyclic vector is not invariant under a unitary representation of the symmetry, which however exists in the GNS representation of $\omega$. Obviously, in this case, $\omega$ is not invariant under the algebraic symmetry. This situation naturally arises when one starts from a pure invariant state $\omega_0$ and  the physically relevant state is not $\omega_0$, but another state $\omega\in {\cal D}_{\omega_0}$. The second, much more severe, situation is when there is no unitary map in the GNS representation of $\omega$ which fulfills (i) and  (ii).  In algebraic quantum theories, this second case is often called {\em spontaneous breaking of symmetry}.

\end{remark}
\section{The $*$-algebra of a  quantum field and its quasifree states }
This chapter mostly deals with the case of a real scalar field, we will denote by $\phi$, on a given always oriented and  time oriented,   {\em globally hyperbolic spacetime}   $\Mb =({\cal M}, g, \ogth,\tgth)$ of dimension $n \geq 2$, where $g$ is the metric with signature $(+,-,\dots,-)$, $\ogth$ the {\em orientation} and $\tgth$ the {\em time orientation}. Regarding geometrical notions, we adopt throughout the definitions of \cite{chapt:BD}. Minkowski spacetime will be denoted by $\mathbb M$ and its metric by $\eta$.

The results we discuss can be extended to charged and higher spin fields. As is well known a quantum field is a {\em locally covariant notion}, functorially defined in {\em all} globally hyperbolic  spacetimes simultaneously (see \cite{chapt:CR}). Nevertheless, since this chapter is devoted to discussing algebraic  {\em states} of a QFT in a given manifold we can deal with a fixed spacetime. Moreover we shall not construct the $*$-algebras as  {\em Borchers-Uhlmann-like}  algebras (see \cite{chapt:BD})  nor use  the {\em deformation  approach} (see \cite{chapt:RKK})  to define the algebra structure, in order to simplify the technical structure and focus on the properties of the states.

\subsection{The algebra of observables of a real scalar Klein-Gordon field}\label{sec:kg}  In order to  deal with QFT in curved spacetime, a convenient framework is the algebraic one. This is due to various reasons. Especially because, in the absence of Poicar\'e symmetry, there is no preferred Hilbert space representation of the field operators, but several unitarily inequivalent representations naturally show up.
Furthermore, the standard definition of the field operators based on the decomposition of field solutions in positive and negative frequency part is not allowed here, because there is no preferred notion of (Killing) time.  

In the rest of the chapter $ C_0^\infty(\Mc)$ denotes the real vector space of compactly-supported  and {\em real}-valued smooth  function on the manifold $\Mc$.

 The elementary algebraic object, i.e., a {\bf scalar quantum field} $\phi$ over the  globally hyperbolic spacetime $\Mb$ is captured by a unital $*$-algebra  ${\cal A}(\Mb)$
called the CCR algebra of the quantum field $\phi$.

\begin{definition}[CCR algebra]\label{def:ccr}
The {\bf CCR algebra} of the quantum field $\phi$ over $\Mb$ is the unital $*$-algebra presented by the following generators and relations (cf.~Section~\ref{sec:genrel}). The generators consist ({\bf smeared abstract}) {\bf field operators}, $\phi(f)$, labeled by functions $f \in C_0^\infty(\Mc)$ (the identity $\II$ is of course included in the construction of the corresponding freely generated algebra). These generators satisfy the following relations:\\

{\bf $\mathbb R$-Linearity}: $\phi(af+bg) - a\phi(f) - b\phi(g) = 0$\quad if  $f,g \in C_0^\infty(\Mc)$ and $a,b \in  \mathbb R$.\\

{\bf Hermiticity}: $\phi(f)^* - \phi(f) = 0$\quad  for  $f \in C_0^\infty(\Mc)$.\\

{\bf Klein-Gordon}:  $\phi\left( (\Box_\Mb + m^2 + \xi R) g\right) =0$\:\:   for  $g \in C_0^\infty(\Mc)$.\\

{\bf  Commutation relations}: $[\phi(f), \phi(g)] - iE(f,g) \II = 0$\quad  for $f \in C_0^\infty(\Mc)$.
\end{definition}

\noindent 
Above $E$ denotes the {\em advanced-minus-retarded fundamental solution}, also called the {\em causal propagator}, see \cite{chapt:BD} and (1) in Remark \ref{remark2} below.
The Hermitian elements of ${\cal A}(\Mb)$ are the {\em elementary observables} of the free field theory associated with the Klein-Gordon field $\phi$. The non-Hermitian elements play an auxiliary r\^{o}le.
It should however be evident that ${\cal A}(\Mb)$  is by no means sufficient to faithfully describe physics involved with the quantum field $\phi$. For instance ${\cal A}(\Mb)$ does not include 
any element which can be identified with the {\em stress energy tensor} of $\phi$.  Also the local interactions like $\phi^4$ cannot be described as elements of this algebra either. We shall tackle this problem later. 

According to the discussion in Section~\ref{sec:genrel}, the above abstract definition is sufficient to uniquely define ${\cal A}(\Mb)$ up to isomorphism. An alternative, more concrete and explicit, construction using tensor products of spaces $C_0^\infty(\Mc)$  is presented in \cite{chapt:BD}. That construction yields a concrete representative of the isomorphism class of ${\cal A}(\Mb)$.

\begin{remark}$\label{remark2}\null$

{\bf (1)} Let  $\Sol$ indicate the real vector space of {\em real} smooth solutions $\psi$ with {\em compact Cauchy data} of the KG equation $ (\Box_\Mb + m^2 + \xi R)\psi=0$ where $\Box_\Mb   \stackrel {\mbox{\scriptsize  def}} {=} g^{ab}\nabla_a \nabla_b$. Let us, as usual, use the notation ${\cal D}(\Mc)   \stackrel {\mbox{\scriptsize  def}} {=} C_0^\infty(\Mc) \oplus i C_0^\infty(\Mc)$ for the space of {\em complex} test functions and ${\cal D}'(\Mc)$ is the dual space of distributions.
Interpreting the {\bf advanced-minus-retarded fundamental solution} of the KG operator as a linear map $$E : C_0^\infty(\Mc) \to  \Sol\:,$$ we can naturally extend it by $\mathbb{C}$-linearity to the continuous linear map $$E :  {\cal D}(\Mc) \to  {\cal D}'(\Mc)\:,$$ which defines the bilinear functional
\begin{eqnarray}
E(f_1,f_2)   \stackrel {\mbox{\scriptsize  def}} {=} \int_{\Mc} f_1 (E f_2)\, \dvol_{\Mb} \quad   \mbox{if} \quad  f_1,f_2 \in  C_0^\infty(\Mc) \:,\label{1}
\end{eqnarray}
which is the one appearing in the commutation relations above. Of course, in agreement with the commutation relations,
\begin{eqnarray}
\mathbb R \ni E(f,g) = -E(g,f) \quad \mbox{if $f,g \in C_0^\infty(\Mc)$.}
\end{eqnarray}
As a map  $C_0^\infty(\Mc)\to \Sol$, $E$ satisfies
\begin{equation}  Ker(E) = \{  (\Box_\Mb + m^2 + \xi R)h \:|\:  h \in {\cal D}(\Mc) \}\label{0}\:.\end{equation}
Everything is a consequence of the fact that  $\Mb$ is globally hyperbolic (see \cite{chapt:BD}). 
Since $E(f,h)=0$ if the support of $f$ does not intersect $J_\Mb^+(\mbox{supp} h)\cup J_\Mb^-(\mbox{supp} h)$, we immediately have from the {\em commutation relations} requirement  that the following important fact holds, distinguishing observable fields (Bosons) form unobservable ones (Fermions):
\begin{proposition}[Causality]
Referring to ${\cal A}(\Mb)$,
 $\phi(f)$ and $\phi(h)$ commute if the supports of $f$ and $h$ are causally separated.
\end{proposition}
 From standard properties of $E$ (see \cite{chapt:BD}) one also finds, if $\Sigma \subset \Mc$ is a smooth 
space-like Cauchy surface $f,h \in C_0^\infty(\Mc)$ and $\psi_f   \stackrel {\mbox{\scriptsize  def}} {=} Ef$ and $\psi_h   \stackrel {\mbox{\scriptsize  def}} {=} Eh$ are elements of $\Sol$,
\begin{eqnarray}
E(f,g) = \int_{\Sigma}\left( \psi_f \nabla_{\nb} \psi_h -  \psi_h\nabla_{\nb} \psi_f\right) \: d \Sigma\:,\label{simpsol}
\end{eqnarray}
where $d\Sigma$ is the standard measure induced by the metric $g$ on $\Sigma$ and ${\nb}$ the future directed 
normal unit vector field to $\Sigma$.

{\bf (2)} As  $E :  {\cal D}(\Mc) \to  {\cal D}'(\Mc)$ is continuous, due to {\em Schwartz kernel theorem} \cite{hormander}, it  defines a distribution, 
indicated with the same symbol $E \in  {\cal D}' (\Mc \times \Mc)$, uniquely determined by 
 $$E(f_1,f_2) = E(f_1\otimes f_2)\quad f_1,f_2 \in  {\cal D}(\Mc)\:, $$ and this leads to  an equivalent interpretation of the left-hand side of (\ref{1}), which is actually a bit more useful, because it permits to consider the action of $E$ on 
non-factorized test functions $h \in  {\cal D}'(\Mc \times \Mc)$.

{\bf (3)} The condition indicated as {\em Klein-Gordon} is  the requirement that $\phi$ {\em distributionally} 
satisfies the equation of Klein-Gordon. Obviously $\Box_{\Mb}$ appearing in it coincides with its formal transposed (or adjoint) operator which should appear in the distributional version of KG equation.

{\bf (4)} Everything we will say holds equally for $m^2$ and $\xi$ replaced by corresponding smooth real functions, also in the case where $m^2$ attains negative values.  Also the case $m^2<0$ does not 
produce  technically difficult problems.
\end{remark}

\noindent  Linearity and  Commutation relations  conditions together with (\ref{0}) imply the elementary but important result which proves also the converse implication in the property {\em Klein-Gordon}.

\begin{proposition}\label{iff}
Referring to $\Ac(\Mb)$ the following facts hold.
\begin{eqnarray}
\phi(f) = \phi(g) \quad \mbox{if and only if $f-g \in Ker(E)$,}
\end{eqnarray}
so that, in particular, 
\begin{eqnarray}
\phi(f) = 0 \quad \mbox{if and only if $f= (\Box_\Mb + m^2 + \xi R) g$\:   for \: $g \in C_0^\infty(\Mc)$.}
\end{eqnarray}
\end{proposition}

\begin{proof} 
$\phi(f) = \phi(g)$ is equivalent to  $\phi(f-g)=0$ and thus $iE(h,(f-g)) = [\phi(h),\phi(f-g)]=0$
for all $h \in {\cal D}(\Mc)$. From (\ref{1}) one  has, in turn,  that $E(f-g)=0$ that is $f-g \in Ker(E)$. Finally (\ref{0}) implies the last statement. $\Box$
\end{proof}

\noindent The smeared field $\phi(f)$ can be thought of as localized within the support of its argument $f$. However, $\phi(f)$ really depends on $f$ only up to addition of terms from $Ker(E)$. We can use this freedom to move and shrink the support of $f$ to be arbitrarily close to any Cauchy surface, which is a technically useful possibility.
\begin{lemma}
 Let $\psi \in \Sol$ and  let $\Sigma$ be  a smooth space-like Cauchy surface of the globally hyperbolic spacetime $\Mb$. For every  open neighborhood $O$ of $\Sigma$, it is  possible to pick out  a function $f_\psi \in C_0^\infty(\Mc)$ {\em whose support is contained in} $O$, such that $\psi= Ef_\psi$.
\end{lemma}
 The proof of this elementary, but important, fact can be found in \cite{chapt:BD} and in \cite{wald94} (see also the proof of our Proposition~\ref{prp:hadform}). This  result  immediately implies the validity of the so called  {\em Time-slice axiom} for the CCR algebra (see \cite{chapt:BD}).

\begin{proposition}[``Time-slice axiom'']\label{tsa} Referring to the globally hyperbolic spacetime $\Mb$ and the algebra $\Ac(\Mb)$,  let $O$ be any fixed neighborhood of a Cauchy surface $\Sigma$. Then $\Ac(\Mb)$ is generated by $\II$ and the elements $\phi(f)$ with $f \in C_0^\infty(\Mc)$ and $\mathrm{supp} f \subset O$.
\end{proposition}

\subsection{States and $n$-point functions}
 Let us focus on states. We start form the observation that  the generic element of  ${\cal A}(\Mb)$ is always   of the form
\begin{multline}\label{a}
	a = c_{(0)}\II + \sum_{i_1}c_{(1)}^{i_1}\phi(f^{(1)}_{i_1}) + \sum_{i_1,i_2}c_{(2)}^{i_1i_2}\phi(f^{(2)}_{i_1})\phi(f^{(2)}_{i_2}) \\ {}
	+ \cdots + \sum_{i_1, \ldots,  i_n}  c_{(n)}^{i_1 \cdots i_n} \phi(f^{(n)}_{i_1})\cdots \phi(f^{(n)}_{i_n})\:,
\end{multline}
where $n$ is arbitrarily large but finite, $c_{(k)}^{i_1 \cdots i_k} \in \mathbb C$ and $f_k^{(j)} \in C_0^\infty(\Mc)$, with all sums arbitrary but  finite.  
Due to~\eqref{a}, if $\omega\colon {\cal A}(\Mb) \to \mathbb C$ is a state, its action on a generic element of ${\cal A}(\Mb)$
is known as soon as the full class of the so-called {\bf $n$-point functions} of $\omega$ are known. We mean the maps:
$$C_0^\infty(\Mc)\times \cdots \times C_0^\infty(\Mc) \ni  (f_1,\ldots, f_n)\:\: \mapsto  \:\: \omega(\phi(f_1)\cdots \phi(f_n))  \stackrel {\mbox{\scriptsize  def}}{=} \omega_n(f_1, \ldots, f_n)$$
At this point, the multilinear functionals $\omega_n(f_1,\ldots, f_n)$ are not yet forced to satisfy any continuity properties (in fact we have not even discussed any topologies on $\Ac(\Mb)$ and how the states should respect it). However, in the sequel we will only be dealing with the cases where $\omega_n$ is continuous in the usual test function topology on $C^\infty_0(\Mc)$. Then, by the Schwartz kernel theorem~\cite{hormander}, we can write, as it is anyway customary, the $n$-point function in terms of its distributional kernel: $$\omega_n(f_1, \ldots, f_n)  = \int_{\Mb^n} \omega_n(x_1,\ldots, x_n) f_1(x_1) \cdots f_n(x_n) \:\dvol_{\Mb^n} \: .$$
It is worth stressing that a choice of a family of integral kernels $\omega_n$, $n=1,2,\ldots$, extends by linearity and the rule $\omega(\II)   \stackrel {\mbox{\scriptsize  def}} {=} 1$ to a normalized linear functional on all of $\Ac(\Mb)$. However, this functional generally does {\em not} determine a state, because the positivity requirement $\omega(a^*a)\geq 0$ may not be valid. However if two states have the same set of $n$-point functions they necessarily coincide in view of~\eqref{a}.
\begin{remark}
As defined above, the $n$-point functions $\omega_n(f_1,\ldots,f_n)$ need not be symmetric in their arguments. However, they do satisfy some relations upon permutation of the arguments. The reason is that the products $\phi(f_1)\cdots \phi(f_n)$ and $\phi(f_{\sigma(1)})\cdots \phi(f_{\sigma(n)})$, for any permutation $\sigma$, are not completely independent in $\Ac(\Mb)$. It is easy to see that the CCR $*$-algebra is {\bf filtered}, namely that $\Ac(\Mb) = \bigcup_{n=0}^{\infty} \Ac_n(\Mb)$, where each linear subspace $\Ac_n(\Mb)$ consists of linear combinations of $\II$ and products of no more than $n$ generators $\phi(f)$, $f\in C^\infty_0(\Mc)$. The product $\phi(f_1)\cdots \phi(f_n)$ belongs to $\Ac_n(\Mb)$, as does $\phi(f_{\sigma(1)})\cdots \phi(f_{\sigma(n)})$. The commutation relation $[\phi(f),\phi(g)] = iE(f,g)\II$ then implies that the product $\phi(f_1)\cdots \phi(f_n)$ and the same product with any two $f_i$'s swapped, hence also $\phi(f_{\sigma(1)})\cdots \phi(f_{\sigma(n)})$ for any permutation $\sigma$, coincide ``up to lower order terms,'' or more precisely coincide in the quotient $\Ac_n(\Mb) / \Ac_{n-1}(\Mb)$. Thus, without loss of generality, the coefficients $c_{(n)}^{i_1\cdots i_n}$ in~\eqref{a} can be taken to be, for instance, fully symmetric in their indices. So, in order to fully specify a state, it would be sufficient to specify only the fully symmetric part of each $n$-point function $\omega_n(f_1,\ldots,f_n)$.
\end{remark}
Once a state $\omega$ is given, we can implement the GNS machinery obtaining a $*$-representation $\pi_\omega : {\cal A}(\Mb) \to \Lf ({\cal D}_\omega)$ over the Hilbert space ${\cal H}_\omega$
including the dense invariant  linear subspace ${\cal D}_{\omega}$.  The {\bf smeared field operators} appear  here as the densely defined  symmetric operators:
$$\hat{\phi}_\omega (f)   \stackrel {\mbox{\scriptsize  def}} {=} \pi_\omega(\phi(f)): {\cal D}_\omega \to {\cal H}_\omega\:,\quad f\in C_0^\infty(\Mc)\:.$$
We stress that in general  $\hat{\phi}_\omega (f) $ is not self-adjoint nor essentially self-adjoint  on ${\cal D}_\omega$ (even if we are considering real smearing functions). That is why we introduce the following definition:

\begin{definition}[Regular states]\label{defreg}
A state  $\omega$ on $\Ac(\Mb)$ and its GNS representation are said to be {\bf regular} if $\hat{\phi}_\omega (f)$ is essentially self-adjoint on ${\cal D}_\omega$ for every $f\in C_0^\infty(\Mc)$. 
\end{definition}

\noindent There are some further elementary technical  properties of $\omega_2$ and $E$ that we list below. 
\begin{proposition}\label{prp:2pt}
Consider a state $\omega : {\cal A}(\Mb)\to \mathbb C$ and define  $P  \stackrel {\mbox{\scriptsize  def}} {=} \Box_\Mb + m^2 + \xi R$. The two-point function, $\omega_2$, satisfies the following facts for  $f,g \in C_0^\infty(\Mc)$:
\begin{align}
&\omega_2(Pf,g)= \omega_2(f,Pg)=0\label{KGomega}\:,\\
&\omega_2(f,g)- \omega_2(g,f) =  iE(f,g) \label{ASE}\:,\\
&Im (\omega_2(f,g)) = \frac{1}{2} E(f,g)\label{2}\:,\\
&\frac{1}{4}| E(f,g)|^2  \leq  \omega_2(f,f) \omega_2(g,g)\:.\label{3} 
\end{align}
\end{proposition}

\begin{proof}
The first identity trivially arises from $\omega_2(Pf,g)= \omega(\phi(Pf)\phi(g))=0$ and
$\omega_2(f,Pg)= \omega(\phi(f)\phi(Pg))=0$ in view of  the definition of $\phi(h)$. Next,
$$
\omega_2(f,g)- \omega_2(g,f) = \omega([\phi(f),\phi(g)]) = \omega(i E(f,g)\II)= iE(f,g)\omega(\II) = iE(f,g)\:.
$$
The third identity then  follows immediately since $E(f,g)$ is real.
Using its GNS representation and the Cauchy-Schwartz inequality we find that
$$|\omega_2(f,g)| \leq |\langle \hat{\phi}_\omega(f) \Psi_\omega | \hat{\phi}_\omega(f) \Psi_\omega \rangle|^{1/2} \: |\langle \hat{\phi}_\omega(g) \Psi_\omega | \hat{\phi}_\omega(g) \Psi_\omega \rangle|^{1/2}$$
namely
$$|\omega_2(f,g)|^2 \leq \omega_2(f,f) \omega_2(g,g)\:.$$
So that, in particular 
$$|Im (\omega_2(f,g))|^2 \leq \omega_2(f,f) \omega_2(g,g)$$
and thus, due to (\ref{2}), we end up with (\ref{3}). $\Box$
\end{proof}

\subsection{Symplectic and Poisson reformulation,  faithful representations, induced isomorphisms}
We recall for the reader the following elementary definitions.
\begin{definition}[Symplectic vector space]
A ({\bf real}) {\bf  symplectic form} over the real vector space $V$ is a {\em bilinear}, {\em  antisymmetric} map  $\tau : V \times V \to \mathbb R$. $\tau$ is said to be {\bf weakly non-degenerate} 
if $\tau(x,y)=0$ for all $x\in V$ implies $y=0$. In this case $(V,\tau)$ is said to be a ({\bf real}) {\bf symplectic vector space}.
\end{definition}

Next, we would like to define a Poisson vector space. In the finite dimensional case, it is simply a pair $(V,\Pi)$, where $V$ is a real vector space and $\Pi \in \Lambda^2 V$, which is the same as being a bilinear, antisymmetric form on the (algebraic) linear dual $V^*$. However, in our cases of interest, $V$ is infinite dimensional and $\Pi$ belongs to a larger space than $\Lambda^2 V$, that could be defined using linear duality. Constructions involving linear duality necessarily bring into play the topological structure on $V$ (or lack thereof). We will not enter topological questions in detail, so we content ourselves with a formal notion of duality, which will be sufficient for our purposes.

\begin{definition}[Poisson vector space]
Two real vector spaces $V$ and $W$, together with a bilinear pairing $\langle\cdot,\cdot\rangle\colon W\times V \to \mathbb{R}$, are in {\bf formal duality} when the bilinear pairing is non-degenerate in either argument ($\langle x, y \rangle = 0$ implies $x=0$ if it holds for all $y\in V$, and it implies $y=0$ if it holds for all $x\in W$). Given such $V$ and $W$ in formal duality, we call $(V,\Pi,W,\langle\cdot,\cdot\rangle)$ a ({\bf real}) {\bf Poisson vector space} if $\Pi\colon W\times W \to \mathbb{R}$ is a bilinear, antisymmetric map, called the {\bf Poisson bivector}. $\Pi$ is said to be {\bf weakly non-degenerate} if $\Pi(x,y) = 0$ for all $x\in W$ implies $y = 0$.
\end{definition}

At this level, there are only subtle differences between symplectic and Poisson vector spaces. In fact, the two structures have often been confounded in the literature on QFT on curved spacetime~\cite{BSZ,FP,FV,DappiaggiSiemens,hack-schenkel,GW1}. The differences become more pronounced when we consider symplectic differential forms and Poisson bivector fields on manifolds locally modeled on the vector space $V$. A form is a section of an antisymmetric power of the cotangent bundle, while a bivector field is a section of an antisymmetric power of the tangent bundle. In infinite dimensional settings, one has to choose a precise notion of tangent and cotangent bundle, among several inequivalent possibilities. This ambiguity is reflected in our need to introduce formal duality for the definition of a Poisson vector space.

The above abstract definitions are concretely realized in the Proposition that we present below. Let us use the formula on the right-hand side of~\eqref{simpsol} to define a bilinear, antisymmetric map $\tau\colon \Sol \times \Sol \to \mathbb{R}$ by
\begin{equation}\label{eq:taudef}
	\tau(\psi,\xi) = \int_{\Sigma} (\psi \nabla_{\nb} \xi - \xi \nabla_{\nb} \psi) \: d\Sigma .
\end{equation}
Defining the space of equivalence classes ${\cal E} = C^\infty_0(\Mc) / (\Box_\Mb + m^2 + \xi R)C^\infty_0(\Mc)$ and recalling Equations~\eqref{0} and~\eqref{1}, the advanced-minus-retarded fundamental solution defines a bilinear, antisymmetric map $E\colon {\cal E} \times {\cal E} \to \mathbb{R}$ by
\begin{equation}
	E([f],[g]) = E(f,g).
\end{equation}
Furthermore, there is a well-defined bilinear pairing $\langle\cdot, \cdot\rangle \colon {\cal E} \times \Sol \to \mathbb{R}$ given by
\begin{equation}
	\langle [f], \psi \rangle = \int_{\Mc} f \psi \: \dvol_{\Mb} .
\end{equation}

Given the above definitions for the Klein-Gordon real scalar field, we have the following.
\begin{proposition}\label{prp:symp-pois}
The spaces $\Sol$ and ${\cal E}$ are in formal duality, with respect to the pairing $\langle\cdot, \cdot\rangle$. The pair $(\Sol,\tau)$ is a symplectic vector space, while $(\Sol,E,{\cal E},\langle\cdot,\cdot\rangle)$ is a Poisson vector space. Moreover, the bilinear forms $\tau$ and $E$ respectively induce linear maps
\begin{equation}
	\tau\colon \Sol \to {\cal E}
		\quad \text{and} \quad
	E \colon {\cal E} \to \Sol
\end{equation}
that are bijective, mutually inverse and such that $\tau(\psi,\xi) = \langle \tau \psi, \xi \rangle$ and $E([f],[g]) = \langle [f], E[g] \rangle$.
\end{proposition}
\begin{proof}
The content of this proposition is discussed in detail in~\cite[Sec.5]{kh-conal} or~\cite[Sec.3]{kh-covar}, though a basic version can be found already in~\cite[Sec.3.2]{wald94}. We only indicate a few salient points. The non-degeneracy of the pairing $\langle\cdot, \cdot\rangle$ implies, provided there exist linear operators $\tau$ and $E$ such that $E([f],[g]) = \langle [f], E[g] \rangle$ and $\tau(\psi,\xi) = \langle \tau\psi, \xi \rangle$, that they are unique. These operators can be exhibited rather concretely. The linear map $E\colon {\cal E}\to \Sol$ is already defined by Remark~\ref{remark2}{\bf (1)}, in view of Equation~\eqref{0}. The definition of $\tau$ in~\eqref{eq:taudef} is independent of the choice of Cauchy surface $\Sigma \subset \Mc$. Let $\Sigma^+,\Sigma^- \subset \Mc$ be Cauchy surfaces, respectively to the future and to the past of the Cauchy surface $\Sigma$, and let $\chi\in C^\infty(\Mc)$ be such that $\chi \equiv 0$ to the past of $\Sigma^-$ and $\chi
  \equiv 1$ to the future of $\Sigma^+$. Then, we have the identity $\tau\psi = [(\Box_\Mb + m^2 + \xi R)(\chi\psi)]$. Finally, the non-degeneracy or bijectivity of $\langle\cdot, \cdot\rangle$, $\tau$ and $E$, considered either as bilinear forms or linear operators, strongly rely on the hyperbolic character and well-posedness of the Klein-Gordon equation. $\Box$
\end{proof}

Given the isomorphism between ${\cal E}$ and $\Sol$ and the close relationship between $E$ and $\tau$, it is not surprising these two spaces and bilinear forms have often been used interchangeably in the context of the QFT of the Klein-Gordon real scalar field. However, this interchangeability may fail for more complicated field theories, as we remark next. This is another reason why it is important to keep track of the difference between the respective symplectic and Poisson vector spaces, $(\Sol,\tau)$ and $(\Sol,E)$!
\begin{remark}\label{rem:gauge}
References~\cite[Sec.5]{kh-conal} and~\cite[Sec.3]{kh-covar} also address in detail the question of whether similar statements hold for gauge theories (electrodynamics, linearized gravity, etc.)\ or for theories with constraints (massive vector field, etc.). Related questions were also studied in~\cite{hack-schenkel}. The answer turns out to be rather subtle. The bilinear forms $\tau$ and $E$ can essentially always be defined. A reasonable choice of the spaces ${\cal E}$ and $\Sol$ also make sure that the linear maps $\tau\colon \Sol\to {\cal E}$ and $E\colon {\cal E} \to \Sol$ are also well-defined and are mutually inverse. However, the pairing $\langle\cdot, \cdot\rangle$ appearing in the formulas $\tau(\psi,\xi) = \langle \tau\psi,\xi \rangle$ and $E([f],[g]) = \langle [f], E[g] \rangle$, need no longer be non-degenerate. Hence, the bilinear forms $\tau$ and $E$ may be degenerate themselves. The conditions under which these degeneracies do or do not occur subtly depend on 
 the geometry of the gauge transformations and the constraints of the theory.
\end{remark}

We now turn to applying the above symplectic and Poisson structures to the study of the properties of the CCR algebra $\Ac(\Mb)$ of a Klein-Gordon field.

\begin{definition}[CCR algebra of a Poisson vector space]\label{def:ccr-pois}
Let $(V,\Pi,W,\langle\cdot,\cdot\rangle)$ be a Poisson vector space, defined with respect to a formal duality between $V$ and another space $W$. The corresponding {\bf CCR algebra} $\Ac(V,\Pi,W,\langle\cdot,\cdot\rangle)$ is defined as the unital $*$-algebra presented by the generators $A(x)$, $x\in W$, subject to the relations $A(ax+by) - aA(x) - bA(y) = 0$, $A(x)^*-A(x)=0$ and $[A(x),A(y)] - i\Pi(x,y)\II = 0$, for any $a,b\in \mathbb{R}$ and $x,y\in W$.
\end{definition}
This generic definition allows us to state and prove the following useful result.
\begin{proposition}[Simplicity and faithfulness]\label{prp:simpl}
Given that the spaces $V$ and $W$ are in formal duality and the Poisson bivector of the Poisson vector space $(V,\Pi,W,\langle\cdot,\cdot\rangle)$ is weakly non-degenerate (as a bilinear form on $W$), the corresponding CCR algebra $\Ac(V,\Pi,W,\langle\cdot,\cdot\rangle)$ is simple. Further, it admits only zero or faithful representations.
\end{proposition}
Before giving the proof, we note its main consequence. It is not hard to see that the definition of the CCR algebra $\Ac(\Mb)$, as given in Definition~\ref{def:ccr}, coincides with the alternative definition $\Ac(\Mb)   \stackrel {\mbox{\scriptsize  def}} {=} \Ac(\Sol,E)$, using the notation of Proposition~\ref{prp:symp-pois} and referring to the formal duality between $\Sol$ and ${\cal E}$. The explicit homomorphism acts on the generators as $\phi(f) \mapsto A([f])$. Thus, given Proposition~\ref{prp:symp-pois}, we have the immediate
\begin{corollary}
The CCR algebra $\Ac(\Mb)$ of a real scalar quantum field is simple and admits only either zero or faithful representations.
\end{corollary}
\begin{remark}
The result established in the Corollary above is {\em not} valid form more complicated QFTs like {\em electromagnetism}~\cite{sdh} and {\em linearized gravity}~\cite{fewster-hunt}. The physical reason is the appearance of the {\em gauge invariance}. Mathematically it is related to the fact that the Poisson bivector corresponding to our $E$ is {\em degenerate} on the space ${\cal E}$ of compactly supported observables, as discussed in~\cite[Sec.5]{kh-conal} and~\cite[Sec.3]{kh-covar}.
\end{remark}

The proof of Proposition~\ref{prp:simpl} makes use of the following two lemmas.
\begin{lemma}\label{lem:degen}
Let $\Pi$ be a bilinear form (we need not even assume it to be antisymmetric) on a vector space $W$. Further, let $v_i \in W$, $i=1,\ldots,N$, be a set of linearly independent vectors and $c^{i_1\cdots i_k}$ a collection of scalars, not all zero, with each index running through $i_j = 1,\ldots,N$. Then, if
\begin{equation}
	\sum_{i_1,\ldots,i_k} c^{i_1\cdots i_k}
		\Pi(v_{i_1}, u_1) \cdots \Pi(v_{i_k}, u_k) = 0
\end{equation}
for each set of vectors $u_i\in W$, $i=1,\ldots,k$. Then there exists a non-zero vector $w\in W$ such that $\Pi(w,u) = 0$ for any $u\in W$.
\end{lemma}
\begin{proof}
The proof is by induction on $k$. Let $k=1$, then the right-hand side of the equation in the hypothesis is $\Pi(w',u_1)$, where
\begin{equation}
	w' = \sum_{i} c^{i} v_{i} .
\end{equation}
Since not all $c^{i}$ are zero and the $v_i$, $i=1,\ldots,N$ are linearly independent, we have $w' \ne 0$. We can then set $w = w'$ and we are done, since $u_1$ can be arbitrary.

Now, assume that the case $k-1$ has already been established. Note that
we can write the right-hand side of the above equation as $\Pi(w', u_k)$,
where
\begin{equation}
	w' = \sum_{i_1,\ldots,i_k} c^{i_1\cdots i_k}
			\Pi(v_{i_1},u_1) \cdots \Pi(v_{i_{k-1}},u_{k-1}) v_{i_k} .
\end{equation}
If $w' \ne 0$ for some choice of $u_i\in W$, $i=1,\ldots,k-1$, then we can set $w = w'$ and we are done, since $u_k$ can be arbitrary.

Consider the case when $w' = 0$ for all $u_i \in W$, $i=1,\ldots,k-1$. Then, choose $j_k$ such that $c^{i_1\cdots i_{k-1} j_k}$ are not all zero. Since, by linear independence, the coefficients of the $v_{i_k}$ in $w'$ must vanish independently, we have
\begin{equation}
	\sum_{i_1,\ldots,i_{k-1}} c^{i_1\cdots i_{k-1} j_k}
		\Pi(v_{i_1},u_1) \cdots \Pi(v_{i_{k-1}},u_{k-1}) = 0
\end{equation}
for all $u_i\in W$, $i=1,\ldots,k-1$. In other words, by the inductive
hypothesis, the last equality implies the existence of the desired non-zero
$w\in W$, which concludes the proof. $\Box$
\end{proof}

A bilinear form $\Pi$ on $W$ naturally defines a bilinear form $\Pi^{\otimes k}$ on the $k$-fold tensor product $W^{\otimes k}$. Let $S\colon W^{\otimes k} \to W^{\otimes k}$ denote the (idempotent) full symmetrization operator and denote its image, the space of fully symmetric $k$-tensors, by $S^k W   \stackrel {\mbox{\scriptsize  def}} {=} S(W^{\otimes k})$. Of course, $\Pi^{\otimes k}$ also restricts to $S^k W$. If $\Pi$ is antisymmetric, then $\Pi^{\otimes k}$ is symmetric when $k$ is even and antisymmetric when $k$ is odd.
\begin{lemma}\label{lem:tens-nondegen}
If the antisymmetric bilinear form $\Pi$ is weakly non-degenerate on $W$, then the antisymmetric bilinear form $\Pi^{\otimes k}$ is weakly non-degenerate on $S^kW$.
\end{lemma}
\begin{proof}
Assume the contrary, that $\Pi^{\otimes k}$ is degenerate. By its (anti-)symmetry, we need only consider the degeneracy in its first argument. That is, there
exists a vector $v = \sum_{i_1,\ldots i_k} d^{i_1\cdots i_k} v_{i_1} \otimes
\cdots \otimes v_{i_k}$, where $v_i\in W$, $i=1,\ldots,N$, constitute a
linearly independent set and the $d^{i_1\cdots i_k}$ coefficients are not all
zero and are symmetric under index interchange, such that
\begin{equation}
	\Pi^{\otimes k}(v, S(u_1\otimes \cdots \otimes u_k)) = 0 .
\end{equation}
for any $u_i\in W$, $i=1,\ldots,k$.  But then, the above equality is precisely of the form of the hypothesis of Lemma~\ref{lem:degen}, with
\begin{equation}
	c^{i_1\cdots i_k}
		= k! \, d^{i_1 \cdots i_k} ,
\end{equation}
due to the symmetry of $d^{i_1\cdots i_k}$ under index interchanges.  Therefore, by Lemma~\ref{lem:degen}, there must exist a $w\in W$ such that $\Pi(w,u) = 0$ for all $u\in W$, which contradicts the weak non-degeneracy of $\Pi$ on $W$. Therefore, $\Pi^{\otimes k}$ cannot be degenerate on $S^kW$, and hence is weakly non-degenerate. $\Box$
\end{proof}

\begin{proof}[of Proposition~\ref{prp:simpl}]
Suppose that $\Ac(V,\Pi,W,\langle\cdot,\cdot\rangle)$ is not simple, and so has a non-trivial two-sided ideal ${\cal I}$. If we can deduce that $\II\in {\cal I}$, then any non-trivial two-sided ideal must be all of $\Ac(V,\Pi,W,\langle\cdot,\cdot\rangle)$, implying that the algebra is simple.

Take any non-zero element $a\in {\cal I}$ and recall the idea behind Equation~\eqref{a}. That is, there exists integers $k,N\ge 0$, linearly independent elements $v_i\in W$, $i=1,\ldots,N$, and complex coefficients $c_{(l)}^{i_1\cdots i_l}$, $i_j = 1,\ldots, N$ and $l=0,\ldots,k$, such that
\begin{multline}
	a = c_{(0)}\II + \sum_{i_1} c_{(1)}^{i_1} A(v_{i_1})
		+ \sum_{i_1,i_2} c_{(2)}^{i_1 i_2} A(v_{i_1}) A(v_{i_2}) \\ {}
		+ \cdots
		+ \sum_{i_1,\ldots,i_k} c_{(k)}^{i_1\ldots i_k} A(v_{i_1}) \cdots A(v_{i_k}) ,
\end{multline}
where not all of the components of $c_{(k)}^{i_1\ldots i_k}$ are zero. If $k=0$, the $\II\in {\cal I}$ and we are done. If $k>0$, note that ${\cal I}$ also contains the iterated commutator $[\cdots[a,A(u_1)],\ldots,A(u_k)]$, for any $u_i\in W$, $i=1,\cdots,k$. A straight forward calculation shows that, up to (non-zero) numerical factors, the iterated commutator is equal to
\begin{equation}\label{eq:Pi-tens}
	\Pi^{\otimes k}\left( \sum_{i_1,\ldots,i_k} c_{(k)}^{i_1,\ldots,i_k} S(v_{i_1} \otimes \cdots \otimes v_{i_k}), S(u_1 \otimes \cdots \otimes u_k) \right) \II .
\end{equation}
By Lemma~\ref{lem:tens-nondegen}, since $\Pi$ is weakly non-degenerate on $W$, $\Pi^{\otimes k}$ is weakly non-degenerate on $S^kW$. Since elements of the form $S(u_1\otimes \cdots \otimes u_k)$ generate $S^kW$, there must exist at least one element of $S^kW$ of that form such that the coefficient in front of $\II$ in~\eqref{eq:Pi-tens} is non-zero. Therefore, $\II\in {\cal I}$ and we are done. $\Box$
\end{proof}

Automorphisms of the CCR algebra $\Ac(\Mb)$ are important because the composition of a state with an automorphism gives a way to define more states, once at least one is known. The identity $\Ac(\Mb) \cong \Ac(\Sol,E)$ allows us to construct lots of automorphisms of $\Ac(\Mb)$, induced by transformations of $\Sol$ or ${\cal E}$ that, respectively, leave $\tau$ or $E$ invariant.
\begin{proposition}[Induced homomorphism]\label{propind}
Let $\Ac(V,\Pi,W,\langle\cdot,\cdot\rangle)$ be as in Definition~\ref{def:ccr-pois} and
let $\sigma\colon W \to W$ be a linear map such that
\begin{equation}
	\Pi(\sigma x, \sigma y) = \Pi(x,y) ~~
	\text{(resp.~$\Pi(\sigma x, \sigma y) = -\Pi(x,y)$)} ,
\end{equation}
for all $f,g \in W$. Then, there exists a homomorphism (resp.~anti-linear homomorphism) of unital $*$-algebras, $\alpha^{(\sigma)}\colon \Ac(V,\Pi,W,\langle\cdot,\cdot\rangle) \to \Ac(V,\Pi,W,\langle\cdot,\cdot\rangle)$ uniquely defined by its values
\begin{equation}
	\alpha^{(\sigma)}(A(x))   \stackrel {\mbox{\scriptsize  def}} {=} A(\sigma x) ,
\end{equation}
for each $x\in W$, on the generators of $\Ac(V,\Pi,W,\langle\cdot,\cdot\rangle)$. Also, if $\sigma$ is bijective, then $\alpha^{(\sigma)}$ is an automorphism.
\end{proposition}

\begin{remark}
In view of Proposition~\ref{prp:symp-pois} and the isomorphism ${\cal E} \cong \Sol$, in the case of the CCR algebra $\Ac(\Mb) \cong \Ac(\Sol,E)$ of a real scalar quantum field, the linear endomorphisms of ${\cal E}$ that preserve the Poisson bivector $E$ can be equivalently specified by linear endomorphisms of $\Sol$ that preserve the symplectic form $\tau$.
\end{remark}

\begin{proof}
Recall the definition of an algebra presented by generators and relations by its universal property, as discussed in Section~\ref{sec:genrel}, as well as such a presentation of the algebra $\Ac(V,\Pi,W,\langle\cdot,\cdot\rangle)$ given in Definition~\ref{def:ccr-pois}.

Let us denote by $\Ac(W)$ the algebra freely generated by the elements of the vector space $W$. Following our notation, the map embedding the generators in this algebra can be denoted as $A\colon W \to \Ac(W)$. The composition $A\circ \sigma$ is another such map. Therefore, by the universal property, there exists a unique homomorphism $\beta \colon \Ac(W) \to \Ac(W)$ such that $\beta(A(x)) = A(\sigma x)$, for all $x\in W$, and $\beta(\II) = \II$.

We now need to check whether $\beta$ leaves invariant the kernel of the projection $\Ac(W) \to \Ac(V,\Pi,W,\langle\cdot,\cdot\rangle)$. This kernel is the two-sided ideal generated by the relations $A(ax+by)-aA(x)-bA(y)=0$, $A(x)^*-A(x)=0$ and $[A(x),A(y)] - i\Pi(x,y)\II = 0$, for any $a,b\in \mathbb{R}$ and $x,y\in W$, so it is sufficient to check the invariance of these relations. The first two are obviously invariant. The last commutator identity is invariant upon invoking the hypothesis that $\sigma$ preserves $\Pi$, up to sign. We deal with the two cases separately.

In the case when $\sigma$ preserves $\Pi$, we have
\begin{equation}
	[A(\sigma x),A(\sigma y)] - i\Pi(x,y)\II =
	[A(\sigma x),A(\sigma y)] - i\Pi(\sigma x,\sigma y)\II .
\end{equation}
Hence, the homomorphism $\beta$ induces a uniquely defined homomorphism on the quotiented algebra, which we call $\alpha^{(\sigma)} \colon \Ac(V,\Pi,W,\langle\cdot,\cdot\rangle) \to \Ac(V,\Pi,W,\langle\cdot,\cdot\rangle)$, which given by $\alpha^{(\sigma)}([a]) = [\beta a]$, and which has all the desired properties.

In the case when $\sigma$ changes the sign of $\Pi$, we need to change perspective slightly. Recall that we defined $\Ac(\Mb)$ as a complex algebra, which then automatically has the structure of a real algebra. Equivalently, we could have also defined it directly as a real algebra, by throwing in an extra generator $i$, satisfying the relations $i^2 = -\II$, $[i,\II] = [i,A(x)] = 0$ and $i^* = -i$. If the homomorphism $\beta$ is extended to this generator as $\beta(i) = -i$, then it preserves the new relations that need to be satisfied by $i$ and also the commutator identity, since
\begin{equation}
	[A(\sigma x),A(\sigma y)] - (-i)\Pi(x,y)\II =
	[A(\sigma x),A(\sigma y)] - i\Pi(\sigma x,\sigma y)\II .
\end{equation}
Hence, the real algebra homomorphism $\beta$ induces a uniquely defined homomorphism on the quotiented algebra, which also happens to be an anti-linear homomorphism in the sense of complex algebras, which we call $\alpha^{(\sigma)} \colon \Ac(V,\Pi,W,\langle\cdot,\cdot\rangle) \to \Ac(V,\Pi,W,\langle\cdot,\cdot\rangle)$, and which has all the desired properties.

Finally, when $\sigma$ is a bijection, we can use the universal property of $\Ac(V,\Pi,W,\langle\cdot,\cdot\rangle)$, as was done in Section~\ref{sec:genrel}, to show that $\alpha^{(\sigma^{-1})} = (\alpha^{(\sigma)})^{-1}$. Therefore, $\alpha^{(\sigma)}$ is an isomorphism and hence an automorphism of the algebra. $\Box$
\end{proof}

We end this section by noting that there is another structure that is induced on the space ${\cal E} \cong C^\infty_0(\Mc)/Ker(E)$ in the presence of a state $\omega$ on $\Ac(\Mb)$, namely the symmetrized part of the $2$-point function
\begin{equation}
	\omega^S_2(f,g)   \stackrel {\mbox{\scriptsize  def}} {=} \frac{1}{2} (\omega_2(f,g) + \omega_2(g,f)) ,
\end{equation}
with $f,g\in C^\infty_0(\Mc)$ and $\omega_2(f,g)   \stackrel {\mbox{\scriptsize  def}} {=} \omega(\phi(f) \phi(g))$. By hermiticity, the symmetrized $2$-point function is always real and non-negative, which was essentially already noted in~\eqref{ASE} and~\eqref{2} of Proposition~\ref{prp:2pt}. Also, by Proposition~\ref{iff}, $\phi(f)$ depends only on the equivalence class $[f] \in {\cal E}$. Hence, $\omega^S_2 \colon {\cal E} \times {\cal E} \to \mathbb{R}$ defines a real symmetric bilinear form. Finally, the inequality~\eqref{3} from Proposition~\ref{prp:2pt}, which we can rewrite as
\begin{equation}\label{3'}
	\frac{1}{4} |E([f],[g])|^2 \le \omega^S_2([f],[f]) \omega^S_2([g],[g]) ,
\end{equation}
shows that $\omega^S_2$ is non-degenerate on ${\cal E}$, since it majorizes $E$, which is already known to be non-degenerate by Proposition~\ref{prp:symp-pois}. Thus, a state $\omega$ on $\Ac(\Mb)$ induces a positive scalar product $\omega^S_2$ on ${\cal E}$ (and also on $\Sol$ by the isomorphism of Proposition~\ref{prp:symp-pois}). We will use this scalar product structure and the inequality~\eqref{3'} to construct quasifree states in the next section.

\subsection{Quasifree states, also known as  Gaussian states} 
There is a plethora of states on ${\cal A}(\Mb)$, the first class we consider is that of the {\em quasifree} or {\em Gaussian states}. They mimic the Fock representation of Minkowski vacuum and they are completely determined from the two-point function by means of a prescription generalizing  the well known Wick procedure which also guarantees  essential self-adjointness of the field operators $\hat{\phi}_\omega$ since they are regular (Definition \ref{defreg}).

\begin{definition}[Quasifree states]\label{defqfs}
An algebraic state $\omega : {\cal A}(\Mb) \to \mathbb C$ is said to be {\bf quasifree} or {\bf Gaussian} if its $n$-point functions agree with the so-called   {\bf Wick procedure}, in other words they  satisfy the following pair of requirements for all choices of $f_k \in C_0^\infty(\Mc)$,\\

{\bf (a)} $\omega_{n}(f_1,\ldots, f_{n}) = 0$,  \quad for $n=1,3,5, \ldots$\\

{\bf (b)} $\omega_{n}(f_1, \ldots, f_{n}) = \sum_{\mbox{partitions}} \omega_2(x_{i_1}, x_{i_2})\cdots  \omega_2(x_{i_{n-1}}, x_{i_n})$, \quad for $n=2,4,6, \ldots$\\

\noindent For the case of $n$ even,  the  {\em partitions}  refers to the class of all possible decomposition of  set $\{1,2,\ldots, n\}$ into $n/2$ pairwise disjoint  subsets of  $2$ elements $$\{i_1,i_2\}, \{i_3,i_4\}\ldots \{i_{n-1},i_n\}$$ with $i_{2k-1}< i_{2k}$ for $k=1,2,\ldots, n/2$.
\end{definition}

\noindent We will  prove in the next section   that quasifree states exist in a generic curved spacetime for a massive scalar field and $\xi=0$.  Instead  we intend  to clarify here  the structure of the GNS representation of quasifree states, proving that it is a Fock representation. 
The characterization theorem  relies on the following intermediate result.

\begin{proposition}[One-particle structure]\label{propSK}
Consider the symplectic vector space $(\Sol,\tau)$, cf.~Proposition~\ref{prp:symp-pois}.\\
{\bf (a)} If a real scalar product $\mu : \Sol \times \Sol \to \mathbb R$ satisfies 
\begin{eqnarray}
\frac{1}{4} |\tau(x,y)|^2 \leq \mu(x,x)\mu(y,y) \quad \forall x,y \in \Sol \label{4}
\end{eqnarray}
then there exists a pair $(K,H)$, called {\bf one-particle structure} associated to $(\Sol, \tau, \mu)$  where  $H$ is a complex Hilbert space and $K\colon \Sol \to H$ is a map satisfying\\

(i)  $K$  is $\mathbb R$ linear and  $K(\Sol) + i K(\Sol)$ is dense in $H$ (though $K(\Sol)$, as a real subspace of $H$, need not be dense by itself),\\

(ii)  $\langle Kx| Ky \rangle = \mu(x,y) + \frac{i}{2}\tau(x,y)$ for all $x,y \in \Sol$.\\

\noindent {\bf (b)} If $(K',H')$ satisfies (i) in (a)  and $\tau(x,y)= 2Im (\langle K' x| K' y\rangle_{H'})$, then the scalar product $\mu$ on $\Sol$ obtained from  (ii) in (a) also satisfies (\ref{4}).

\noindent {\bf (c)} A  pair $(H', K')$ satisfies (i) and (ii) in (a) if and only if there is an isometric surjective operator
$V : H\to H'$ with $VK=K'$. 
\end{proposition}

\begin{proof}
 Barring different conventions on signs the proof is given  in Proposition 3.1 in \cite{kw}. $\Box$
\end{proof}

\noindent A characterization  theorem for quasifree states can now be proved using the lemma
above with the following theorem that can be obtained  by Lemma A.2, Proposition 3.1
and a comment on p.77 in  \cite{kw} (again modulo different conventions on signs) where the approach based on  Weyl $C^*$-algebras is pursued. For quasifree states the approaches relying on CCR $*$-algebras and  Weyl $C^*$-algebras are technically equivalent.
The fact that, on a Fock space, an operator as the one in (\ref{phi}) is {\em essentially self-adjoint} in the indicated domain \cite{BR} is well know and can be proved directly, for instance, using analytic vectors.

\begin{theorem}[Characterization of quasifree states]\label{teoE}
Consider the $*$-algebra ${\cal A}(\Mb)$ associated to a real scalar KG field.  Suppose that
$\mu $ is a real scalar product on $\Sol$ which verifies~\eqref{4}.  The following hold.\\
\noindent {\bf (a)} There exists a quasifree state $\omega$ on  ${\cal A}(\Mb)$  such that
\begin{equation}
	\omega_{2}(f,g) = \mu(Ef,Eg) + \frac{i}{2} E(f,g) \:, \quad \forall f,g \in C_0^\infty(\Mc)\:.
\end{equation}

\noindent {\bf (b)}  The GNS triple $({\cal H}_\omega, {\cal D}_{\omega}, \pi_{\omega}, \Psi_\omega)$ consists of the following:\\

(i)  ${\cal H}_\omega$ is the bosonic (symmetrized) Fock space with the one-particle subspace being $H$, of the one structure particle $(K,H)$ in Proposition \ref{propSK};\\

(ii)   $\Psi_\omega$ is the vacuum vector of the Fock space;\\

(iii) ${\cal D}_\omega$ is the dense  subspace of the finite complex linear combinations of  $\Psi_\omega$ and all of the vectors
\begin{equation}
	a^\dagger(\psi_1) \cdots a^\dagger(\psi_n) \Psi_\omega \quad \mbox{for $n=1,2,\ldots$ and $\psi_k \in \Sol$}
\end{equation}
where $a^*(\psi)$ is the standard creation operator\footnote{It holds that $[a(\psi), a^\dagger(\xi)] = \langle K\psi |K\xi \rangle$, $[a(\psi),a(\xi)] = 0 = [a^\dagger(\psi), a^\dagger(\xi)]$ if $\xi,\psi \in \Sol$, and $a(\xi)$, $a(\psi)$ are defined on ${\cal D}_\omega$, with $a^\dagger(\psi)= a(\psi)^\dagger|_{{\cal D}_\omega}$.}  corresponding to the solution $\psi \in \Sol$.\\

(iv)  $\pi_\omega$ is completely determined by $a$ and $a^\dagger$, with $a(\psi)$ being the annihilation operator corresponding to the solution $\psi \in \Sol$,
\begin{equation}
	\hat{\phi}_\omega(f) = \pi_\omega(\phi(f))  = a(Ef) + a^\dagger(Ef) \quad \forall f \in C_0^\infty(\Mc)\:,\label{phi}
\end{equation}
and, in particular,  $\omega$ is regular, meaning that $\hat{\phi}_\omega(f)$ is essentially self-adjoint on ${\cal D}_\omega$. \\

\noindent {\bf (d)} The quasifree state $\omega$ determined by $\mu$ is pure if and only if the image $K(\Sol)$ is dense in the one-particle subspace $H$, thus strengthening (i) of Proposition~\ref{propSK}. This condition  is equivalent to:
\begin{equation}
	\mu( \psi,\psi) = \frac{1}{4}\sup_{\xi\neq 0}  \frac{|\tau(\psi,\xi)|^2}{\mu( \xi,\xi)}\:.\label{purecond}
\end{equation}
\end{theorem}

\begin{remark}$\null$\\
{\bf (1)}  $K$ is always injective because of (ii) in Proposition \ref{propSK}, since $\tau$ is non-degenerate.\\
{\bf (2)} The requirement~\eqref{4} is equivalent to saying that there is a bounded operator $J$ everywhere defined in the real Hilbert space obtained by taking the completion  ${\cal R}$ of $\Sol$ with respect to the 
real scalar product induced by $\mu$, such that $\frac{1}{2}\tau(\psi,\xi)= \omega_2(\psi,J\xi)$, for $\psi,\xi \in \Sol$, and $||J||\leq 1$.
It also holds that $J^\dagger=-J$. It is not so difficult to prove that the corresponding state $\omega$, as defined above, is pure if and only if $JJ=-I$, that is $J$ is anti unitary.  In this case  $({\cal R}, \mu, \frac{1}{2}\tau, J)$ defines an {\bf almost K\"ahler structure} on ${\cal R}$.
\end{remark}

\subsection{Existence of quasifree states in globally hyperbolic spacetimes}\label{existenceqfs}
In four-dimensional  Minkowski spacetime $\Mb   \stackrel {\mbox{\scriptsize  def}} {=} \mathbb M$, a distinguished real scalar product  $\mu$ on $\Sol \cong C_0^\infty(\mathbb M)/Ker(E)$\footnote{Recall that this isomorphism was established in Proposition~\ref{prp:symp-pois}, based on the well-posedness properties of the Klein-Gordon equation. From now one, we will be making use of this isomorphism implicitly.}  can easily be defined as follows in a Minkowski reference frame with coordinates $(t, \vec{x}) \in \mathbb R \times \mathbb R^3$. Consider $f\in C_0^\infty(\mathbb M)$ and the associated solution of KG equation $\psi_f   \stackrel {\mbox{\scriptsize  def}} {=} Ef$:
\begin{eqnarray}
\psi_f(t, \vec{x}) = \int_{\mathbb R^3}  \frac{\phi_f({\vec k}) e^{i\vec{x}\cdot\vec{k} -i t E(\vec{k}) }  + 
 \overline{\phi_f({\vec k})} e^{-(i\vec{x}\cdot\vec{k} - i t E(\vec{k})) }}{(2\pi)^{3/2}\sqrt{2E(\vec{k})}}\:
d\vec{k}
\end{eqnarray}
where $E(\vec{k})  \stackrel {\mbox{\scriptsize  def}} {=} \sqrt{\vec{k}^2 + m^2}$ (we assume here $m>0$) and $\phi_f \in {\cal S}(\mathbb R^3)$ (the Schwartz test function space)
is obtained by the smooth compactly supported Cauchy data of $\psi_f$ on the Cauchy surface defined by $t=0$.
If defining
\begin{eqnarray}\label{muM}\mu_{\mathbb M}([f],[f'])   \stackrel {\mbox{\scriptsize  def}} {=} Re \int_{\mathbb R^3}  \overline{\phi_f(\vec{k})} \phi_{f'}(\vec{k})d\vec{k}\end{eqnarray}
we obtain a well defined real scalar product on  $\Sol  \cong {\cal E}$ which satisfies (\ref{3'}) as can be proved by direct inspection with elementary computations. The arising quasifree state $\omega_{\mathbb M}$ is nothing but the {\bf Minkowski vacuum} and we find the standard QFT free theory for a real scalar field in  Minkowski spacetime.  The integral kernel of $\omega_{\mathbb M}$ in this case is a proper distribution of ${\cal D}'(\mathbb R^4\times \mathbb R^4)$  and reads
\begin{eqnarray}\omega_{{\mathbb M}2}(x,y) = w\mbox{-}\lim_{\epsilon \to 0^+}  \frac{m^2}{(2\pi)^2}
\frac{K_1\left( m \sqrt{(|\vec{x}-\vec{y}|^2 - (t_x-t_y- i\epsilon)^2)}\right) }{ m \sqrt{|\vec{x}-\vec{y}|^2 - (t_x-t_y- i\epsilon)^2}} \label{W}\end{eqnarray}
where the {\em weak} limit is understood in the standard distributional sense and the branch cut in the 
complex plane to uniquely define the analytic functions appearing in (\ref{W}) is assumed to stay along the negative real axis. 
Another equivalent expression for $\omega_{\mathbb M2}$ is given in terms of Fourier transformation of distributions,
\begin{eqnarray}
\omega_{{\mathbb M}2}(x,y) = \frac{1}{(2\pi)^3} \int_{\mathbb R^4} e^{-ip(x-y)} \theta(p^0) \delta(p^2 + m^2) d^4p \label{fourieromega}\:.
\end{eqnarray}
where $px = p^0x^0 - \sum_{j=1}^3 p^jx^j$ is the Minkowski scalar product.
The above formula is convenient for showing the following important property of $\omega_{{\mathbb M}2}$.
\begin{proposition}\label{prp:mink-smooth}
If $f\in C^\infty_0(\Mc)$, then $\omega_{{\mathbb M}2}(x,f)$ and $\omega_{{\mathbb M}2}(f,y)$ are smooth.
\end{proposition}
\begin{proof}
Let $\hat{f}(p) = \int_{\mathbb{R}^4} e^{i p y} f(y) \, d^4y$. Since $f\in C^\infty_0(\Mc)$, $\hat{f}$ must be a Schwartz function. Then, since $\omega_{{\mathbb M}2}(f,y) = \overline{\omega_{{\mathbb M}2}(y,f)}$, it is enough to consider
\begin{align*}
	\omega_{{\mathbb M}2}(x,f)
	&= \frac{1}{(2\pi)^3}
		\int_{\mathbb{R}^4} d^4 p \, e^{-i p x} \theta(p^0) \delta(p^2 + m^2)
		 \int_{\mathbb{R}^4} d^4 y \, f(y) e^{i p y} \\
	&= \frac{1}{(2\pi)^3} \int_{\mathbb{R}^3} d\mathbf{k} \,
		e^{-i p_{\mathbf{k}} x} \frac{\hat{f}(p_{\mathbf{k}})}{\sqrt{\mathbf{k}^2 + m^2}} \: ,
\end{align*}
where $p_{\mathbf{k}} = \left(\sqrt{\mathbf{k}^2 + m^2}, \mathbf{k}\right)$. Since $\hat{f}$ is Schwartz, so is the above integrand. It is then easy to see from this integral representation that $\omega_{{\mathbb M}2}(x,f)$ is smooth. $\Box$
\end{proof}
In view of the definition of quasifree state Definition \ref{defqfs}, all the $n$-point functions of $\omega_{\mathbb M}$ are distributions of ${\cal D}'((\mathbb R^4)^n)$.
It turns out that  the associated one-particle structure $(H_{\mathbb M}, K_{\mathbb M})$ is 
$$H_{\mathbb M} = L^2(\mathbb R^3, d\vec{k})\:, \quad K_{\mathbb M} : \Sol \ni \psi_f \mapsto \phi_f  \in L^2(\mathbb R^3, d\vec{k}))$$
The condition in part (d) of Theorem~\ref{teoE} is true and thus Minkowski vacuum is pure. In spite of the Poincar\'e non-invariant approach, the pictured procedure leads to a Poincar\'e invariant structure as we shall see later.

More generally, a natural {\em pure quasifree} state $\omega_\zeta$ exists as soon as the globally hyperbolic spacetime admits a time-like Killing field $\zeta$, i.e., in {\em stationary} spacetimes, provided $m>0$, $\xi=0$ and  it when holds $g(\zeta,\zeta) \geq c>0$ uniformly on a  smooth space-like Cauchy surface, for some constant  $c$ 
\cite[\textsection 4.3]{wald94}. In that case, a $\zeta$-invariant Hermitian  scalar product   can be  constructed out of a certain  auxiliary Hermitian scalar product  $(\psi_f| \psi_g)$ induced by the stress energy  tensor $$T_{ab}(\overline{\psi_f} ,\psi_g)   \stackrel {\mbox{\scriptsize  def}} {=}
\frac{1}{2}\left(\nabla_{a}\overline{\psi_f} \nabla_{b}\psi_g + \nabla_{b}\overline{\psi_f}  \nabla_{a}\psi_g  \right)
 -\frac{1}{2} g_{ab} \left(\nabla^c \overline{\psi_f}  \nabla_c \psi_g - m^2 \overline{\psi_f}  \psi_g  \right)
$$
 evaluated on solutions $\psi_f,\psi_g \in \Sol + i \Sol \stackrel {\mbox{\scriptsize  def}}{=} \Sol_{\mathbb C}$ of KG equation and contracted with $\zeta$ itself.
$$( \psi_f |\psi_g )  \stackrel {\mbox{\scriptsize  def}} {=} \int_{\Sigma} T^{ab}(\overline{\psi_f} ,\psi_g) n_a \zeta_b \: d\Sigma\:.$$
This positive Hermitian form does not depend on the Cauchy surface $\Sigma$ and is $\zeta$-invariant in view of the Killing equation for $\zeta$ and $\nabla^aT_{ab}(\overline{\psi_f},\psi_g)=0$ which holds as a consequence of KG equations for $\psi_f$ and $\psi_g$. This Hermitian scalar product gives rise to a complex Hilbert space ${\cal H}_0$ obtained by taking the completion of $\Sol_{\mathbb C}$. It turns out that the time evolution generated by $\zeta$ in $\Sol_{\mathbb C}$ is implemented by a strongly continuous unitary group on 
${\cal H}_0$, with self-adjoint generator $H$. The spectrum of $H$ is bounded away from zero and thus $E^{-1}$ exists as a bounded, everywhere defined, operator on ${\cal H}_0$.
Let ${\cal H}_{0}^+$ be the positive spectral closed subspace of $h$ and let $P_+: {\cal H}_0 \to {\cal H}_0^+$ be the corresponding orthogonal projector. The distinguished scalar product defining the quasifree pure (because (\ref{purecond}) holds) state $\omega_\zeta$ is finally defined by means of the real scalar
product
\begin{eqnarray} \mu_\zeta(\psi, \psi')   \stackrel {\mbox{\scriptsize  def}} {=} Re \left( P_+\psi \left| \frac{1}{2}H^{-1}P_+\psi' \right. \right) \quad \psi, \psi' \in \Sol\:. \label{muM2}\end{eqnarray}
 This procedure can be viewed as a  rigorous version of the popular one based of positive frequency mode decomposition  with respect to the notion of time associated to $\zeta$ in particular because it can easily be proved that $$H \psi = i \zeta^a \partial_a\psi$$
when $\psi \in \Sol_{\mathbb C}$. Therefore $P_+\psi$ entering the right hand side of  (\ref{muM2}) contains  ``positive frequencies'' only, since $P_+$ project on the positive part of the spectrum of the energy $H$.
The state $\omega_\zeta$ coincides with the Minkowski vacuum in Minkowski spacetime when $\zeta = \partial_t$ with respect to any Minkowski coordinate system. This result has a well-known \cite{wald94} important consequence.

\begin{theorem}[Existence of quasifree states] Consider a globally hyperbolic spacetime $\Mb$ and assume that $\xi=0$ and $m>0$ in the definition of $\Ac(\Mb)$. There exist quasifree states on  $\Ac(\Mb)$.
\end{theorem}

\noindent {\em Sketch of proof}.
Take a smooth space-like Cauchy surface $\Sigma \subset \Mc$.  It is always possible to smoothly deform $\Mc$ in the past of $\Sigma$ 
obtaining an overall globally hyperbolic spacetime still admitting $\Sigma$ as a Cauchy surface and such that 
 the open past of $\Sigma$, $\Mb^-$, (in the deformed  spacetime) has the following property.
There is a second Cauchy surface $\Sigma_1$ in $\Mb^-$ whose open past  $\Mb_1^-$ includes a smooth time-like Killing field $\zeta$ satisfying the sufficient requirements for defining and associate quasifree state $\omega_\zeta$ on ${\cal A}(\Mb_1^-)$.  However, if $\Mb^+$ denotes the 
open future of $\Sigma$ (in the original spacetime), Propositions \ref{iff} and \ref{tsa} easily imply that
${\cal A}(\Mb^+) = {\cal A}(\Mb^-) =   {\cal A}(\Mb_1^-)$. 
Therefore $\omega_\zeta$ is a  state on ${\cal A}(\Mb^+) =  {\cal A}(\Mb)$.  Again  Propositions \ref{iff} and \ref{tsa} and the very definition of quasifree state easily prove that $\omega_\zeta$ is quasifree on ${\cal A}(\Mb)$ if it is quasifree on ${\cal A}(\Mb_1^-)$. $\Box$

\subsection{Unitarily inequivalent quasifree states gravitationally produced}
\label{sec:uni-ineq}
Coming back to what already pronounced in  (6) in Remark \ref{rem1}, we have the following definition.

\begin{definition} Two states $\omega_1$ and $\omega_2$ on ${\cal A}(\Mb)$ and the respective GNS representations are said to be {\bf unitarily equivalent}\footnote{It should be evident that the given definition does not depend on the particular GNS representation chosen for each state $\omega_i$.}  if there is an isometric surjective operator $U: {\cal H}_{\omega_1} \to {\cal H}_{\omega_2}$ such that $U\hat{\phi}_{\omega_1}(f) U^{-1} = \hat{\phi}_{\omega_2}(f)$ for every $f \in C_0^\infty(\Mc)$. 
\end{definition}

\begin{remark}
Notice that it is not necessary that $U\Psi_{\omega_1} = \Psi_{\omega_2}$ and it generally does not happen. As a consequence ${\cal H}_{\omega_2}$ includes vector states {\em different from the Fock vacuum} which are however quasifree.
\end{remark}
The question if a pair of states are unitarily equivalent 
naturally arises in the following situation.
 Consider a time-oriented globally hyperbolic spacetime $\Mb$ such that, in the future of a Cauchy surface $\Sigma_+$, the spacetime  is stationary with respect to the Killing vector field $\xi_+$ and it is also stationary in the past of another Cauchy surface $\Sigma_-$, in the past of $\Sigma_+$, referring to another Killing vector field $\xi_-$.  For instance we can suppose that $\Mb$ coincides to (a portion of) Minkowski spacetime in the two mentioned stationary regions and a gravitational curvature bump takes place between them.
This way, two preferred quasifree states $\omega_+$ and $\omega_-$ turn out to be defined on the whole 
algebra ${\cal A}(\Mb)$, not only in the algebras of observables localized in the two respective static regions regions.   The natural question is whether or not the GNS representations of $\omega_+$ and $\omega_-$ are unitarily equivalent, so that, in particular,  the state $\omega_-$ can be represented as a vector state $U\Psi_{\omega_-}$ in the 
Hilbert space ${\cal H}_{\omega_+}$ of the state $\omega_{+}$.  Notice that, even in the case the isometric surjective operator $U$ exists making the representations unitarily equivalent,
$U\Psi_{\omega_-} \neq  \Psi_{\omega_+}$ in general, so that $U\Psi_{\omega_-}$ may have non-vanishing projection in the subspace containing states with  $n$ particles in ${\cal H}_{\omega_+}$. This phenomenon is physically interpreted as {\em creation of particles due to the gravitational field} and $U$ has the natural interpretation of an $S$ {\em matrix}.\\
The following crucial result holds for pure quasifree states \cite{wald94}.
A more general result appears  in  \cite{Verch}  since it avoids the assumption that the states are pure and it deals with the notion of {\em quasiequivalence} of quasifree states. Quasiequivalence is weaker notion of equivalence, which essentially corresponds to unitary equivalence ``up to multiplicity''~\cite[Sec.2.4.4]{BR}. In particular, quasiequivalence reduces to unitary equivalence for irreducible representations, as for instance those induced by pure states.

\begin{theorem}[Unitary equivalence of pure quasifree states]
If $\Mb$ is a globally hyperbolic spacetime, consider two pure quasifree states $\omega_1$ and $\omega_2$ on ${\cal A}(\Mb)$ respectively induced by the scalar product $\mu_1$ and $\mu_2$ on $\Sol(\Mb) \cong {\cal E}$ and indicate by ${\cal R}_{\mu_1}$ and ${\cal R}_{\mu_2}$ the real Hilbert spaces obtained by respectively completing $\Sol$. \\
The pure states $\omega_1$ and $\omega_2$ may be unitarily equivalent only if they induce equivalent norms on 
$\Sol$, that is there are constants $C,C'>0$ with 
$$C\mu_1(x,x) \leq \mu_2(x,x) \leq C'\mu_1(x,x) \quad \forall x \in \Sol\:.$$
When the condition is satisfied  there is a unique bounded operator   $Q : {\cal R}_{\mu_1} \to {\cal R}_{\mu_1} $
such that 
$$\mu_1(x, Q y)= \mu_2(x,y)-\mu_1(x,y)\quad \forall x,y \in \Sol\:.$$
 In this case $\omega_1$ and $\omega_2$ are unitarily equivalent if and only if $Q$ is Hilbert-Schmidt%
	\footnote{Note that this result is stated incorrectly in Theorem~4.4.1
	of~\cite{wald94}, where the condition on the operator $Q$ is
	incorrectly given as \emph{trace class} instead of
	\emph{Hilbert-Schmidt}. The correct condition is actually given in
	Equation~(4.4.21) of~\cite{wald94} as the Hilbert-Schmidt property of
	the operator $\mathcal{E}$ and the mistake appears in identifying the
	corresponding property of $Q$. We thank Rainer Verch and especially Ko
	Sanders for bringing this to our attention.}
in ${\cal R}_{\mu_1}$.
\end{theorem}

\noindent In general, the said condition fails when $\omega_1$ and $\omega_2$ are stationary states associated with two 
stationary regions (in the past and in the future)  of a spacetime, as discussed in the introduction of this section \cite{wald94}, \cite[Ch.7]{fulling}.  It happens in particular when the Cauchy surfaces have infinite volume.
In this case the states turn out to be unitarily inequivalent. On the other hand  there is no natural preferred choice between $\omega_+$ and $\omega_-$ and this fact suggests that  the algebraic formulation is more useful in QFT in curved spacetime than the, perhaps more familiar, formulation in a Hilbert space.

\subsection{States invariant under the action of spacetime  symmetries}\label{statessym}
The quasifree state $\omega_\zeta$ on ${\cal A}(\Mb)$ mentioned to exist above for stationary globally hyperbolic spacetimes for massive scalar fields  is {\em invariant} under the action of $\zeta$ (which we assume to be complete for the sake of simplicity) in the following sense.  Just because $\zeta$ is a Killing field (and the subsequent construction does depend on the fact that $\zeta$ is time-like), the action of the one-parameter group of isometries  $\{\chi^{(\zeta)}_t\}_{t \in \mathbb R}$ generated by $\zeta$ leaves $\Sol$ invariant. This is equivalent to saying that when $\{\chi^{(\zeta)}_t\}_{t \in \mathbb R}$ acts on $C_0^\infty(\Mc)$ it  preserves $E$ (the commutation relations of quantum fields  are consequently preserved in particular).
In view of Proposition \ref{propind}, a one-parameter group of $*$-algebra isomorphisms 
$\alpha^{(\zeta)}_t : {\cal A}(\Mb) \to {\cal A}(\Mb)$
arises this way, completely defined by the requirement beyond the obvious  $\alpha_t(\II) = \II$ if $t\in \mathbb R$ and
$$\alpha^{(\zeta)}_t \left( \phi(f)\right)   \stackrel {\mbox{\scriptsize  def}} {=} \phi\left(  f\circ \chi^{(\zeta)}_{-t}\right)\:, \quad t \in \mathbb R\:, \quad f \in C_0^\infty(\Mc) \:.$$
 It turns out that, if $\omega_\zeta$ is constructed by the procedure above mentioned when $\zeta$ is time-like,
$\omega_\zeta$ is $\zeta$-invariant in the sense that: 
\begin{eqnarray}\omega_\zeta \circ  \alpha^{(\zeta)}_t = \omega_\zeta \quad \forall t \in \mathbb R\:.\label{inv}\end{eqnarray}
When passing to the GNS representation, Proposition \ref{propsymm} implies that there is a one-parameter group of unitary operators such that

(i) \:\: $U^{(\zeta)}_t \Psi_{\omega_\zeta} = \Psi_{\omega_\zeta}$ ,\:  $U^{(\zeta)}_t({\cal D}_{\omega_\zeta})= 
{\cal D}_{\omega_\zeta}$\:,

(ii) \:\:  $U^{(\zeta)}_t \pi_{\omega_\zeta}(a) U^{(\zeta)*}_t   \stackrel {\mbox{\scriptsize  def}} {=} \pi_{\omega_\zeta}\left(\alpha^{(\zeta)}_t(a)\right)$
for all $t\in \mathbb R$ and $a \in {\cal A}(\Mb)$.\\
Moreover we know that $\{U^{(\zeta)}_t\}_{t\in \mathbb R}$ is {\em strongly continuous} if and only if 
$$\lim_{t\to 0} \omega_\zeta\left(a^* \alpha^{(\zeta)}_t(a) \right) = \omega_\zeta(a^*a)\:,\quad \forall a\in {\cal A}(\Mb)\:.$$
In this case, Stone's theorem entails that there is a unique self-adjoint operator $H^{(\zeta)}$ with $e^{-itH^{(\zeta)}} = U^{(\zeta)}_t$ for every $t \in \mathbb R$ and $H^{(\zeta)}\Psi_{\omega_\zeta} =0\:.$ 
If $\sigma(H^{(\zeta)}) \subset [0,+\infty)$ and 
$\Psi_{\omega_\zeta}$ is, up to factors, the unique eigenvector of $H^{(\zeta)}$ with eigenvalue $0$, $\omega_\zeta$
is said to be a {\bf ground state} (this definition generally applies to an  invariant state under the action of a time-like Killing symmetry, no matter if the state is quasifree).

\begin{remark} The one-parameter group $U_t^{(\zeta)}$ associated with the time-like-Killing vector field $\zeta$
has the natural interpretation of {\bf time evolution} with respect to the {\em notion of time} associated with $\zeta$ and, in case the group is strongly continuous $H^{(\zeta)}$ is the natural {\bf Hamiltonian} operator 
associated with that evolution.  However, for a generic time-oriented globally hyperbolic spacetime, no notion of Killing time is suitable and consequently, no notion of (unitary) time evolution is possible. Time evolution {\em \`a la Schroedinger} is not a good notion to be extended to QFT in curved spacetime. Observables do not evolve, they are localized in bounded regions of spacetime by means of the smearing procedure. Causal relations are encompassed by the {\em Time-slice axiom} (see \cite{chapt:BD}) which is a theorem for free fields (Proposition \ref{tsa}). 
\end{remark}

\noindent Abandoning the case of time-like Killing symmetries, it is worth stressing that, generally speaking, every isometry $\gamma: \Mb \to \Mb$, not necessarily Killing and not necessarily time-like if Killing,
induces a corresponding automorphism of unital $*$-algebras, $\beta^{(\gamma)}$ of ${\cal A}(\Mb)$, via Proposition \ref{propind}, completely defined by the requirements  $\beta^{(\gamma)}(\phi(f))  \stackrel {\mbox{\scriptsize  def}} {=} \phi\left( f\circ \gamma^{-1} \right)$. If a state $\omega$ is invariant under $\beta^{(\gamma)}$, we can apply Proposition \ref{propsymm}, in order to unitarily implement this symmetry in the GNS representation of $\omega$. 
Some discrete symmetries  can be represented in terms of anti-linear automorphisms, like the time reversal in Minkowski spacetime. Again  $\beta^{(\gamma)}(\phi(f))  \stackrel {\mbox{\scriptsize  def}} {=} \phi\left( f\circ \gamma^{-1} \right)$ completely determine the anti-linear automorphism via Proposition \ref{propind}. If a state $\omega$ is invariant under $\beta^{(\gamma)}$, we can apply Proposition \ref{propsymm}, in order to implement this symmetry anti-unitarily in the GNS representation of $\omega$.

\begin{remark}$\null$

{\bf (1)} It easy to prove that, if the state $\omega : \Ac(\Mb) \to \mathbb C$ is invariant under the (anti-linear) automorphism $\beta: \Ac(\Mb) \to \Ac(\Mb)$ is {\em quasifree}, the spaces with fixed number of particles of the GNS Fock representation of $\omega$ are separately invariant under the action of the unitary (resp.  anti-unitary) operator $U^{(\beta)}$ implementing $\beta$ in the Fock representation of $\omega$ in view of Proposition \ref{propind}.

{\bf (2)} A known result \cite{Kay} establishes the following remarkable uniqueness result (actually proved for Weyl algebras, but immediately adaptable to our CCR framework).
\begin{proposition}[Uniqueness of pure invariant quasifree states]\label{uniquenessqfis} Assume that a quasifree state $\omega : \Ac(\Mb) \to \mathbb C$  is pure and  invariant under a one-parameter group of automorphisms 
$\{\beta_t\}_{t \in \mathbb R}$ of $\Ac(\Mb)$, giving rise to a strongly continuous unitary group $\{U_t\}_{t \in \mathbb R}$ implementing $\{\beta_t\}_{t \in \mathbb R}$ in the GNS representation of $\omega$. The pure quasifree state $\omega$  is uniquely determined by $\{\beta_t\}_{t \in \mathbb R}$ if the self-adjoint generator of $\{U_t\}_{t \in \mathbb R}$ restricted to  the one-particle Hilbert space of $\omega$ is positive without zero eigenvalues.
\end{proposition}
\end{remark}

\noindent Let us focus on the {\em Minkowski vacuum}, that is  the quasifree  state $\omega_{\mathbb M}$ on four dimensional Minkowski spacetime $\mathbb M$ defined in  Section \ref{existenceqfs} by the two-point function (\ref{W}). As a matter of fact, $\omega_{\mathbb M}$ turns out to be invariant under the natural action 
of {\em orthochronous proper Poincar\'e group} and that the corresponding unitary representation of this connected Lie (and thus topological) group is strongly continuous.  In particular the self-adjoint generator of time displacements (with respect to every timelike direction), in the one-particle Hilbert space, satisfies the hypotheses of Proposition \ref{uniquenessqfis}. As $\omega_{\mathbb M}$ is pure, it {\em is therefore the unique pure quasifree state invariant under the orthochronous proper Poincar\'e group}.   $\omega_{\mathbb M}$ is a ground state with respect to any Minkowski  time evolution and,
by direct inspection, one easily sees that the state it is also invariant under the remaining discrete symmetries of Poincar\'e group $T$, $P$ and $PT$ which are consequently (anti-)unitarily implementable in the GNS Hilbert space.
Finally, it turns out that the one-particle space is irreducible under the action of the orthochronous proper Poincar\'e group, thus determining an {\em elementary particle} in the sense of the {\em Wigner classification}, with mass $m$ and zero spin.

\section{Hadamard quasifree states in curved spacetime}
The algebra of observables generated by the field $\phi(f)$ smeared with smooth functions is too 
small to describe important observables in QFT in curved spacetime. Maybe the most important is the stress energy tensor (obtained as a functional derivative of the action with respect to $g^{ab}$) that, for our Klein-Gordon field it reads, where $G_{ab}$ is the standard Einstein  tensor 
\begin{align}T_{ab}   \stackrel {\mbox{\scriptsize  def}} {=} &\; (1-2\xi)\nabla_{a}\phi\nabla_{b}\phi-2\xi\phi\nabla_{a}\nabla_b\phi-\xi \phi^2 G_{\mu\nu} \nonumber \\&\;+g_{ab}\left\{2\xi \phi^2 \box \phi+\left(2\xi-\frac12\right)\nabla^c\phi\nabla_{c}\phi+\frac12 m^2 \phi^2 
\right\} \label{Tabgen}\:.
\end{align}
It concerns products of fields evaluated at the same point of spacetime, like $\phi^2(x)$.  This observable, as usual smeared with a function $f \in C_0^\infty(\Mc)$,  could be formally interpreted as 
\begin{eqnarray} \phi^2(f) = \int_{\Mc} \phi(x)\phi(y) f(x) \delta(x,y) \:\:\dvol_{\Mb} \:. \label{phi2} \end{eqnarray}
However this object does not belong to ${\cal A}(\Mb)$.  Beyond  the fact that $T_{ab}$ describe the local content of energy, momentum and stress of the field,
 the stress-energy tensor is of direct relevance for describing the back reaction on the quantum fields on the spacetime geometry  through the  semi-classical Einstein equation 
\begin{equation} G_{ab}(x)  = 8\pi  \omega(T_{ab}(x))\label{backreaction}\end{equation}
or also, introducing a smearing procedure
$$ \int_{\Mc}  G_{ab}(x)  f(x) \:\:\dvol_{\Mb}  = 8\pi   \int_{\Mc} \omega( T_{ab}(x))    f(x) \:\:\dvol_{\Mb} \:, $$
where  $\omega( T_{ab}(x) )$ has the interpretation of the (integral kernel of the) expectation value of the quantum observable $T_{ab}$  with respect to some quantum state $\omega$. 
Barring technicalities due to the appearance of derivatives,  the overall problem is here to provide (\ref{phi2}) with a precise mathematical meaning, which in fact, is equivalent to a suitable  enlargement  the algebra ${\cal A}(\Mb)$. 

\subsection{Enlarging the observable algebra in Minkowski spacetime}\label{SectEnlM}
In flat spacetime $\Mb= \mathbb M$, for free QFT, at the level of {\em expectation values} and {\em quadratic forms} the above mentioned  enlargement of the algebra is performed exploiting  a {\em physically meaningful  reference state}, the  unique Poincar\'e invariant quasifree (pure) state introduced in Section \ref{existenceqfs} and discussed at the end of Section \ref{statessym}, $\omega_{\mathbb M}$.  We   call this state {\em Minkowski vacuum}.

Let us  first focus on the elementary observable $\phi^2$. We shall indicate it with $:\spa \phi^2(x) \spa:$ and we define it as a {\em Hermitian quadratic form} on ${\cal D}_{\omega_{\mathbb M}}$.

We start by  defining the operator on ${\cal D}_{\omega_{\mathbb M}}$ for $f,g \in C_0^\infty(\mathbb R^4)$
\begin{eqnarray}  :\spa \hat{\phi}(f) \hat{\phi}(g)\spa: \:   \stackrel {\mbox{\scriptsize  def}} {=} \hat{\phi}(f)\hat{\phi}(g)  -  \langle \Psi_{\omega_{\mathbb M}}|\hat{\phi}(f)\hat{\phi}(g) \Psi_{\omega_{\mathbb M}} \rangle I \label{PT}\end{eqnarray}
(As usual $\hat{\phi}(f)  \stackrel {\mbox{\scriptsize  def}} {=} \hat{\phi}_{\omega_{\mathbb M}}(f)$ throughout this section.)
Next, for $\Psi \in {\cal D}_{\omega_{\mathbb M}}$ we analyze  its integral kernel, assuming that it exists, $\langle \Psi|:\spa \hat{\phi}(x) \hat{\phi}(y)\spa: \Psi\rangle$ which  is symmetric since the antisymmetric part of the right-hand side of (\ref{PT}) vanishes in view of the commutation relations of the field.
The explicit form of the distribution  $\omega_{{\mathbb M} 2}(x,y) =\langle \Psi_{\omega_{\mathbb M}}|\hat{\phi}(x)\hat{\phi}(y) \Psi_{\omega_{\mathbb M}}\rangle$ appears in (\ref{W}).
We prove below that the mentioned formal kernel $\langle \Psi|:\spa \hat{\phi}(x) \hat{\phi}(y)\spa: \Psi\rangle$ not only exists but it also is  a jointly  smooth function. Consequently we are allowed to define, for any $\Psi \in {\cal D}_{\omega_{\mathbb M}}$,
 \begin{equation}\langle \Psi| :\spa\hat{\phi}^2\spa:(f) \Psi \rangle     \stackrel {\mbox{\scriptsize  def}} {=} \int_{\mathbb M^2} \langle \Psi|:\spa \hat{\phi}(x) \hat{\phi}(y)\spa: \Psi\rangle f(x) \delta(x,y)\:\: \dvol_{\mathbb M^2}(x,y)\:.\label{phi22}\end{equation}
Finally, the polarization identity uniquely defines $:\spa\hat{\phi}^2\spa:(f)$ as a symmetric quadratic form  ${\cal D}_{\omega_{\mathbb M}}\times{\cal D}_{\omega_{\mathbb M}}$.
\begin{align}
\langle \Psi'| &:\spa\phi^2\spa:(f) \Psi \rangle   \stackrel {\mbox{\scriptsize  def}} {=} \frac{1}{4} \left( \langle \Psi' + \Psi | :\spa\phi^2\spa:(f)  (\Psi'+\Psi) \rangle 
 -\langle \Psi' - \Psi |  :\spa\phi^2\spa:(f)  (\Psi'-\Psi)\rangle\right. \nonumber\\
&\left. -i \langle \Psi' + i\Psi |  :\spa\phi^2\spa:(f)  (\Psi'+i\Psi)
+i  \langle \Psi' -i \Psi | :\spa\phi^2\spa:(f)  (\Psi' -i\Psi) \rangle
\right)\label{polar}
\end{align}
There is no guarantee that an operator $:\spa\phi^2\spa:(f)$ really exists on ${\cal D}_{\omega_{\mathbb M}}$ satisfying (\ref{phi22})\footnote{By Riesz lemma, it exists if an only if the map ${\cal D}_{\omega_{\mathbb M}} \ni \Psi' \mapsto \langle \Psi'| :\spa\phi^2\spa:(f) \Psi \rangle$ is continuous for every $\Psi \in {\cal D}_{\omega_{\mathbb M}}$}, however if it exists, since ${\cal D}_{\omega_{\mathbb M}}$ is dense and (\ref{polar}) holds,  it is uniquely determined by the class of 
the expectation values $\langle \Psi| :\spa\hat{\phi}^2\spa:(f) \Psi \rangle$ on the states  $\Psi \in {\cal D}_{\omega_{\mathbb M}}$.
As promised, let us  prove that the kernel defined in  (\ref{phi22}) is a smooth function. 
First of all, notice that, as a general result arising from the  GNS   construction, every $\Psi \in {\cal D}_{\omega}$ can be written as \begin{equation}\Psi =  \sum_{n \geq 0,  i_1, \ldots, i_n \geq 1 } C^{(n)}_{i_1 \ldots i_n} \hat{\phi}(f^{(n)}_{i_1}) \cdots  \hat{\phi}(f^{(n)}_{i_n}) \Psi_{\omega_{\mathbb M}}\label{Psiexp}\end{equation} where only a finite number of coefficients $C^{(n)}_{i_1 \ldots i_n} \in \mathbb C$ is non-vanishing and the term in the sum  corresponding to $n=0$ is defined to have the form $c^0 \Psi_{\omega_\mathbb M}$.
We have
\begin{multline} \label{expansion}  \langle \Psi|\hat{\phi}(x)\hat{\phi}(y) \Psi\rangle  =
 \sum_{n \geq 0,  i_1, \ldots, i_n \geq 1 } \sum_{m \geq 0,  j_1, \ldots, j_n \geq 1 } \overline{C^{(m)}_{j_1 \ldots j_n}}  C^{(n)}_{i_1 \ldots i_n} \\
 \left\langle \left.\Psi_{\omega_{\mathbb M}}\right|
 \hat{\phi}(f^{(m)}_{j_m}) \cdots  \hat{\phi}(f^{(m)}_{j_1})  \hat{\phi}(x)\hat{\phi}(y)  \hat{\phi}(f^{(n)}_{i_1}) \cdots  \hat{\phi}(f^{(n)}_{i_n}) \Psi_{\omega_{\mathbb M}} \right\rangle\:.
\end{multline}
Taking advantage of the quasifree property of $\omega_{\mathbb M}$, hence using the expansion of $n$-point functions in terms of the $2$-point function of Definition \ref{defqfs}, we can re-arrange the right hand side of (\ref{expansion}) as (all the sums are over finite terms)
\begin{multline}\label{expansion2}
	\langle \Psi|\hat{\phi}(x)\hat{\phi}(y) \Psi\rangle
	= C^{0}_{\Psi} \omega_{\mathbb{M}2}(x,y) \\
		+ \sum_{m\ge 0, j\ge 1} \sum_{m'\ge 0, j'\ge 1} C^{(m)(m')}_{\Psi,j,j'}
			\omega_{\mathbb{M}2}(f^{(m)}_j,x) \omega_{\mathbb{M}2}(f^{(m')}_{j'},y) \\
		+ \sum_{m\ge 0, j\ge 1} \sum_{n\ge 0, i\ge 1} C^{(m)(n)}_{\Psi,j,i}
			\omega_{\mathbb{M}2}(f^{(m)}_j,x) \omega_{\mathbb{M}2}(y,f^{(n)}_{i}) \\
		+ \sum_{n'\ge 0, i'\ge 1} \sum_{n\ge 0, i'\ge 1} C^{(n')(n)}_{\Psi,i',i}
			\omega_{\mathbb{M}2}(x,f^{(n')}_{i'}) \omega_{\mathbb{M}2}(y,f^{(n)}_{i}) \: ,
\end{multline}
with all sums finite and some $C_\Psi$-coefficients that depend on the state $\Psi$. We can be more specific about the first coefficient, in fact, according to the formula from Definition~\ref{defqfs}, we have $C^0_\Psi = \langle \Psi | \Psi \rangle$. Recall also, from Proposition~\ref{prp:mink-smooth}, that $y \mapsto \omega_{\mathbb{M}2}(f,y)$ and $x \mapsto \omega_{\mathbb{M}2}(x,f)$ are smooth for any test function $f\in C^\infty_0(\Mc)$. Hence, we can interpret Equation~\eqref{expansion2} as saying that
\begin{equation}\label{expansion3}
	\langle \Psi|:\spa \hat{\phi}(x) \hat{\phi}(y)\spa:\Psi\rangle
	= \langle \Psi|\hat{\phi}(x)\hat{\phi}(y) \Psi\rangle
		- \langle \Psi| \Psi \rangle \omega_{\mathbb M}(x,y)
	\in C^\infty (\Mc\times \Mc) \: .
\end{equation}

 More complicated operators, i.e. {\bf Wick polynomials}  and corresponding differentiated Wick polynomials, generated by {\bf Wick monomials}, $:\spa\hat{\phi}^n\spa:(f)$, of arbitrary order $n$,   can analogously be defined as quadratic forms, by means of a recursive procedure of subtraction of divergences. The stress energy operator is a differentiated  Wick polynomial of order $2$. \\ 
The procedure for defining $:\spa\hat{\phi}^n\spa:(f)$ as a quadratic form is as 
follows. First define recursively, where the tilde just means that the indicated element has to be omitted,
\begin{align}
:\spa \hat{\phi}(f_1)\spa: \: &   \stackrel {\mbox{\scriptsize  def}} {=} \:\hat{\phi}(f_1)\nonumber \\
:\spa \hat{\phi}(f_1) \cdots \hat{\phi}(f_{n+1}) \spa: \: &   \stackrel {\mbox{\scriptsize  def}} {=} \: :\spa \hat{\phi}(f_1) \cdots \hat{\phi}(f_n) \spa:  \hat{\phi}(f_{n+1})\nonumber \\
 &\:\: \:\:\quad  - \sum_{l=1}^n :\spa \hat{\phi}(f_1) \cdots  \widetilde{\hat{\phi}(f_l)} \cdots \hat{\phi}(f_n) \spa: \: \omega_{\mathbb M 2}(f_l,f_{n+1})\:. \label{phin}
\end{align}
These elements of $\Ac(\mathbb M)$ turn out to be symmetric under interchange of $f_1, f_2, \ldots f_n$ as it can be proved by induction\footnote{Observe in particular that $:\spa\hat{\phi}(f) \hat{\phi}(g)\spa: - :\spa\hat{\phi}(g) \hat{\phi}(f)\spa: = iE(f,g)\II- \omega_{\mathbb M2}(iE(f,g)\II)\II=0$.}.
By induction, it is next possible to prove that, for $n\geq 2$ and  $\Psi \in {\cal D}_{\omega_{\mathbb M}}$, there is a {\em jointly  smooth kernel}  $$\langle \Psi | :\spa \hat{\phi}(x_1) \cdots \hat{\phi}(x_n) \spa: \Psi \rangle$$  which produces $\langle \Psi | :\spa \hat{\phi}(f_1) \cdots \hat{\phi}(f_n) \spa: \Psi \rangle$ by integration.
This result  arises from (\ref{phin}) as a consequence of the fact that 

(a) $\omega_{\mathbb M}$ is quasifree so that Definition \ref{defqfs} can be used to compute the said kernels, 

(b) $\Psi \in {\cal D}_{\omega_{\mathbb M}}$ so that the expansion (\ref{Psiexp}) can be used, 

(c) the functions $F_k : x \mapsto \omega_{\mathbb M 2}(x,f_k) = \overline{\omega_{\mathbb M 2}(f_k,x)}$ are smooth when  $f_k \in C_0^\infty(\Mc)$ as was mentioned above.\\
Indeed, we have
\begin{eqnarray}  \langle \Psi|:\spa\hat{\phi}(x_1)\cdots \hat{\phi}(x_n)\spa: \Psi\rangle  =
 \sum_{n \geq 0,  i_1, \ldots, i_n \geq 1 } \sum_{m \geq 0,  j_1, \ldots, j_n \geq 1 } \overline{C^{(m)}_{j_1 \ldots j_n}}  C^{(n)}_{i_1 \ldots i_n}\nonumber \\
 \left\langle \left.\Psi_{\omega_{\mathbb M}}\right|
 \hat{\phi}(f^{(m)}_{j_m}) \cdots  \hat{\phi}(f^{(m)}_{j_1})  :\spa\hat{\phi}(x_1)\cdots\hat{\phi}(x_n)\spa:  \hat{\phi}(f^{(n)}_{i_1}) \cdots  \hat{\phi}(f^{(n)}_{i_n}) \Psi_{\omega_{\mathbb M}} \right\rangle \label{expansionn}\:.
\end{eqnarray}
after having expanded the normal product $ :\spa\hat{\phi}(g_1)\cdots\hat{\phi}(g_n)\spa: $ in the right-hand side, one can evaluate the various $n$-point functions arising this way by applying Definition \ref{defqfs}.
It turns out that all terms $\omega_{\mathbb M2}(x_i,x_j)$ always  appear in a sum with  corresponding terms $-\omega_{\mathbb M2}(x_i,x_j)$  arising by the definition (\ref{phin}) and thus give no contribution. The remaining factors are of the form $F_k(x_j)$ and thus are smooth.\\
We therefore are in a position to write the definition of $\langle \Psi| :\spa \hat{\phi}^n\spa:(f) \Psi \rangle$ if $\Psi \in {\cal D}_{\omega_{\mathbb M}}$
\begin{equation}\langle \Psi| :\spa \hat{\phi}^n\spa:(f) \Psi \rangle =
\int_{\Mc^n} \langle \Psi | :\spa \hat{\phi}(x_1) \cdots \hat{\phi}(x_n) \spa: \Psi \rangle f(x_1) \delta(x_1,\ldots x_n) \dvol_{\Mb^n}\label{phinN}\end{equation}
Exactly as before, polarization extends the definition to a quadratic form on ${\cal D}_{\omega_{\mathbb M}}\times {\cal D}_{\omega_{\mathbb M}}$. There is no guarantee that operators fitting these quadratic forms really exist.

\begin{remark} $\null$
The definition (\ref{phin}) can be proved to be formally equivalent to the formal definition
\begin{equation}:\spa \hat{\phi}(x_1) \cdots \hat{\phi}(x_n) \spa:  \:   \stackrel {\mbox{\scriptsize  def}} {=} \left. \frac{1}{i^n}\frac{\delta^n}{\delta f(x_1)\cdots \delta f(x_n)} \right|_{f=0}e^{i \hat{\phi}(f) +\frac{1}{2}\omega_{\mathbb M2}(f,f)} \label{formT}
\end{equation}
Though the exponential converges in the strong operator topology  to a unitary operator, the {\bf Weyl generator}, restricted to the dense domain ${\cal D}_{\omega_{\mathbb M}}$,
$e^{i \hat{\phi}(f)}$ can be viewed here as a formal series and 
this  series can be truncated at finite, sufficiently large, order in view of linearity of the exponent and  $f=0$.
\end{remark}
\subsection{Enlarging the observable algebra  in curved spacetime}
The discussed definition of Wick polynomials is {\em equivalent} in Minkowski spacetime to the more popular one based on the well known re-ordering procedure of creation and annihilation operators as can be proved by induction. Nevertheless this second approach is not natural in curved spacetime because, to be implemented, it needs the existence of a physically preferred reference state as Minkowski vacuum 
in flat spacetime, which in the general case it is not given. To develop a completely covariant theory another approach has been  adopted,  which generalises to curved spacetime the previously outlined definition of Wick polynomials based on a ``divergence subtraction'' instead of a re-ordering procedure.  The idea is that, although it is not possible to 
uniquely assign each spacetime with a physically distinguishable state, it is possible to  select a type of  divergence in common with all  physically relevant states is every spacetime.
These  preferred quasifree states with the same type of divergence   ``resembling'' Minkowski vacuum in a generic spacetime  are called {\em Hadamard states}.  Minkowski vacuum belongs to this class and these states are  remarkable also in view of their  {\em microlocal features}, which revealed to be of crucial importance  for the  technical advancement of the  theory, as we will describe later.
Exploiting these distinguished states,  it is possible  to generalize the outlined approach  
in order to enlarge ${\cal A}(\Mb)$, including other algebraic elements as the stress-energy tensor operator  \cite{stress, HW04}. Actually this is nothing but the first step to generalize the ultraviolet renormalization procedure to curved spacetime \cite{wald94,bfk,bf,hw,hw2b}.
The rest of the chapter is devoted to discuss some elementary properties of Hadamard states.

Let us quickly remind some local features of (pseudo)Riemannian differential geometry \cite{O'Neill}, necessary to introduce the notion of Hadamard states from a geometric viewpoint.
If $(\Mc, g)$ is a smooth Riemannian or Lorentzian manifold, an open set $C\subset M$
 is said a {\bf normal convex} neighborhood
 if
 there is a open set $W \subset TM$ with the form $W=\{ (q,v)\:|\: q\in C, v\in S_q\}$ where  $S_q\subset T_qM$ is 
 a star-shaped open neighborhood of the origin, such that $$exp\spa\rest_W : (q,v) \mapsto
 exp_q v$$ is a diffeomorphism onto $C\times C$. It is clear that $C$ is connected and  there is only 
one geodesic segment
 joining any pair $q,q'\in C$ if we require that it is completely contained in $C$. It is  $[0,1] \ni t\mapsto exp_q(t ((exp_q)^{-
1}q'))$. Moreover if $q\in C$ and we fix a basis $\{e_\alpha|_q\}\subset T_qM$,
  $$t = t^\alpha e_\alpha|_q \mapsto exp_q(t^\alpha e_\alpha|_q)\:,\quad  t\in S_q$$ defines a set of 
 coordinates on $C$ centered in $q$ which is called the {\bf normal Riemannian coordinate 
system}  centered in   $q$.
 In $(M,g)$ as above, $\sigma(x,y)$ indicates the {\bf squared (signed) geodesic distance} of $x$ from $y$. With our signature  $(+,-,\cdots,-)$, it is defined as
 $$\sigma(x,y)   \stackrel {\mbox{\scriptsize  def}} {=} -g_x(exp_x^{-1}y,exp_x^{-1}y)\:.$$
$\sigma(x,y)$ turns out to be  smoothly defined on $C\times C$ if $C$ is a convex normal 
neighborhood where we also have $\sigma(x,y)=\sigma(y,x)$.
 The class of the convex normal neighborhoods of a point $p\in \Mc$ is a fundamental system of
 neighborhoods of $p$ \cite{friedlander,BEE}.\\
In Euclidean manifolds $\sigma$ defined as above  is everywhere nonnegative with the standard Euclidean choice of the signature. 

In a convex neighborhood $C$ of a spacetime $\Mb$, taking in particular advantage of several properties of $\sigma$, it is possible to define a local  {\em approximate solution} of KG equation, technically called a {\em parametrix}, which has essentially the same short-distance singularity of the two point function of Minkowski vacuum. Its construction uses only the local geometry and the parameters defining the equation of motion but does not refers to particular states, which are global objects. The technical idea can be traced back to Hadamard \cite{Hadamard} (and extensively
 studied by Riesz \cite{Riesz})
 and it is therefore called {\em Hadamard parametrix}. In the rest of the chapter we only consider a four dimensional spacetime, essentially following \cite{HackMoretti}. A quick technical discussion on the general case (details and properties of the constructions strongly depend of the dimension of the spacetime) also in relation with heath kernel expansion, can be found  in \cite{stress} (see also \cite{garabedian, friedlander, gunther,BGP} for more extended discussions also on different types of parametrices and their use in field theory).  In a convex neighborhood $C$ of a  four dimensional spacetimes the {\bf  Hadamard parametrix}  of order $N$
of the two-point function has the form
\begin{equation} H^{(N)}_{\epsilon}(x,y)=\frac{u(x,y)}{(2\pi)^2\sigma_\epsilon(x,y)}+\sum\limits^N_{n=0}v_n \sigma^n\log\left(\frac{\sigma_\epsilon(x,y)}{\lambda^2}\right)\label{ZNn}\end{equation}
where $x,y \in C$,  $T$ is any local time coordinate increasing towards the future,
$\lambda>0$ a  length scale
 and 
\begin{eqnarray} \sigma_\epsilon(x,y)   \stackrel {\mbox{\scriptsize  def}} {=} \sigma(x,y) +2i\epsilon (T(x)-T(y)) + \epsilon^2\label{sigma}\:,\end{eqnarray}
finally, the cut in the complex domain of the $\log$ function is assumed along the negative axis in (\ref{ZNn}).
Recursive differential equations 
 (see the appendix  A of \cite{stress} and also  \cite{Riesz,garabedian,friedlander,fulling,Moretti2, Thomas}) determine  $u=u(x,y)$ and all the {\bf Hadamard coefficients} $v_n= v_n(x,y)$ in $C$ as smooth functions, when assuming $u(x,x)=1$ and  $n= 0, 1, 2,\ldots $.
These recurrence relations  have been obtained by requiring that the sequence of the $H_{0}^{(N)}(x,y)$
defines a local, $y$-parametrized,``approximate solution'' of the KG equation for $\sigma(x,y)\neq 0$  (with some further details we can say that the error with respect to a true solution  is of order $\sigma^N$ for each $N$). That solution
would be  exact in the  $N\to \infty$ limit of the sequence provided the limit
exists.  The limit exists in the analytic case, but in the smooth general case the sequence diverges.
However, as proved in \cite[\textsection 4.3]{friedlander}, if $\chi: \mathbb R \to [0, 1]$ is a smooth function with
$\chi(r)=1$ for $|r| \leq 1/2$ and $\chi(r)=0$ for $|r|>0$ one can always find a sequence of numbers 
$0< c_1 < c_2 < \cdots <c_n \to +\infty$  for that
\begin{equation} v(x,y)   \stackrel {\mbox{\scriptsize  def}} {=} \sum_{n=0}^\infty v_n(x,y) \sigma(x,y)^n \chi(c_n \sigma(x,y))\label{V}\end{equation}
uniformly converges, with all derivatives,  to a $C^\infty$ function on $C\times C$. A parametrix $H_\epsilon$
\begin{equation} H_{\epsilon}(x,y)=\frac{u(x,y)}{(2\pi)^2\sigma_\epsilon(x,y)}+ v(x,y)\log\left(\frac{\sigma_\epsilon(x,y)}{\lambda^2}\right)\label{Z}\end{equation}
arises this way. This parametrix distributionally satisfies KG equation in both arguments up to jointly smooth functions of $x$ and $y$. In other words, there is a smooth function $s$ defined in $C\times C$ such that  if $f,g \in C_0^\infty(C)$ and defining $P  \stackrel {\mbox{\scriptsize  def}} {=} \Box_\Mb + m^2 + \xi R$,
\begin{equation}\lim_{\epsilon \to 0^+} \int_{C\times C}  H_{\epsilon}(x,y) (Pf)(x) g(y) \dvol_{\Mc\times \Mc} =
\int_{C\times C} s(x,y) f(x) g(y) \dvol_{\Mc\times \Mc}\:.\end{equation}
The analog holds swapping the role of the test functions.
We are in a position to state our main definition.
\begin{definition}
 \label{def_HadamardFormScalar}
With $\Mb$ four dimensional, we say that a (not necessarily quasifree) state $\omega$ on
 $\Ac(\Mb)$ and its two point function  $\omega_2$ are {\bf  Hadamard} if $\omega_2 \in {\cal D}'(\Mc \times \Mc)$ 
and  every point of $\Mb$ admits an open normal  neighborhood $C$ where 
\begin{equation}\omega_2(x,y) - H_{0^+}(x,y)=  w(x,y)  \quad \mbox{for some }\:\:  w\in C^\infty(C\times C)\:. \label{hadamard}\end{equation}
Here $0^+$ indicates  the standard weak distributional limit as $\epsilon \to 0^+$ (``first integrate against test functions and next  take the limit'').
\end{definition}

\begin{remark}\label{remHadpar} $\null$

{\bf (1)}  The given definition does not depend either on the choice of $\chi$ or the sequence of the $c_n$ used in (\ref{V}) since different choices simply change $w$ as one may easily prove. Similarly, the definition does not   depend  on the choice of the local time function $T$ used in the definition of $\sigma_\epsilon$. This fact is far from being obvious and requires a more detailed analysis  \cite{kw}.

{\bf (2)}   Using the following result arising form recurrence relations determining the Hadamard coefficients, one finds that the distribution $$\left(v(x,y)  - \sum_{k=0}^N v_n(x,y) \sigma(x,y)^n\right) \ln \sigma_{0^+}(x,y)$$
is a function in  $C^N(O \times O)$. Exploiting  this result,  it is not difficult to prove that the requirement (\ref{hadamard}) is equivalent to the following requirement:
\begin{equation} \omega_2(x,y) - H^{(N)}_{0^+}(x,y) = w_N(x,y)\quad \mbox{for each $N\ge 1$, with }\:\: w_N \in C^N(C\times C)\label{hadamardN}\:.\end{equation}
The equivalent definition of  Hadamard state in \cite{ra1} was, in fact, nothing but  Definition \ref{def_HadamardFormScalar} with (\ref{hadamard}) replaced by (\ref{hadamardN}).

{\bf (3)} Minkowski vacuum $\omega_{\mathbb M}$ defined by the two point function (\ref{W}) is Hadamard. In particular, for $m>0$, it holds\footnote{The function $z \mapsto I_1(\sqrt{z})/\sqrt{z}$, initially defined for $Re(z)>0$, admits a unique analytic extension on the whole space $\mathbb C$ and the formula actually refers to this extension.} 
$$\omega_{\mathbb M2}(x,y) =\frac{1}{4\pi^2} \frac{1}{\sigma_{0^+}(x,y)} + \frac{m^2}{2(2\pi)^2}\frac{I_1(m \sqrt{\sigma(x,y)})}{m \sqrt{\sigma(x,y)}} \ln \left(m^2 \sigma_{0^+}(x,y)\right) + w(x,y)$$
where $w$ is smooth. The result holds also for $m=0$ and in that case, only the first term in the right-hand side does not vanish in the expansion above. Similarly, quasifree states {\em invariant under  the  symmetries generated by a timelike Killing vector field $\zeta$}  as the states  considered in Sect. \ref{existenceqfs} (with all the hypotheses specified therein) are Hadamard \cite{fnw,wald94} if the spacetime admits spacelike Cauchy surfaces normal to $\zeta$, that is if the spacetime is {\em static}.  This last condition is essential because there 
are spacetimes admitting timelike Killing vectors  but not spacelike Cauchy surfaces normal to them which do not admit {\em invariant Hadamard} quasifree  states,  like Kerr spacetime and Schwartzschild-de~Sitter spacetime \cite{kw}.

{\bf (4)} Referring to the literature before the cornerstone  results \cite{ra1,ra2} (we consider in Sec. \ref{secmicrohad}),  Definition \ref{def_HadamardFormScalar} properly refers to {\em locally} Hadamard states. This is because there also exists a notion of {\em global Hadamard state} (Definition 3.4 in \cite{ra1}), discussed in \cite{kw} in a completely rigorous way for the first time. This apparently more restrictive global condition  essentially requires  (see \cite{kw,ra1} for the numerous technical  details), for a certain open neighborhood ${\cal N}$ of a Cauchy surface of $\Mb$
such that  $\sigma(x,y)$ is always well defined if $(x,y) \in {\cal N}\times {\cal N}$ (and this neighbourhood can always be constructed independently from the Hadamard requirement), that (\ref{hadamardN}) is valid producing the known singularity for causally related arguments, and there are no further singularities for arbitrarily far, spacelike separated, arguments $(x,y) \in {\cal N}\times {\cal N}$.  In this regard a technically  important result, proved in the appendix B of \cite{kw}, is that, analogous to Proposition~\ref{prp:mink-smooth} in the case of Minkowski space, 
\begin{equation}
\Mc \ni x \mapsto \omega(\phi(f) \phi(x)) = \overline{\omega(\phi(x)\phi(f))}
	\in C^\infty(\Mc)
 \label{Fomega}
\end{equation} 
if $f \in C_0^\infty(\Mc)$ and $\omega$ is a quasifree {\em globally} Hadamard state on $\Ac(\Mb)$. We shall prove this fact later using the microlocal approach.
This fact has an important consequence we shall prove later using the microlocal approach: if $\omega$
and $\omega'$ are (locally) Hadamard states, then $\Mc \times \Mc \ni (x,y) \mapsto \omega_2(x,y) -\omega'_2(x,y)$ is smooth. This fact is far from obvious,  since  Definition \ref{def_HadamardFormScalar} guarantees only that the difference is smooth when $x$ and $y$ belong to the same  sufficiently small neighborhood.\\
 An important feature of the global Hadamard condition for a quasifree Hadamard state is that it {\em propagates} \cite{fsw,wald94}: If it holds in a neighborhood of a Cauchy surface it holds in a neighborhood of any other Cauchy surface. We shall come back later to this property making use of the local notion only. This fact, together with the last comment in (3) proves that quasifree Hadamard states for massive fields (and $\xi =0$) exist in globally hyperbolic spacetimes by means of a deformation argument similar to the one exploited in Sect. \ref{existenceqfs}.\\
 We shall not insist on the distinction between the {\em global} and the {\em local}  Hadamard property because, in \cite{ra2}, it was established  that a local  Hadamard state on $\Ac(\Mb)$ is also a global one (the converse is automatic).
It was done  exploiting the {\em microlocal approach}, which we shall discuss shortly.

{\bf (5)} It is possible to prove that \cite{wald94} if a globally hyperbolic spacetime has one (and thus all) compact Cauchy surface, all pure quasifree Hadamard states for the massive KG field (with $\xi=0$) are unitarily equivalent.  
However it is not sufficent to deal with folia of only \emph{pure} quasifree Hadamard states as this excludes very significant examples. Consider the massive KG field (with $\xi=0$) on an ultrastatic spacetime with a compact Cauchy surface. Both the unique time-translation invariant pure state and any {\em thermal} (KMS) state with temperature $T>0$ are Hadamard, but they are not unitarily equivalent, since the former is pure while the latter is not.  There is, in fact, a more general result \cite{Verch} (actually stated in terms of Weyl algebras). Consider an open  region $O$ which  defines a  globally hyperbolic spacetime $\Ob$ in its own right, in a  globally hyperbolic spacetime $\Mb$, such that $\overline{O}$ is compact, and a pair of quasifree Hadamard states $\omega_1,\omega_2$ for the massive KG field ($\xi=0)$
on $\Ac(\Mb)$. It is possible to prove that the restriction to $\Ac(\Ob)\subset \Ac(\Mb)$ of  any density matrix state associated to the GNS construction of $\omega_1$ coincides with the restriction to  $\Ac(\Ob)$ of some  density matrix state associated to the GNS construction of $\omega_2$. 
\end{remark}

\noindent  It is now  possible to recast all the content of Sect. \ref{SectEnlM} in a generic globally hyperbolic spacetime $\Mb$
enlarging the algebra of observables $\Ac(\Mb)$, at the level of quadratic forms, defining the  expectation values  of Wick monomials $:\spa\phi^n\spa:(f)$ with respect to Hadamard states $\omega$ or vector states $\Psi \in {\cal D}_\omega$ with $\omega$  Hadamard. Remarkably, all of that  can be done simultaneously for all states in the said class without picking out any reference state. This is the first step for a completely {\em local} and {\em covariant} definition.
First, define for smooth functions $f_k$ supported in a convex normal neighborhood $C$ 
\begin{equation}
:\spa \phi(f_1) \cdots \phi(f_{n}) \spa:_H   \stackrel {\mbox{\scriptsize  def}} {=} 
\int_{\Mc^n} :\spa \phi(x_1) \cdots \phi(x_n) \spa:_H   f_1(x_1) \ldots f_n(x_n) \:\dvol_{\Mb^n}(x_1,\ldots, x_n)\:, \label{phinH}
\end{equation}
where we have defined the {\em completely symmetrized} formal kernels,
\begin{equation}:\spa \phi(x_1) \cdots \phi(x_n) \spa:_H   \:   \stackrel {\mbox{\scriptsize  def}} {=} \left. \frac{1}{i^n}\frac{\delta^n}{\delta f(x_1)\cdots \delta f(x_n)} \right|_{f=0}e^{i \phi(f) +\frac{1}{2}H_{0^+}(f,f)} \label{formTH}\:.
\end{equation}
Notice that $H_{0^+}$ can be replaced with its symmetric part $H^S_{0^+}$ and that, in (\ref{phinH}), only the symmetric part of the product $f_1(x_1) \ldots f_n(x_n)$ produces a contribution to the left-hand side.
Equivalently, these monomials regularized with respect to the Hadamard parametrix can be define recursively as
\begin{align}
:\spa \phi(f_1)\spa:_H \: &   \stackrel {\mbox{\scriptsize  def}} {=} \:\phi(f_1)\nonumber \\
:\spa \phi(f_1) \cdots \phi(f_{n+1}) \spa:_H \: &   \stackrel {\mbox{\scriptsize  def}} {=} \: :\spa \phi(f_1) \cdots \phi(f_n) \spa:_H  \phi(f_{n+1})\nonumber \\
 &\:\: \:\:\quad  - \sum_{l=1}^n :\spa \phi(f_1) \cdots  \widetilde{\phi(f_l)} \cdots \phi(f_n) \spa:_H \: \tilde{H}(f_l,f_{n+1})\:, \label{phin-H} \\
 & \quad \text{where $\tilde{H} = H^S_{0^+} + \frac{i}{2}E$,}\nonumber
\end{align}
in analogy with the relation between Equations~\eqref{phin} and~\eqref{formT}.
Now consider a quasifree Hadamard  state $\omega$ and indicate by $\omega_\Psi$ the generic state 
indexed by the normalized vector $\Psi \in {\cal D}_\omega$ (so that $\omega = \omega_{\Psi}$ when 
$\Psi$ is the Fock vacuum). 
By induction, it is  possible to prove that, for $n\geq 2$,
there is a {\em jointly  smooth kernel}  $$\omega_\Psi(:\spa\phi(x_1) \cdots \phi(x_n) \spa:_H  )$$  which produces $\omega_\Psi(:\spa\phi(f_1) \cdots \phi(f_n) \spa:_H  ) $ by integration when the supports of the functions $f_k$ belong to $C$.\\
Exactly as for the Minkowski vacuum representation, this result  arises from (\ref{phin-H}) as a consequence of the following list of facts:

(a) $\omega$ is quasifree so that Definition \ref{defqfs} can be used to compute the said kernels, 

(b) $\Psi \in {\cal D}_{\omega}$ so that the expansion (\ref{Psiexp}) can be used, 

(c) the functions in (\ref{Fomega}) are smooth (see (4) in Remark \ref{remHadpar}  and section \ref{secmicrohad}),

(d) the local singularity of two-point functions of quasifree Hadamard states is the same as the one of $H_{0^+}$.

\noindent Consider a normalized $\Psi \in {\cal D}_\omega$, given without loss of generality by
\begin{equation}
	\Psi =  \sum_{n \geq 0,  i_1, \ldots, i_n \geq 1 } C^{(n)}_{i_1 \ldots i_n} \hat{\phi}_\omega(f^{(n)}_{i_1}) \cdots  \hat{\phi}_\omega(f^{(n)}_{i_n}) \Psi_{\omega}\label{PsiexpN} \:,
\end{equation} where only a finite number of coefficients $C^{(n)}_{i_1 \ldots i_n} \in \mathbb C$ is non-vanishing, which defines the algebraic state $\omega_\Psi(\cdot) = \langle \Psi | (\cdot) \Psi \rangle$. Then, for instance, with the same argument used to achieve (\ref{expansion3}) we have
\begin{equation}\label{expansion2H}
	\omega_\Psi(:\spa\phi(x_1) \phi(x_2)\spa:_H)
	- \omega(\phi(x_1)\phi(x_2)) + \tilde{H}(x_1,x_2) \in C^\infty(\Mc\times \Mc) \: ,
\end{equation}
where the smoothness is assured because the resulting expression consists of a linear combination of products like $\omega(\phi(x_1)\phi(g)) \omega(\phi(f) \phi(x_2))$, with some test functions $f$ and $g$. Note that the combination of the second and third terms in~\eqref{expansion2H} can be rewritten as
\begin{align*}
	\omega(\phi(x_1)\phi(x_2)) - \tilde{H}(x_1,x_2)
	&= \omega(\phi(x_1)\phi(x_2)) - H^S_{0^+}(x_1,x_2) - \frac{i}{2} E(x_1,x_2) \\
	&= \frac{1}{2} \omega(\phi(x_1)\phi(x_2)) - H_{0^+}(x_1,x_2) \\
	& \quad {}
		+ \frac{1}{2} \omega(\phi(x_2)\phi(x_1)) - H_{0^+}(x_2,x_1) ,
\end{align*}
which is obviously smooth by the very definition of the Hadamard property of $\omega$. Hence $\omega_\Psi(:\spa\phi(x_1) \phi(x_2)\spa:_H)$ is also smooth.
We are in a position to define the expectation values of the Wick monomials for $f \in C_0^\infty(\Mc)$
such that its support is included in $C$,
\begin{equation}\omega_\Psi(:\spa \phi^n\spa:_H (f))=
\int_{\Mc^n}\omega_\Psi(:\spa \phi(x_1) \cdots \phi(x_n) \spa:_H )f(x_1) \delta(x_1,\ldots, x_n) \dvol_{\Mb^n}\label{phinHN}\end{equation}
Exactly as before, polarization extends the definition to a quadratic form on ${\cal D}_{\omega}\times {\cal D}_{\omega}$. There is no guarantee that operators fitting these quadratic forms really exist. The question of their existence as operators will be addressed later, in Section~\ref{sec:alg-wick}.

\begin{remark}$\label{remhad3}\null$ 

{\bf (1)} The restriction on the support of $f$ is not very severe. The restriction  can be removed making use of  a partition of unity (see for example \cite{stress,HW04} referring to more generally differentiated Wick polynomials).

{\bf (2)} The given definition of $\omega(:\spa \phi^n\spa:_H (f))$ is affected by several ambiguities due to the effective construction of $H_\epsilon$. A complete classification of these ambiguities, promoting Wick polynomials to properly defined elements of a $*$-algebra,  can be presented  from a very general  viewpoint,  adopting a {\em locally covariant framework} \cite{hw,km}, we shall not consider in this introductory review  (see \cite{chapt:CR}). We only say that these ambiguities are completely  described by a class of scalar polynomials in the mass and Riemann  curvature tensor and their covariant derivatives. The finite order of these polynomials is fixed by scaling properties of Wick polynomials.
The coefficients of the polynomials are smooth functions of the parameter $\xi$. 
 We stress that this classification is the first step of the ultraviolet renormalization program which, in curved spacetime and differently from flat spacetime where all curvature vanish, starts with classifying the finite renormalization counterterms of Wick polynomials instead of only dealing with time-ordered Wick polynomials.

{\bf (3)} Easily extending the said definition, using the fact that $\omega_\Psi(:\spa \phi(x_1)\phi(x_2) \spa:_H )$ is smooth and thus can be differentiated,
 one can define a notion of {\em differentiated Wick polynomials} which include, in particular, the {\em stress energy tensor} as a Hermitian quadratic form evaluated on Hadamard states or vector states in the dense subspace ${\cal D}_\omega$ in the GNS Hilbert space of a Hadamard state $\omega$. This would be enough to 
implement the computation of the back reaction of the quantum matter in a given state to the geometry of the spacetime  through (\ref{backreaction}) especially in cosmological scenario (see \cite{chapt:HP}). This program has actually been initiated much earlier than the algebraic approach was adopted in QFT in curved spacetime \cite{BirrellDavies} and the notion of Hadamard state was invented, through several steps,  in this context. The requirements  a physically sensible object  $\omega(:\spa T_{ab}\spa:_H (x))$ should satisfy was clearly discussed by several authors, Wald in particular (see \cite{wald94} for a complete account and \cite{HackMoretti} for more recent survey). The most puzzling  issue in this context perhaps concerns the interplay of the conservation requirement $\nabla_a \omega(:\spa T^{ab}\spa:_H (x))=0$ and the appearance  of the {\em trace anomaly}.  We shall come back to these issues later, at the end of Section~\ref{sec:alg-wick}.
\end{remark}

\subsection{The notion of wavefront set and its elementary properties}
Microlocal analysis permits us to completely reformulate the theory of Hadamard states into a much more powerful formulation where, in particular, the Wick polynomials can be defined as proper operators and not only Hermitian quadratic forms.

Following \cite{Thomas,Strohmaier}, let us start be introducing the notion of  wavefront set. To motivate it, let us recall that a smooth function  on $\mathbb R^m$ with compact support has a rapidly decreasing Fourier transform. If we take a distribution $u$ in ${\cal D}^\prime(\mathbb R^m)$ and multiply it by an $f\in{\cal D}(\mathbb R^m)$ with $f(x_0)\neq 0$, then $uf$ is an element of ${\cal E}^\prime(\mathbb R^m)$, {\it i.e.}, a distribution with compact support. If $fu$ were smooth, then its Fourier transform $\widehat{fu}$ would be smooth and rapidly decreasing (with all its derivatives). The failure of $fu$ to be smooth in a neighbourhood of $x_0$ can therefore be quantitatively described by the set of directions in Fourier space%
	\footnote{Our convention for the Fourier transform is so that $f(x) = \frac{1}{(2\pi)^m}\int e^{-ikx} \hat{f}(k)\, d^m k$. This convention agrees with those of~\cite{Thomas,ra1,ra2}, but has the opposite sign in the exponential with respect to~\cite{Strohmaier}. This means that our wavefont sets need to be negated to be compared to those of~\cite{Strohmaier}. Fortunately, in all cases where this is done, the wavefront sets happen to be negation symmetric.} %
where $\widehat{fu}$ is not rapidly decreasing. Of course it could happen that we choose $f$ badly and therefore `cut' some of the singularities of $u$ at $x_0$. To see the full singularity structure of $u$ at $x_0$, we therefore need to consider all 
 test functions which are non-vanishing at $x_0$. With this in mind, one first defines the wavefront set of distributions on (open subsets of) $\mathbb R^m$ and then extends it to curved manifolds in a second step.\\

\noindent In the rest of the chapter ${\cal D}(\Mc)  \stackrel {\mbox{\scriptsize  def}} {=} C_0^\infty(\Mc, \mathbb C)$ for every smooth manifold $\Mc$.
An open neighbourhood $G$ of $k_0\in\mathbb R^m$ is called {\bf conic} if $k\in G$ implies $\lambda k\in G$ for all $\lambda >0$. 

\begin{definition}[Wavefront set]
 \label{def_WaveFrontSet}
Let $u\in{\cal D}^\prime(U)$, with open $U\subset \mathbb R^m$.
A point $(x_0,k_0)\in U\times (\mathbb R^m\setminus\{0\})$ is called a {\bf regular directed} point of $u$ if there is $f\in{\cal D}(U)$ with $f(x_0)\neq 0$ such that, for every $n\in\mathbb N$, there is a constant $C_n\geq 0$ fulfilling
$$|\widehat{fu}(k)|\le C_n (1+|k|)^{-n}$$
for all $k$ in an open conic neighbourhood of $k_0$. The {\bf wavefront set} $WF(u)$, of $u \in {\cal D}'(U)$
  is the complement in $U \times (\mathbb R^m\setminus\{0\})$ of the set of all regular directed points of $u$.
\end{definition}

\begin{remark}
Obviously, if $u,v \in {\cal D}'(U)$  the wavefront set is not additive and, in general, one simply has $WF(u+v) \subset WF(u) \cup WF(v)$.
\end{remark}

As, an elementary example, let us consider the wavefront set of the distribution $\delta_y(x) = \delta(x-y)$ on $\mathbb{R}^n$~\cite[p.103]{Strohmaier}:
\begin{equation}\label{eq:wf-delta1}
	WF(\delta_y) = \{ (y,k_y) \in T^*\mathbb{R}^n \mid k_y \ne 0 \} .
\end{equation}

\noindent If $U \subset \mathbb R^m$ is an open and non-empty subset,  $T^*U$ is naturally identified with $U \times \mathbb R^m$. In the rest of the chapter  $T^*U \setminus 0    \stackrel {\mbox{\scriptsize  def}} {=} \{(x,p) \in T^*U \:|\: p \neq 0\}$. \\
If $U \subset \mathbb R^m$ is an open non-empty set,  $\Gamma \subset T^*U \setminus 0$ is a {\bf cone}
when   $(x,\lambda k) \in \Gamma$ if $(x,k) \in \Gamma$ and $\lambda>0$.
If the mentioned  cone $\Gamma$  is closed {\em in   the topology of} $T^*U\setminus 0$, we define 
$${\cal D}_\Gamma'  \stackrel {\mbox{\scriptsize  def}} {=} \{ u \in {\cal D}'(U)\:|\: WF(u) \subset \Gamma\}\:.$$ 

\begin{remark}
All these  definitions can be restated for the case of  $U$ replaced with a general smooth manifold and we shall exploit this opportunity shortly.
\end{remark}
We are in a position  to define a relevant notion of convergence \cite{hormander}.

\begin{definition}[Convergence in H\"ormander pseudotopology] 
If $u_j \in {\cal D}_\Gamma'(U)$ is a sequence and $u \in {\cal D}_\Gamma'(U)$, we write $u_j \to u$ in ${\cal D}_{\Gamma}'(U)$ if both the conditions below hold.

(i) $u_j \to u$ weakly in ${\cal D}'(U)$ as $j \to +\infty$,

(ii) $\sup_j \sup_V |p|^N |\widehat{\phi u_j}(p)| <\infty$, $N=1,2, \ldots$, if $\phi \in {\cal D}(U)$ and $V \subset T^*U$ is any  closed cone, whose projection on $U$ is $\supp(\phi)$,
such that $\Gamma \cap V  = \emptyset$. \\
 In this case, we say that $u_j$  converges to $u$ in the {\bf H\"ormander pseudotopology}. 
\end{definition}
It turns out that test functions (whose wavefront set is always empty as said below) are dense even with respect to that notion of convergence \cite{hormander}.
\begin{proposition}
 If $u \in {\cal D}'_\Gamma(U)$, there is a sequence of smooth functions $u_j \in {\cal D}(U)$ such that $u_j \to u$ in $ {\cal D}'_\Gamma(U)$.
\end{proposition}

\noindent Let us immediately state a few elementary  properties of wavefront sets \cite{hormander1, hormander,Strohmaier,FJ}.  We remind the reader that $x\in U$ is a {\bf regular point} of a distribution $u\in{\cal D}^\prime(U)$ if there is an open neighborhood $O\subset U$ of $x$ such that $\langle u, f \rangle = \langle h_u, f\rangle$ for some $h_u\in {\cal D}(U)$ and every $f \in {\cal D}(U)$ supported in $O$. The closure of the complement of the set of regular points is the {\bf singular support} of $u$ by definition.

\begin{theorem}[Elementary properties of $WF$]\label{WFelem}
 \label{thm_PropertiesWavefront} Let $u\in{\cal D}^\prime(U)$, $U\subset \mathbb R^m$ open and non-empty.

{\bf (a)}  $u$ is smooth if and only if $WF(u)$ is empty. More precisely, the {\em singular support} of $u$ is the projection of $WF(u)$ on $\mathbb R^m$.

{\bf (b)}  If $P$ is a partial differential operator on $U$  with smooth coefficients: $$WF(Pu)\subset WF(u)\,.$$

{\bf (c)}  Let  $V\subset\mathbb R^m$ be an open set and let $\chi:V\to U$ be a diffeomorphism. The pull-back $\chi^*u \in {\mathcal D}'(V)$ of $u$ defined by $\chi^*u(f)=u(\chi_*f)$ for all $f\in{\cal D}(V)$ fulfils
$$WF(\chi^*u)=\chi^*WF(u)  \stackrel {\mbox{\scriptsize  def}} {=} \left\{(\chi^{-1}(x),\chi^*k)\;|\;(x,k)\in WF(u)\right\}\,,$$
where $\chi^*k$ denotes the pull-back of $\chi$ in the sense of {\em cotangent vectors}.

{\bf (d)} Let $V\subset \mathbb R^n$ be an open set and $v \in {\cal D}'(V)$, then $WF(u\otimes v)$ is included in
$$(WF(u) \times WF(v))\cup ((\supp u \times \{0\})\times WF(v)) \cup 
(WF(u)\times (\supp v \times \{0\})) \:.$$

{\bf (e)} Let $V\subset \mathbb R^n$,  $K \in {\cal D}'(U\times V)$ and   $f \in {\cal D}(V)$, then
$$WF({\cal K}f) \subset \{(x, p)\in TU \setminus 0 \:|\:  (x,y,p, 0) \in WF(K) \:\:\mbox{for some}\: \:y \in \supp(f) \}\:,$$ 
where  ${\cal K} :{\cal D}(V) \mapsto {\cal D}'(U)$ is the continuous linear map associated to $K$
in view of Schwartz kernel theorem.
\end{theorem}

The result (e), with a suitably improved  statement,  can be extended to to the case of $f$ replaced by a distribution \cite{hormander}.\\
From (c) we conclude that  the wavefront set transforms covariantly under diffeomorphisms as a subset of $T^*U$, with $U$ an open subset of $\mathbb R^m$. Therefore we can immediately extend the definition of $WF$ to distributions on a  manifold $\Mc$ simply by patching together wavefront sets in different coordinate patches of $\Mc$ with the help of a partition of unity. As a result, for $u\in{\cal D}^\prime(\Mc)$, $WF(u)\subset T^*\Mc\setminus 0$. Also the notion of convergence in the H\"ormander pseudotopology easily extends to manifolds. All the statements of theorem \ref{WFelem} extend to the case where $U$ and $V$ are smooth manifolds.

\begin{figure}
\begin{center}
	\includegraphics{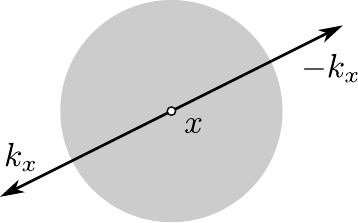}
\end{center}
\caption{Wavefront set of $\delta(x,y)$ on $\Mc\times \Mc$, defined in~\eqref{eq:wf-delta}, consists of points of the form $(x,x,k_x,-k_x)$, $(x,k_x) \in T^*\Mc\setminus 0$.}
\label{fig:wf-delta}
\end{figure}

Following up on~\eqref{eq:wf-delta1}, an elementary example of a distribution on a manifold is $\delta(x,y)$ defined on $\Mc\times \Mc$. Its wavefront set is (Figure~\ref{fig:wf-delta})
\begin{equation}\label{eq:wf-delta}
	WF(\delta) = \{ (x,x,k_x,-k_x) \in T^*\Mc^2\setminus 0 \mid (x,k_x) \in T^*\Mc \setminus 0 \} \: .
\end{equation}
The necessity of the sign reversal in the covector $-k_x$ corresponding to the second copy of $\Mc$ can be seen from the formula $\delta(x,y) = \delta(x-y)$ on $\mathbb{R}^n$. 

To conclude this very short survey,  we wish to stress some remarkable results of wavefront set technology respectively concerning  (a) the theorem of  {\em propagation of singularities}, (b) the {\em product of distributions}, (c) {\em composition of kernels}.

Let us start with  an elementary version of the celebrated {\em theorem of propagation of singularities} formulated as in \cite{Strohmaier}.

\begin{remark}\label{remarkps}$\null$

{\bf (1)} Let us remind the reader that if, in local coordinates, $P= \sum_{|\alpha| \leq m} a_\alpha(x) \partial^\alpha$ is a differential operator of order $m\ge 1$ (it is assumed that $a_\alpha\neq 0$ for some $\alpha$ with $|\alpha|=m$)   on a manifold $\Mc$, where $a$ is a {\em multi-index} \cite{hormander}, and $a_\alpha$ are smooth coefficients,  then the polynomial $\sigma_P(x,p) = \sum_{|\alpha| =  m} a_\alpha(x) (ip)^\alpha$ is called the {\bf principal symbol} of $P$. It is possible to prove that $(x, \xi) \mapsto \sigma_P(x,p)$ determines a  well defined function on $T^*\Mc$ which, in general is complex valued.   The {\bf characteristic set} of $P$, indicated by
$char(P) \subset T^*\Mc\setminus 0$, denotes the set of zeros of $\sigma_P$ made of {\em non-vanishing} covectors.
The principal symbol  $\sigma_P$ can be used as a {\em Hamiltonian function} on $T^*\Mc$ and the maximal solutions of Hamilton equations define the {\bf local flow} of $\sigma_P$ on $T^*\Mc$.   

{\bf (2)} 
The principal symbol of the Klein-Gordon operator is  $-g^{ab}(x) p_ap_b$. 
It is an easy exercise \cite{Strohmaier} to prove that if $\Mb$ is a Lorentzian manifold and 
$P$ is a {\bf normally hyperbolic operator}, i.e., the principal symbol is the same as the one of Klein-Gordon operator, then the integral curves of  the local flow of $\sigma_P$ are nothing but the lift to $T^*\Mc$
of the  {\em geodesics} of the metric $g$ parametrized by an {\em affine parameter}.
Finally,  $char(P) = \{(x,p) \in T^*\Mc \setminus 0 \:|\: g^{ab}(x) p_ap_b =0\}$  
\end{remark}

\begin{theorem}[Microlocal regularity and propagation of singularities]\label{tps}
Let $P$ be a differential operator on a manifold $\Mc$ whose principal symbol is real valued, if $u,f \in {\cal D}'(\Mc)$ are such that $Pu=f$ then the following facts hold.

{\bf (a)}  $WF(u) \subset char(P) \cup WF(f)$,

{\bf (b)} $WF(u)\setminus Wf(f)$ is invariant under the local flow of $\sigma_P$ on  $T^*\Mc \setminus WF(f)$.  
\end{theorem}
Let us conclude with the famous H\"ormander definition of {\em product of distributions} \cite{hormander1,hormander}. We need a preliminary definition.
If $\Gamma_1, \Gamma_2 \subset T^*\Mc\setminus 0 $ are closed cones,
$$\Gamma_1+ \Gamma_2  \stackrel {\mbox{\scriptsize  def}} {=} \left\{(x, k_1+k_2) \subset T^*\Mc \;|\;(x,k_1)\in \Gamma_1,\; (x,k_2)\in \Gamma_2\:\: \mbox{for some $x\in \Mc$}\right\}.$$

\begin{theorem}[Product of distributions]\label{teoprod}
Consider  a pair of closed cones $\Gamma_1, \Gamma_2  \subset  T^*\Mc \setminus 0$. If 
$$\Gamma_1 + \Gamma_2 \not \ni (x, 0) \quad \mbox{for all $x \in \Mc$,}$$ then there is a unique bilinear  map,  the {\bf product} of $u_1$ and $u_2$,
$${\cal D}'_{\Gamma_1} \times  {\cal D}'_{\Gamma_2} \ni (u_1,u_2) \mapsto u_1u_2 \in {\cal D}'(\Mc),$$ such that

(i)  it reduces to the standard pointwise product  if $u_1, u_2 \in {\cal D}(\Mc)$,

(ii) it is jointly sequentially continuous in the H\"ormander pseudotopology: If $u_j^{(n)} \to u_j$
in $D_{\Gamma_j}(\Mc)$ for $j=1,2$ then $u^{(n)}_1u^{(n)}_2 \to u_1u_2$ in ${\cal D}_{\Gamma}(\Mc)$, where  $\Gamma$ is a closed cone in $T^*\Mc\setminus 0$ defined as  $\Gamma   \stackrel {\mbox{\scriptsize  def}} {=}  \Gamma_1\cup \Gamma_2\cup \left(\Gamma_1\oplus\Gamma_2\right)$. \\
In particular the following bound always holds if the above product is defined:  
\begin{eqnarray} WF(u_1u_2)\subset  \Gamma_1\cup \Gamma_2\cup \left(\Gamma_1+\Gamma_2\right)\,. \label{boundWF}\end{eqnarray}
\end{theorem}

From the examples~\eqref{eq:wf-delta1} and~\eqref{eq:wf-delta} and the simple observation that
\begin{equation}
	\mathbb{R}^n\setminus \{0\} + \mathbb{R}^n\setminus \{0\}
	= \mathbb{R}^n \ni 0,
\end{equation}
it is clear that the multiplication of two $\delta$-functions with overlapping supports, as is to be expected, does not satisfy the above conditions.

Let us come to the last theorem concerning the composition of distributional kernels.
Let $X,Y$ be smooth manifolds. If $K \in {\cal D}'(X \times Y)$, the continuous map  associated to $K$ by the Schwartz kernel theorem  will be denoted  by ${\cal K}: {\cal D}(Y) \to {\cal D}'(X)$. We shall also adopt 
the following standard notations:
\begin{align}WF(K)_X &  \stackrel {\mbox{\scriptsize  def}} {=} \{(x,p) \:|\: (x,y, p,0) \in WF(K) \quad \mbox{for some $y\in Y$}\}\nonumber\:, \\
WF(K)_Y &  \stackrel {\mbox{\scriptsize  def}} {=} \{(y,q) \:|\: (x,y,0,q) \in WF(K) \quad \mbox{for some $x\in X$}\}\nonumber\:, \\
WF'(K) &  \stackrel {\mbox{\scriptsize  def}} {=} \{(x,y,p,q) \:|\: (x,y,p,-q) \in WF(K)\}\nonumber\:, \\
WF'(K)_Y &  \stackrel {\mbox{\scriptsize  def}} {=} \{(y,q) \:|\: (x,y,0,-q) \in WF(K) \quad \mbox{for some $x\in X$}\}\:.\nonumber 
\end{align}

\begin{theorem}[Composition of kernels]\label{thm:kern}
Consider three smooth manifolds $X, Y, Z$ and  $K_1 \in {\cal D}'(X\times Y)$, $K_2\in {\cal D}'(Y \times Z)$.
If $WF'(K_1)_Y \cap WF(K_2)_Y = \emptyset$  and the projection $$\supp K_2 \ni (y,z) \mapsto z \in Z$$ is proper
(that is, the inverse of a compact set is compact), then the composition ${\cal K}_1 \circ {\cal K}_2$ is well defined, giving rise to $K\in {\cal D}'(X,Z)$,  and reduces to the standard one when the kernel are smooth. It finally holds (the symbol $\circ$ denoting the composition of relations)
\begin{multline} 
	WF'(K) \subset  WF'(K_1) \circ WF'(K_2)
	 \cup (WF(K_1)_X \times Z \times \{0\}) \\ 
	{} \cup  (X \times \{0\} \times WF'(K_2)_Z) \:.
\end{multline} 
\end{theorem}

Comparing with~\eqref{eq:wf-delta}, note that $WF'(\delta)$ is the diagonal subset $\Delta \subset T^*\Mc\times T^*\Mc$. In the composition of relations, $\Delta$ acts as an identity, which is consistent with the above theorem and the fact that $\delta(x,y)$ acts as an identity for the composition of distributional kernels.

\subsection{Microlocal reformulation}\label{secmicrohad}
Let us focus again on the two-point function of Minkowski quasifree vacuum state. Form (\ref{W}) we see that the singular support of $\omega_{\mathbb M2}(x,y)$ is  the set 
of couples $(x,y) \in \mathbb M \times \mathbb M$ such that $x-y$ is {\em light like}. From  (a) in theorem \ref{WFelem}, we conclude that $WF(\omega_{\mathbb M2})$ must project onto this set. On the other hand
 (\ref{fourieromega}) can be re-written as
\begin{eqnarray}
\omega_{{\mathbb M}2}(x,y) = \frac{1}{(2\pi)^3} \int_{\mathbb R^4} e^{-i(px+qy)} \theta(p^0) \delta(p^2 + m^2) \delta(p+q) d^4q d^4p \label{fourieromega2}\:,
\end{eqnarray}
where translational invariance is responsible for the appearance of  $\delta(p+q)$ in (\ref{fourieromega2}).
From this couple of facts, also noticing the presence of $\theta(p^0)$ in the integrand,  one guesses  that  the wavefront set of the Minkowski two-point function must be
\begin{eqnarray}\label{eq_MinkwoskiWF} WF(\omega_{\mathbb M2})=\left\{(x,y,p,-p)\in T^*\Mc^2 \:|\:  p^2=0,\;p \:||\: (x-y),\; p^0>0\right\}\:.\end{eqnarray}
Identity  (\ref{eq_MinkwoskiWF}) is, in fact,  correct and holds true also for $m=0$ \cite{rs2}.
The condition $p^0>0$ encodes the {\em energy positivity} of the Minkowski vacuum state. Notice that the couples $(x,y) \in \mathbb M \times \mathbb M$ giving contribution to the wavefront set are always connected by a {\em light-like geodesic} co-tangent to $p$. For $x=y$ there are infinitely many such geodesics, if we allow ourselves to consider zero length curves (consisting of a single point) with a given tangent vector.
 
The structure (\ref{eq_MinkwoskiWF})  of the wavefront set of the two-point function of Minkowski vacuum 
is a particular case of the general notion of a Hadamard state. We re-adapt here
the content of the cornerstone papers \cite{ra1,ra2}  to our formulation. We note that we do not make use of the global 
Hadamard condition  (see (4) in Remark \ref{remHadpar}).
 The following theorem collects various results of \cite{ra1,ra2}.

\begin{figure}
\begin{center}
	\includegraphics{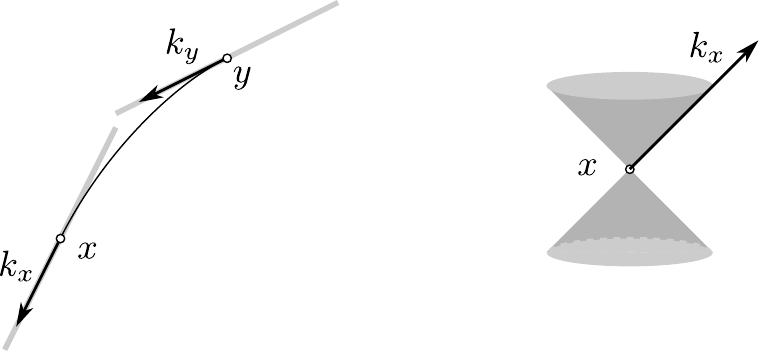}
\end{center}
\caption{The null geodesic relation $(x,k_x) \sim (x,k_y)$ defined in Theorem~\ref{radtheorem}. The points $x$ and $y$ must be linked by a null geodesic, the covectors $k_x$ and $k_y$ must be parallel transported images of each other and both covectors must be coparallel, all with respect to the same null geodesic. Any causal ordering between $x$ and $y$ is admissible. Also, $k_x$, $-k_x$ and $\lambda k_x$ ($\lambda \ne 0$) are all considered coparallel to the same geodesic. In the coincident case, $x=y$, we agree that there are infinitely many (zero-length) null geodesics joining $x$ to itself, corresponding to different non-vanishing null covectors $k_x\in T^*_x\Mc$.}
\label{fig:rel-geod}
\end{figure}

\begin{figure}
\begin{center}
	\includegraphics{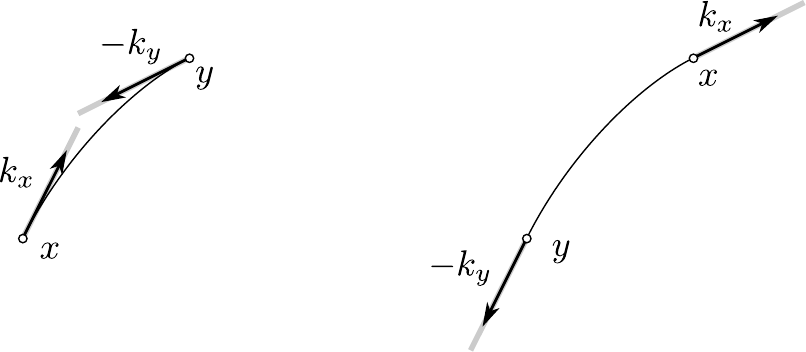}
\end{center}
\caption{The Hadamard form $\mathcal{H}$ of a wavefront set, as defined in Theorem~\ref{radtheorem}. It consists of a subset of points $(x,y,k_x,-k_y) \in T^*\Mc^2$, where $(x,k_x) \sim (y,k_y)$ are linked but the null geodesic relation (Figure~\ref{fig:rel-geod}). The restriction is that $k_x\triangleright 0$, meaning that $k_x(v) \ge 0$ for any future-directed $v\in T_x\Mc$. We illustrate the two possible causal orderings $x \in J^-(y)$ and $x \in J^+(y)$.}
\label{fig:wf-hadamard}
\end{figure}

\begin{theorem}[``Radzikowski theorem'']\label{radtheorem}
 \label{theorad}
For a $4$-dimensional globally hyperbolic (time oriented) spacetime $\Mb$ and referring to the unital $*$-algebra of Klein-Gordon quantum field $\Ac(\Mb)$ with $m^2, \xi \in \mathbb R$ arbitrarily fixed, let $\omega$ be a state on $\Ac(\Mb)$, not necessarily quasifree. \\
{\bf (a)} The following statements are equivalent,\\

 (i)  $\omega$ is Hadamard in the sense of Def. \ref{def_HadamardFormScalar},\\

(ii) the wavefront set of  the two-point function $\omega_2$ has the {\bf Hadamard form} on $\Mb$ or equivalently, it satisfies the {\bf microlocal spectrum condition} on $\Mb$:
\begin{equation}
	WF(\omega_2)=\left\{(x,y,k_x,-k_y)\in T^* \Mc^2\setminus 0\;|\;
(x,k_x)\sim(y,k_y),\;k_x\triangleright 0\right\}\, \stackrel {\mbox{\scriptsize  def}}{=} {\mathcal H}.
\end{equation}
Here, $(x,k_x)\sim(y,k_y)$ means  that there exists a null geodesic $\gamma$ connecting $x$ to $y$ such
that $k_x$ is coparallel and cotangent to $\gamma$ at $x$ and $k_y$ is the parallel transport of $k_x$ from $x$ to $y$ along $\gamma$, Figure~\ref{fig:rel-geod}. $k_x\triangleright 0$ means that $k_x$ does not vanish and is future-directed ($k_x(v) \ge 0$ for all future-directed $v\in T_x\Mc$), Figure~\ref{fig:wf-hadamard}.\\
{\bf (b)} If $\omega'$ is another Hadamard state on  $\Ac(\Mb)$, then $\omega_2-\omega_2' \in C^\infty (\Mc\times \Mc, \mathbb C)$.
\end{theorem}

\begin{proof} 
(a) Suppose that  $\omega$ satisfies  (i),  then it is globally Hadamard in the sense\footnote{Results in \cite{ra1,ra2} are stated for $\xi=0$ in KG operator, however they are generally valid for $m^2$ replaced by a given smooth function,  as specified  at the beginning of p. 533 in \cite{ra1}.} of \cite{ra1} due to Theorem 9.2 in \cite{ra2}. Theorem 5.1 in \cite{ra1} implies that (ii) holds.  Conversely, if (ii) is valid, Theorem 5.1 in \cite{ra1} entails that 
$\omega$ is globally and thus locally Hadamard so that (i) holds true.  (b) immediately arises from Theorem 4.3 in \cite{ra2}. $\Box$
\end{proof}

\begin{figure}
\begin{center}
	\includegraphics{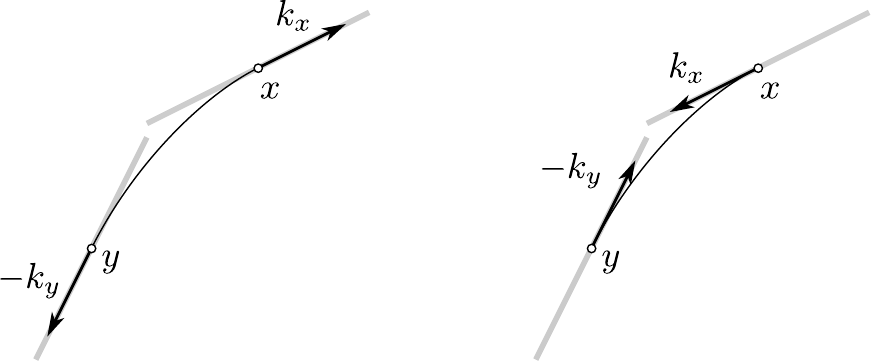}
\end{center}
\caption{The wavefront sets of the retarded fundamental solution $E^+$ of the Klein-Gordon operator, as defined in Proposition~\ref{prp:wf-retadv}, consist of the union of $WF(\delta)$ (Figure~\ref{fig:wf-delta}) and of the points $(x,k_x,y,-k_y)\in T^*\Mc^2$, where $(x,k_x) \sim (y,k_y)$ are linked by the geodesic relation (Figure~\ref{fig:rel-geod}), with the causal precedence condition $x\in J^+(y)$. We illustrate the two cases when $k_x$ is coparallel and anti-coparallel to the future-directed geodesic from $y$ to $x$. The wavefront set of the advanced fundamental solution $E^-$ is defined in the same way, with the exception that we require the causal precedence condition $x\in J^-(y)$ instead.}
\label{fig:wf-retadv}
\end{figure}

It is also helpful to have a characterization of the wavefront set of the retarded and advanced fundamental solutions~\cite{ra1,Strohmaier}.
\begin{proposition}\label{prp:wf-retadv}
The retarded and advanced fundamental solutions of the Klein-Gordon operator $P = \Box_{\Mb}+m^2+\xi R$ on $\Mb$, $E^+,E^-\in {\mathcal D}'(\Mc\times \Mc)$ respectively, have the following wavefront sets (Figure~\ref{fig:wf-retadv}):
\begin{multline}
	WF(E^\pm) = WF(\delta) \\
		{} \cup \left\{ (x,y,k_x,-k_y) \in T^*\Mc^2\setminus 0 \mid
		(x,k_x) \sim (y,k_y), x\in J^\pm(y) \right\} \stackrel {\mbox{\scriptsize  def}}{=} {\mathcal F}_\pm,
\end{multline}
where $\sim$ denotes the same relation as in Theorem~\ref{radtheorem}.
\end{proposition}
With this result and the microlocal  technology previously introduced we can prove some 
remarkable properties of Hadamard states, especially in relation with what was already discussed in (4) in Remark \ref{remHadpar}. The second statement, for $n=4$, implies that the singularity structure of Hadamard states propagates through the spacetime.  

\begin{proposition}\label{prp:hadform}  Consider a state $\omega$ on $\Ac(\Mb)$, with $\omega_2 \in {\cal D}'(\Mc \times \Mc)$, where $\Mb$ is a (time oriented) globally hyperbolic spacetime  with dimension $n\geq 2$. The following facts hold.

{\bf (a)} If $WF(\omega_2)$ has the Hadamard form, then $\Mc \ni x \mapsto \omega_2(x, f)$ is smooth for every $f \in C_0^\infty(\Mc)$.

{\bf(b)} If  $WF(\omega_2\spa\rest_{O\times O})$ has the Hadamard form on $O$,  where $O$ is an open neighborhood  of a smooth spacelike Cauchy surface  $\Sigma$ of $\Mb$, then $WF(\omega_2)$ has the Hadamard form on $\Mb$.
\end{proposition}

\begin{proof} (a) From (e) in Theorem \ref{WFelem} and the Hadamard form of $WF(\omega_2)$ we conclude that $WF(\omega_2(\cdot, f)) = \emptyset$. Next, (a) in Theorem \ref{WFelem} implies the thesis.

(b) The 2-point function $\omega_2(x,y)$ is a bisolution of the Klein-Gordon operator $P = \Box_{\Mb}+m^2+\xi R$, as in~\eqref{KGomega}. So the value of $\omega_2(f,g)$, for $f, g \in C^\infty_0(\Mc)$, depends on the arguments only up to the addition of any term from $P[C^\infty_0(\Mc)]$. In fact, we can choose $h, k \in C^\infty$ such that $\supp (f+P[h])$ and $\supp (g + P[k])$ are both contained in $O$. More precisely, we can define an $S\in {\mathcal D}'(O\times \Mc)$ such that the corresponding operator maps ${\mathcal S}\colon C^\infty_0(\Mc) \to C^\infty_0(O)$ and we have the identity $\omega_2 = {\mathcal S}^t \circ \omega_2 \circ {\mathcal S}$. Then, using the result of Theorem~\ref{thm:kern} on the composition of kernels and the fact that $\omega_2$ has the Hadamard form on $O$, we can show that $\omega_2$ has the Hadamard form on all of $\Mc$.

Consider a smooth partition of unity $\chi_+ + \chi_- = 1$ adapted to the Cauchy surface $\Sigma$. That is, there exist two other Cauchy surfaces, $\Sigma_+$ in the future of $\Sigma$ and $\Sigma_-$ in the past of $\Sigma$, such that $\supp \chi_+ \subset J^+(\Sigma_-)$ and $\supp \chi_- \subset J^-(\Sigma_+)$. Such an adapted partition of unity always exists if $O$ is globally hyperbolic in its own right and, if not, since $\Mb$ is globally hyperbolic, any open neighborhood of $\Sigma$ will contain a possibly smaller neighborhood of $\Sigma$ that is also globally hyperbolic~\cite{BernalSanchez1,BernalSanchez2}.

Let ${\mathcal S}f = f - P[\chi_+E^-f+\chi_-E^+f]$, with the corresponding integral kernel
\begin{equation}
	S(x,y) = \delta(x,y) - P_x[\chi_+(x) E^-(x,y) + \chi_-(x) E^+(x,y)] ,
\end{equation}
where the subscript on $P_x$ means that it is acting only on the $x$ variable. A straight forward calculation shows that $S$ has the desired properties. Multiplication by a smooth function and the application of a differential operator does not increase the wavefront set, hence
\begin{equation}
	WF(S) \subset WF(\delta) \cup WF(E^-) \cup WF(E^+)
\end{equation}
as a subset of $T^*(O \times \Mc)$. The $\delta$-function has the wavefront set
\begin{equation}
	WF(\delta) = \{ (x,x,k_x,-k_x) \in T^*\Mc^2\setminus 0 \mid (x,k_x) \in T^*\Mc \setminus 0 \} \: .
\end{equation}
The wavefront sets ${\mathcal F}_\pm = WF(E^\pm)$ of the retarded and advanced fundamental solutions was given in Proposition~\ref{prp:wf-retadv}. The Hadamard form ${\mathcal H}$ of the wavefront set was defined in Theorem~\ref{radtheorem}. We can now appeal to Theorem~\ref{thm:kern} on the wavefront set of the composition of kernels to how that $WF(\omega_2) = WF({\mathcal S}^t \circ \omega_2 \circ {\mathcal S}) \subset {\mathcal H}$. The first thing to check is that $WF(S)_{\Mc_i}$, $WF(\omega_2)_{\Mc_i}$, $i=1,2$ denoting respectively the first and the second factor in $\Mc\times \Mc$, are all empty, because they contain no element of the form $(x,y,k_x,0)$ or $(x,y,0,k_y)$. Second, due to the hypothesis $WF'(\omega_2)|_O \subset {\mathcal H}'_O$, the symmetry of the composition and the fact that composition with $\delta(x,y)$ leaves any wavefront set invariant, it is sufficient to check that the compositions of wavefront sets as relations satisfy ${\mathcal H}'_O \circ {\mathcal F}'_\pm \subset {\mathcal H}'_\Mc$. 

Consider any $(x,y,k_x,k_y) \in {\mathcal H}'_O$ and $(y,z,k_y,k_z) \in {\mathcal F}'_\pm$, so that $(x,z,k_x,k_z) \in {\mathcal H}'_O \circ {\mathcal F}'_\pm$. Then $(x,k_x) \sim (y,k_y)$ and $(y,k_y) \sim (z,k_z)$ in $\Mc$ according to the relation $\sim$ defined in Theorem~\ref{radtheorem}, so that $(x,k_x) \sim (z,k_z)$ by transitivity of that relation in $\Mc$. The only question is about the allowed orientations of $k_x$ and $k_y$. By the Hadamard condition on $O$, we have $k_x \triangleright 0$ and $k_y \triangleright 0$. On the other hand, the condition of being a point in ${\mathcal F}_\pm$ induces the condition that either both $k_y \triangleright 0$ and $k_z \triangleright 0$ or both $k_y \triangleleft 0$ and $k_z \triangleleft 0$. Combining the two conditions we find that $k_z \triangleright 0$, and hence that $(x,z,k_x,k_z) \in {\mathcal H}'_\Mc$. This concludes the proof.  $\Box$
\end{proof}

\begin{remark} With an elementary re-adaptation, statement (b)  holds true weakening the hypotheses, only requiring that $\omega_2 \in {\cal D}'(\Mc\times \Mc)$ and that it satisfies KG equation in both arguments up to smooth functions $r,l\in C^\infty(\Mc\times \Mc, \mathbb C)$, i.e.
$P_x \omega_2(x,y) =l(x,y)$, $P_y \omega_2(x,y) =r(x,y)$. In this form, it   closes a gap\footnote{The gap is the content of the three lines immediately before th proof {\bf (ii)} ${\bf 3} \Rightarrow {\bf 2}$ on p. 547 of \cite{ra1}: The reasoning presented there cannot exclude elements  of the form either $(x_1,x_2, 0,p_2)$ or 
$(x_1,x_2,p_1, 0)$ from $WF(\omega_2)$ outside ${\cal N}$. The idea of our proof was suggested by N. Pinamonti to the authors.}  present in the proof of the main result of \cite{ra1},  Theorem 5.1 (the fact that {\bf 1} implies {\bf 3}), and proves the statement on p. 548 of \cite{ra1} immediately after the proof of the mentioned theorem. It should be mentioned that the same gap had been previously explicitly identified and filled in the work of Sahlmann and Verch~\cite{sv}. These authors merged the partial proof of Radzikowski with the more restrictive result on the `propagation of Hadamard form' obtained previously in the works~\cite{fsw,fnw,kw} without the methods of microlocal analysis. Somewhat later, the same gap was also filled in the thesis of Sanders~\cite{sanders-phd}, who relied on purely microlocal but somewhat sophisticated methods developed earlier in~\cite{svw}. On the other hand, our method, though sharing some similarity in spirit with the ideas in~\cite{fsw,fnw,kw}, is both purely microlocal and rather elementary. In fact, it only takes advantage of the microlocal analysis in the guise of the theorem on the composition of wavefront sets.
\end{remark}

The microlocal formulation gave rise to noticeable results also closing some long standing problems. 
In particular it was proved that the so called Unruh state describing black hole radiation is Hadamard \cite{DMP11} and that the  analogous state, describing thermal radiation in equilibrium with a black hole, the so called Hartle-Hawking state is similarly Hadamard \cite{Sanders}. These results are physically important because they permit one  to compute the back reaction of the quantum radiation  on the geometry, since the averaged, renormalized  stress-energy tensor  $\omega(:\spa T_{ab}\spa:)$ can be defined in these states as previously discussed ((3) in Remark \ref{remhad3}). Other recent applications concerned the 
definition of relevant Hadamard states in asymptotically flat spacetimes at null infinity \cite{Mo09,GW2}, 
and spacelike infinity \cite{GW1}.
Natural Hadamard states for cosmological models have been discussed \cite{DMP09} also in relation with the problem of the Dark Energy \cite{Dark}.  An improved semiclassical formulation  where Einstein equations and the equation of evolution of the Hadamard quantum state and observables are solved simultaneously has been proposed in \cite{Pi11}. See \cite{Thomas,BDH} for  recent reviews also regarding fields with spin or helicity, in particular \cite{DappiaggiSiemens}  for the vector potential field.

\subsection{Algebra of Wick products}\label{sec:alg-wick}  Let us come to the proof of existence of Wick monomials $:\spa \phi^n\spa:(f)$ as algebraic objects, since we only have defined the expectation values $\omega_\Psi(:\spa \phi^n\spa:(f))$ in (\ref{phinHN}).
 We first introduce normal Wick products {\em defined with respect to a reference quasifree Hadamard 
state $\omega$} \cite{bfk,bf,hw}. 
 Referring to  the GNS triple for $\omega$, $({\cal H}_\omega, {\cal D}_\omega, \pi_\omega, \Psi_\omega)$
Define the elements, symmetric under interchange of $f_1,\ldots, f_n \in {\cal D}(\Mc)$,
$$\hat{W}_{\omega, 0}   \stackrel {\mbox{\scriptsize  def}} {=} \II\:, \quad \hat{W}_{\omega, n}(f_1,\ldots,f_n)   \stackrel {\mbox{\scriptsize  def}} {=} :\spa\hat{\phi}_\omega(f_1) \cdots  \hat{\phi}_\omega(f_n)\spa: _\omega  \quad \in \Ac(\Mb)$$
for $n=1,2,\ldots,$ where as before,
\begin{equation}:\spa \hat{\phi}_\omega(x_1) \cdots \hat{\phi}_\omega(x_n) \spa:_\omega  \:   \stackrel {\mbox{\scriptsize  def}} {=} \left. \frac{1}{i^n}\frac{\delta^n}{\delta f(x_1)\cdots \delta f(x_n)} \right|_{f=0}e^{i \hat{\phi}(f) +\frac{1}{2}\omega_{2}(f,f)}  \label{phinHhat}
\end{equation}
  The operators $\hat{W}_{\omega, n}(f_1,\ldots,f_n)$ can be extended to (or directly defined on) \cite{bf,hw}
an
  invariant subspace of ${\cal H}_\omega$, the  {\bf microlocal domain of 
  smoothness} \cite{bf,hw}, 
  $D_\omega \supset {\cal D}_\omega$, which is dense, invariant under the action of $\pi_\omega(\Ac(\Mb))$
and the associated unitary Weyl operators, and contains $\Psi_\omega$ and all of unit vectors  of ${\cal H}_\omega$ which induce Hadamard quasifree states on $\Ac(\Mb)$.
The map $$f_1\otimes \cdots \otimes f_n \mapsto \hat{W}_{\omega, n}(f_1,\ldots,f_n)$$ uniquely extends by complexification and linearity  to a map defined on $${\cal D}(\Mc)\otimes \cdots \otimes {\cal D}(\Mc)\:.$$
Finally, if $\Psi \in D_\omega$, the map ${\cal D}(\Mc)\otimes \cdots \otimes {\cal D}(\Mc) \ni h \mapsto  \hat{W}_{\omega, n}(h)\Psi$
  turns out to be continuous with respect to the relevant topologies: The one of ${\cal H}_\omega$ in the image and the one of ${\cal D}(\Mc^n)$ in the domain.
 A {\em vector-valued distribution}  ${\cal D}(\Mc^n) \ni h \mapsto  \hat{W}_{\omega, n}(h)$, uniquely arises  this way.
Actually, since $:\spa\hat{\phi}_\omega(f_1) \cdots  \hat{\phi}_\omega(f_{n})\spa: _\omega\:$ is symmetric 
 by construction, the above mentioned  distribution  is similarly symmetric and can be defined on the subspace ${\cal D}_n(\Mc) \subset  {\cal D}(\Mc^n)$ of the symmetric test functions:
 $${\cal D}(\Mc^n) \ni h \mapsto  \hat{W}_{\omega, n}(h)\:.$$
  By Lemma 2.2 in \cite{bf}, if $\Psi\in D_\omega$ the wavefront set  $WF\left(\hat{W}_{\omega, n}(\cdot) \Psi\right)$ of the vector-valued 
distributions 
  $t\mapsto \hat{W}_{\omega, n}(t) \Psi$, is contained in the set 
  \begin{eqnarray}
  {\bf F}_n(\Mb) \stackrel {\mbox{\scriptsize  def}} {=} \{(x_1,k_1,\ldots, x_n,k_n) \in 
(T^*\Mc)^n \setminus \{0\}
  | k_i\in V^-_{x_i}, i=1,\ldots,n\}  \label{Fn}\:,
  \end{eqnarray}
with $V_x^{+/-}$ denoting the set of all nonzero time-like and light-like
   co-vectors at $x$ which are future/past directed.  Theorem \ref{teoprod}, which can be proved to hold in this case too,  implies that we are allowed to define the product between a  distribution $t$
   and a vector-valued distribution $\hat{W}_{\omega, n}(\cdot) \Psi$ 
   provided $WF(t)+ {\bf F}_n(M,{\bf g})\not \ni \{(x,0)\:|\: x \in \Mc^n\}$. To this end, with ${\cal D}_n'(\Mc) \subset {\cal D}'(\Mc^n)$ denoting the subspace of symmetric distributions,   define
   $${\cal E}'_n(\Mc)  \stackrel {\mbox{\scriptsize  def}} {=}  \left\{ t \in {\cal D}_n'(\Mc)\:|\:
   \mbox{$\supp t$ is compact, $WF(t) \subset {\bf G}_n(\Mc)$ } \right\}$$
   where
   $${\bf G}_n(\Mc)  \stackrel {\mbox{\scriptsize  def}} {=} T^*\Mc^n\setminus
   \left( \bigcup_{x\in \Mc}(V_x^+)^n \cup \bigcup_{x\in \Mc}(V_x^-)^n\right)\:.$$
   It holds $WF(t)+ {\bf F}_n(\Mc)\not \ni \{(x,0)\:|\: x \in \Mc^n\}$ for $t\in {\cal E}'_n(\Mc)$. By 
consequence, the product 
   $$t \odot  \hat W_{\omega, n}(\cdot)\Psi$$
   of the distributions $t$ and $\hat{W}_{\omega, n}(\cdot) \Psi$ can be defined for
   every $\Psi \in D_\omega$ and  it turns out to be a well-defined
    vector-valued symmetric distribution,
     ${\cal D}_n(M) \ni f \mapsto t \odot  \hat W_{\omega, n}(f) \Psi$,
     with  values  in $D_\omega$. Thus, we have also defined an operator valued symmetric distribution, ${\mathcal D}_n(M) \ni f \mapsto t \odot \hat{W}_{\omega,n}(f)$, defined on and leaving invariant the domain $D_\omega$, acting as $\Psi \mapsto t\odot \hat{W}_{\omega,n}(f)\Psi$. This fact 
permits us to smear  $\hat{W}_{\omega, n}$ with $t\in {\cal E}'_n(\Mc)$, 
just  defining
  $$ \hat W_{\omega, n}(t)   \stackrel {\mbox{\scriptsize  def}} {=}  \left(t \odot \hat 
W_{\omega,n}\right)(f)\:,$$
where 
$f\in {\cal D}_n(\Mc)$ is equal to $1$ on  $\supp t$.
   It is simple to prove that the definition does not depend on $f$
    and the new smearing operation reduces to the usual one for  $t\in {\cal D}_n(\Mc) \subset {\cal 
E}'_n(M, {\bf g})$. Finally,
   since $f\delta_n \in {\cal E}'_n(\Mc)$ for $f\in {\cal D}(M)$, 
where $\delta_n$ is the Dirac delta supported on the diagonal of $\Mc^n = \Mc\times \cdots \times \Mc$ ($n$ times),
the following operator-valued 
distribution 
   is well-defined
   on  $D_\omega$ which, is then an invariant subspace,
   $$f\mapsto :\spa\hat{\phi}^n\spa:_\omega(f)
   \stackrel {\mbox{\scriptsize  def}} {=}\hat{W}_{\omega,n}(f\delta_n)\:,$$
 
\begin{definition}
  $ :\spa\hat\phi^n\spa:_\omega(f)$ is the 
   {\bf normal ordered product of $n$ field operators with respect to $\omega$}. 
  ${\cal W}_\omega(\Mb)$ is the $*$-algebra generated by
  $\II$ and the operators $\hat{W}_{\omega,n}(t)$
  for all $n\in \mathbb N$ and $t\in {\cal E}'_n(M,{\bf g})$ with involution given by 
   $\hat{W}_{\omega,n}(t)^* \stackrel {\mbox{\scriptsize  def}} {=}
  \hat{W}_{\omega,n}(t)^\dagger\spa\rest_{D_\omega} (= \hat{W}_{\omega,n}(\overline{t}))$.
\end{definition}

\begin{remark} $\null$

{\bf (1)} As proved in \cite{hw}, each product $\hat{W}_{\omega,n}(t)\hat{W}_{\omega,n'}(t')$ can be decomposed as a finite linear combination of terms $\hat{W}_{\omega,m}(s)$ extending the Wick theorem, and other natural identities, in particular related with commutation relations, hold.

{\bf (2)} $\pi_\omega({\cal A}(\Mb))$ 
  turns out to be a sub $*$-algebra of ${\cal W}_\omega(\Mb)$ since 
  $\hat{\phi}_\omega(f) = \:\: :\spa \hat{\phi}\spa:_\omega(f)$ for $f\in {\cal D}(M)$.
   \end{remark}
 If  $\omega,\omega'$ are two quasifree Hadamard states, 
   ${\cal W}_\omega(\Mb)$ and ${\cal W}_{\omega'}(\Mb)$ 
 are isomorphic  (not unitarily
 in general) under a canonical 
 $*$-isomorphism $$\alpha_{\omega' \omega} : {\cal W}_\omega(\Mb) \to  {\cal 
W}_{\omega'}(\Mb)\:,$$
 as shown in Lemma 2.1 in \cite{hw}. Explicitly, $\alpha_{\omega'\omega}$ is induced by linearity from the requirements
\begin{eqnarray}  \alpha_{\omega'\omega}(\II) = \II\:, \quad 
\alpha_{\omega'\omega}(W_{n, \omega}(t)) =  \sum_{k} W_{n-2k, \omega'}(\langle d^{\otimes k}, t\rangle) \:,
\end{eqnarray}
where $d(x_1,x_2)\stackrel {\mbox{\scriptsize  def}} {=} \omega(x_1,x_2)-\omega'(x_1,x_2)$
(only the symmetric part matters here) and
\begin{align}\langle d^{\otimes k}, t\rangle (x_1,\ldots, x_{n-2k}) &\stackrel {\mbox{\scriptsize  def}} {=} \frac{n!}{(2k)! (n-2k)!}\int_{\Mc^{2k}} t(y_1,\ldots, y_{2k}, x_1, \ldots, x_{n-2k})\nonumber \\
&\times \prod_{i=1}^k d(y_{2i-1}, y_{2i}) \dvol_{\Mb}(y_{2i-1})\dvol_{\Mb}(y_{2i}) \label{alpha}\end{align}
for $2k \leq n$ and $\langle d^{\otimes k}, t\rangle =0$ if $2k>n$.\\
  These $*$-isomorphisms also satisfy $$\alpha_{\omega'' \omega'}\circ \alpha_{\omega' \omega} =
  \alpha_{\omega'' \omega}$$ and  $$\alpha_{\omega' \omega} (\hat \phi_\omega(t)) = \:\: \hat 
\phi_{\omega'}(t)\:.$$
The idea behind these isomorphisms is evident: Replace everywhere $\omega$ by $\omega'$.
For instance
$$\alpha_{\omega'\omega}(:\spa \hat{\phi}^2\spa:_\omega(f)) = :\spa \hat{\phi}^2\spa:_{\omega'}(f) + \int_{\Mc} (\omega-\omega')(x,x) f(x) \dvol_{\Mb}\: \II$$
where $\omega-\omega'$ is smooth for (b) in Theorem \ref{radtheorem}.\\
One can eventually define  
an abstract unital $*$-algebra ${\cal W}(\Mb)$, generated by  elements $\II$ and $W_n(t)$ 
with $t \in {\cal E}'_n(\Mc)$,
isomorphic to each concrete unital $*$-algebra  ${\cal 
W}_\omega(\Mb)$ by
$*$-isomorphisms  $\alpha_\omega : {\cal W}(\Mb) \to {\cal W}_\omega(\Mb)$ such 
that,  if $\omega,\omega'$ are quasifree 
Hadamard  states, 
$\alpha_{\omega'}\circ \alpha_\omega^{-1} = \alpha_{\omega' \omega}$. \\
As above ${\cal A}(\Mb)$ is isomorphic to the  $*$-algebra of ${\cal W}(\Mb)$ 
generated by $\II$ and $W_1(f) = :\spa\hat{\phi}\spa:(f) = \phi(f)$ for $f \in {\cal D}(\Mc)$.

\begin{remark}
It is not evident how (Hadamard) states initially defined on $\Ac(\Mb)$ (continuously) extend to states on ${\cal W}(\Mb)$. This problem has been extensively discussed in \cite{HR} in terms of relevant topologies.
\end{remark}
It is now possible to define a notion of local Wick monomial {\em which does not depend on a preferred Hadamard state}. If $t \in {\cal E}'_n(\Mc)$ has support sufficiently concentrated around the diagonal of $\Mc^n$, realizing ${\cal W}(\Mb)$  as ${\cal 
W}_\omega(\Mb)$ for some quasifree Hadamard state $\omega$, we define 
a {\bf local covariant Wick polynomial}   as 
$$W_n(t)_H \stackrel {\mbox{\scriptsize  def}} {=} \alpha^{-1}_\omega \left(\alpha_{H\omega}(W_{n,\omega}(t))\right)$$
where $\alpha_{H\omega}$ is defined as in (\ref{alpha}) replacing $\omega'$ by the Hadamard parametrix $H_{0^+}$.
One easily proves that this definition does not depend on the choice of the Hadamard state $\omega$. The fact that the support of $t$ is supposed to be concentrated around of the diagonal of $\Mc^n$ it is due to the fact that $H_{\epsilon}(x,y)$ is defined only if $x$ is sufficiently close to $y$. This definition is completely consistent with (\ref{phinHN}), where now the 
$:\spa\phi(f_1)\cdots \phi(f_n)\spa:_H$ can be viewed as elements of ${\cal W}(\Mb)$ and not only of  $\Ac(\Mb)$,
 and
 it makes sense to write in particular,
$$:\spa\phi^2\spa:_H(f)\stackrel {\mbox{\scriptsize  def}} {=} 
W_2(f\delta_2)_H = \int_{\Mc^2} :\spa\phi(x)\phi(y)\spa:_H \delta(x,y)
f(x) \dvol_{\Mb^2}(x,y)\:.$$
Analogous monomials $:\spa\phi^n\spa:_H(f)$ are defined similarly as elements of ${\cal W}(\Mb)$. {\em With the said definition (\ref{phinHN}) holds true literally} and not only in the sense of quadratic forms.
\begin{remark} The presented definition  of  locally covariant Wick monomials $:\spa\phi^n\spa:_H(f)$, though satisfying general requirement of locality and covariance \cite{bfv} (see also \cite{chapt:CR}),
 remains however affected by several  ambiguities. A full classification of them is the first step of ultraviolet renormalization program \cite{hw,km}. The algebra   ${\cal W}(\Mb)$ also includes the so-called (locally covariant) time-ordered Wick polynomials,  necessary to completely perform the renormalization procedure \cite{hw2b}.
\end{remark}
The constructed formalism can be extended in order to encompass differentiated Wick polynomials  and it has a great deal of effect concerning the definition of the stress energy tensor operator \cite{stress}. It is defined as an element of ${\cal W}(\Mb)$ by subtracting the universal Hadamard singularity from the two-point function of $\omega$, before computing the relevant derivatives. 
\begin{eqnarray}\label{set}
 :\spa T_{ab}\spa:_H(f) = \int_{\Mc^2} 
D_{ab}(x,y) :\spa\phi(x) \phi(y) \spa:_H \delta(x,y) f(x) \: \dvol_{\Mb^2}(x,y)
\label{EMTensor}
\end{eqnarray}
\noindent $D_{ab}(x,y)$ is a certain symmetrized  second order partial differential operator obtained from (\ref{Tabgen})
  (cf. \cite{stress} Equation (10), and \cite{Thomas} where some minor misprints have been corrected and the signature $({-}{+}{+}{+})$ has been adopted), 
\begin{align*}
D_{ab}(x,y) := &\; D^\text{can}_{ab}(x,y) -\frac13 g_{ab} P_x\\
D^\text{can}_{ab}(x,y) :=&\; (1-2\xi)g^{b^\prime}_b\nabla_{a}\nabla_{b^\prime}-2\xi\nabla_{a}\nabla_b-\xi G_{ab}\\&\;+g_{ab}\left\{2\xi \square_x+\left(2\xi-\frac12\right)g^{c^\prime}_c\nabla^c\nabla_{c^\prime}+\frac12 m^2\right\}  \: .
\end{align*}
Here, covariant derivatives with primed indices indicate covariant derivatives w.r.t. $y$, $g^{b^\prime}_b$ denotes the parallel transport of vectors along the unique geodesic connecting $x$ and $y$, the metric $g_{ab}$ and the Einstein tensor $G_{ab}$ are considered to be evaluated at $x$.
The form of the ``canonical" piece $D^\text{can}_{ab}$ follows from the
definition of the classical stress-energy tensor, while the last term $-\frac13 g_{ab} P_x$,
giving rise to a final contribution $- \frac{g_{ab}}{3}:\spa \phi(x)P\phi(x) \spa:_H$  to the stress-energy operator,
 has been introduced in \cite{stress}. It gives no contribution classically, just in view of the very Klein-Gordon equation satisfied by the fields, however, in the quantum realm, its presence has a very important reason. Because the Hadamard parametrix satisfies the Klein-Gordon equation only up to smooth terms, the term with $P_x$ is non vanishing. Moreover, without this additional term, the above definition of $:\spa T_{ab}\spa:_H$  would not yield a conserved stress-tensor expectation value (see \cite{stress} Theorem 2.1).
On the other hand the added therm is responsible for the appearance of the famous {\em trace anomaly} \cite{wald94}. An extended discussion on conservation laws in this framework appears in \cite{HW04}.

\begin{acknowledgement}
The authors are grateful to R. Brunetti, C. Dappiaggi, C. Fewster, T. Hack,  N.Pinamonti, K. Sanders, A. Strohmaier, Y. Tanimoto and R. Verch for useful discussions and suggestions.
\end{acknowledgement}


\bibliographystyle{spmpsci}
\bibliography{Hadamard}
\end{document}